\documentclass[11pt,onecolumn]{IEEEtran}
\usepackage{amsthm}
\usepackage{times,amssymb,amsmath,amsfonts,float,nicefrac,color,bbm,mathrsfs,float}
\usepackage{algorithm,enumerate,multirow,caption,tikz,graphicx}
\usepackage{algpseudocode}
\usepackage[noadjust]{cite}
\usepackage{subcaption}
\usepackage{dsfont}

\usepackage[top=1in,bottom=1in,left=1in,right=1in]{geometry}
\allowdisplaybreaks

\newcommand{\RN}[1]{%
  \textup{\expandafter{\romannumeral#1}}%
}

\usetikzlibrary{shapes.misc, positioning}
\usetikzlibrary{decorations.pathreplacing}

\tikzset{
  block/.style    = {draw, thick, rectangle, minimum width = 2em},
sblock/.style      = {draw, thick, rectangle, minimum height = 2em,
minimum width = 2em}, 
}

\tikzset{XOR/.style={draw,circle,append after command={
        [shorten >=\pgflinewidth, shorten <=\pgflinewidth,]
        (\tikzlastnode.north) edge (\tikzlastnode.south)
        (\tikzlastnode.east) edge (\tikzlastnode.west)
        }
    }
}

\newcommand\remove[1]{}

\allowdisplaybreaks

\newtheorem{theorem}{Theorem}    
\newtheorem{definition}{Definition}

\newtheorem{observation}{Observation}
\newtheorem{claim}{Claim}
\newtheorem{lemma}{Lemma}
\newtheorem{corollary}{Corollary}
\newtheorem{conjecture}{Conjecture}

\newtheorem{remark}{Remark}

\newtheorem{cnstr}{Construction}

\newcommand{\original}[1]{#1}

\newcommand{\bB}{\mathbb{B}}

\newcommand{\bF}{\mathbb{F}}

\newcommand{\bP}{\mathbb{P}}

\newcommand{\cA}{\mathcal{A}}

\newcommand{\cC}{\mathcal{C}}

\newcommand{\cG}{\mathcal{G}}

\newcommand{\cW}{\mathcal{W}}

\newcommand{\cY}{\mathcal{Y}}

\newcommand{\bolda}{\bold{a}}
\newcommand{\boldb}{\bold{b}}
\newcommand{\boldc}{\bold{c}}

\newcommand{\bolde}{\bold{e}}
\newcommand{\boldy}{\bold{y}}

\DeclareMathOperator{\rank}{rank}

\DeclareMathOperator{\Eval}{Eval}

\DeclareMathOperator{\inter}{inter}

\DeclareMathOperator{\FHT}{FHT}

\DeclareMathOperator{\Majority}{Majority}
\DeclareMathOperator{\flip}{flip}
\DeclareMathOperator{\poly}{poly}

\DeclareMathOperator{\LLR}{LLR}

\DeclareMathOperator{\argmax}{argmax}

\usepackage{hyperref}
\hypersetup{
    colorlinks,
    citecolor=blue,
    filecolor=black,
    linkcolor=red,
    urlcolor=black
}

\usepackage{cleveref}

\newcommand{\R}{\mathbb{R}}

\newcommand{\N}{\mathbb{N}}

\newcommand{\E}{\mathbb{E}}
\newcommand{\F}{\mathbb{F}}

\newcommand{\per}[1]{\left(#1\right)}
\newcommand{\abs}[1]{\left|#1\right|}
\newcommand{\set}[1]{\left\{#1\right\}}
\newcommand{\prob}[2]{\mathbb{P}_{#1}\left[#2\right]}
\newcommand{\comment}[1]{}
\newcommand{\err}{\text{err}}

\newtheorem*{thm*}{Theorem}

\newtheorem*{lem*}{Lemma}

\newcommand{\dist}[2]{\text{dist}\per{{#1},{#2}}}

\newcommand{\weight}[1]{\mathrm{wt}(#1)}
\newcommand{\wt}[1]{\mathrm{wt}(#1)}

\newcommand{\derivative}[2]{\Delta_{#1}{#2}}

\newcommand{\weightdistribution}[3]{A_{{#2},{#1}}\left(#3\right)}
\newcommand{\wdist}[3]{A_{{#2},{#1}}\left(#3\right)}

\DeclareMathOperator{\bis}{bias}
\newcommand{\bias}[1]{\bis(#1)}

\newcommand{\reedmuller}[2]{\mathrm{RM}({#2},{#1})}

\newcommand{\polynomials}[2]{\reedmuller{r}{m}}

\newcommand{\rmrm}{\reedmuller{r}{m}}

\newcommand{\1}{\mathds{1}}

\begin{document}

\title{Reed-Muller Codes: Theory and Algorithms}

\author{Emmanuel Abbe \and \hspace*{.3in} Amir Shpilka \and \hspace*{.3in} Min Ye}

\maketitle
{\renewcommand{\thefootnote}{}\footnotetext{

\vspace{-.2in}
 
\noindent\rule{1.5in}{.4pt}

E. Abbe is with the Mathematics Institute and the School of Computer and Communication Sciences at EPFL, Switzerland, and the Program in Applied and Computational Mathematics and the Department of Electrical Engineering in Princeton University, USA.

A. Shpilka is with the Blavatnik school of Computer Science at Tel Aviv University, Tel Aviv, Israel. Email: shpilka@tauex.tau.ac.il.

M. Ye is with the Data Science and Information Technology Research Center, Tsinghua-Berkeley Shenzhen Institute, Tsinghua Shenzhen International Graduate School, Shenzhen, China. Email: yeemmi@gmail.com

The first and third authors were partly supported by the NSF CAREER Award CCF-1552131. The second author received funding from the Israel Science Foundation (grant number 552/16) and from the Len Blavatnik and the Blavatnik Family foundation.
}
}

\renewcommand{\thefootnote}{\arabic{footnote}}
\setcounter{footnote}{0}

\begin{abstract}
Reed-Muller (RM) codes are among the oldest, simplest and perhaps most ubiquitous family of codes. They are used in many areas of coding theory in both electrical engineering and computer science. Yet, many of their important properties are still under investigation. This paper covers some of the recent developments regarding the weight enumerator and the capacity-achieving properties of RM codes, as well as some of the algorithmic developments. In particular, the paper discusses the recent connections established between RM codes, thresholds of Boolean functions, polarization theory,  hypercontractivity, and the techniques of approximating low weight codewords using lower degree polynomials (when codewords are viewed as evaluation vectors of degree $r$ polynomials in $m$ variables). It then overviews some of the algorithms for decoding RM codes. It covers both algorithms with provable performance guarantees for every block length, as well as algorithms with state-of-the-art performances in practical regimes, which do not perform as well for large block length. Finally, the paper concludes with a few open problems.     
\end{abstract}

\clearpage

\tableofcontents

\clearpage

\section{Introduction}

A large variety of codes have been developed over the past 70 years. These were driven by various objectives, in particular, achieving efficiently the Shannon capacity \cite{Shannon48}, constructing perfect or good codes in the Hamming worst-case model \cite{Hamming50}, matching the performance of random codes, improving the decoding complexity, the weight enumerator, the scaling law, the universality, the local properties of the code \cite{Macwilliams77,Lint99,Blahut03,Helleseth05,Richardson08,Huffman10,Yekhanin12},  and more objectives in theoretical computer science such as in cryptography (e.g., secrete sharing, private information retrieval),  pseudorandomness, extractors, hardness amplification or probabilistic proof systems; see \cite{Abbe15} for references. Among this large variety of code developments, one of the first, simplest and perhaps most ubiquitous code is the Reed-Muller (RM) code.

The RM code was introduced by Muller in 1954 \cite{Muller54}, and Reed developed shortly after a decoding
algorithm decoding up to half its minimum distance \cite{Reed54}. The code construction can be described with a greedy procedure. Consider building
a linear code (with block length a power of two); it must contain the all-0 codeword. If one has to pick
a second codeword, then the all-1 codeword is the best choice under any meaningful criteria. If now one
has to keep these two codewords, the next best choice to maximize the code distance is the half-0 half-1
codeword, and to continue building a basis sequentially, one can add a few more vectors that preserve a
relative distance of half, completing the simplex code, which has an optimal rate for the relative distance
half.  Once saturation is reached at relative distance half, it is less clear how to pick the next codeword,
but one can simply re-iterate the simplex construction on any of the support of the previously picked
vectors, and iterate this after each saturation, reducing each time the distance by half. This gives the RM
code, whose basis is equivalently defined by the evaluation vectors of bounded degree monomials.

As mentioned, the first order RM code is the augmented simplex code or equivalently the Hadamard code, and the simplex code is the dual of the Hamming code that is `perfect'. This strong property is clearly lost once the RM code order gets higher, but RM codes preserve nonetheless a decent distance (at root block length for constant rate). Of course this does not give a `good' code (i.e., a code with constant rate and constant relative distance), and it is far from achieving the distance that other combinatorial codes can reach, such as Golay codes, BCH codes or expander codes \cite{Macwilliams77}. However, once put under the light of random errors, i.e., the Shannon setting, for which the minimum distance is no longer the right figure or merit, RM codes may perform  well again. In \cite{Helleseth04},  Levenshtein and co-authors show  that for the binary symmetric channel, there are codes that  improve on the simplex code in terms of the error probability (with matching length and dimension). Nonetheless, in the lens of Shannon capacity, RM codes seem to perform very well. In fact, more than well; it is plausible that they achieve the Shannon capacity on any Binary-input Memoryless Symmetric (BMS) channel \cite{Costello07,Abbe15,Kudekar16STOC,Kudekar17,AY18} and perform comparably to random codes on criteria such as the scaling law \cite{Hassani18} or the weight enumerator \cite{Sloane70,Macwilliams77,kasami1970weight,kasami1976weight,Kaufman12,Samorod18}. 

The fact that RM codes have good performance in the Shannon setting, and that they seem to achieve capacity, has long been observed and conjectured. It is hard to track back the first appearance of this belief in the literature, but \cite{Kudekar16STOC} reports that it was likely already present in the late 60s.
The claim was mentioned explicitly in a 1993 talk by
Shu Lin, entitled ?RM Codes are Not So Bad? \cite{Lin93}. It appears that a 1994 paper by Dumer and
Farrell contains the earliest printed discussion on this matter \cite{Dumer93}. Since
then, the topic has become increasingly prevalent\footnote{The capacity conjecture for the BEC at constant rate was posed as one of the open problems at the Information Theory Semester at the Simons Institute, Berkeley, in 2015.} \cite{Abbe15,Arikan09,Arikan2010survey,Carlet05,Costello07,Didier06,Mondelli14}.

But the research activity has truly sparked in the recent years, with the emergence of polar codes \cite{Arikan09}. Polar codes are the  close cousins of RM codes. They are derived from the same square matrix but with a different row selection. The more sophisticated and channel dependent construction of polar codes gives them the advantage of being provably capacity-achieving on any BMS channel, due to the polarization phenomenon. Even more impressive is the fact that they possess an efficient decoding algorithm down to the capacity. 

Shortly after the polar code breakthrough, and given the close relationship between polar and RM codes, the hope that RM codes could also be proved to achieve capacity on any BMS started to propagate, both in the electrical engineering and computer science communities. A first confirmation of this was obtained in extremal regimes of the BEC and BSC \cite{Abbe15}, exploiting new bounds on the weight enumerator \cite{Kaufman12}, and a first complete proof for the BEC at constant rate was finally obtained in \cite{Kudekar17}. These however did not exploit the close connection between RM and polar codes,  recently investigated in \cite{AY18}, showing that the RM transform is also polarizing and that a third variant of the RM code achieves capacity on any BMS with the conjecture that this variant is indeed the RM code itself. Nonetheless, the general conjecture that RM codes achieve capacity on any BMS channel remains open.

Polar codes and RM codes can be compared in different ways. In most performance metrics, and putting aside the decoding complexity, RM codes seem to be superior to polar codes \cite{Mondelli14,AY18}. Namely, they seem to achieve capacity universally and with an optimal scaling-law, while polar codes have a channel-dependent construction with a suboptimal scaling-law \cite{Hassani14,Hassani18}. However, RM codes seem more complex both in terms of  obtaining performance guarantees (as evidenced by the long standing  conjectures) and in terms of their decoding complexity. 

The efficient decoding of RM codes is the second main challenge regarding RM codes. Many algorithms have been propose since Reed's algorithm \cite{Reed54}, such as \cite{Green66,Be86,Sidel92,Sakkour05,Dumer04,Dumer06,Dumer06a}, and newer ones have appeared in the post polar code period \cite{Saptharishi17,YA18,Santi18}. Some of these already show that at various block-lengths and rates that are relevant for communication applications, RM codes are indeed competing or even superior to polar codes \cite{Mondelli14,YA18}, even compared to the improved versions considered for 5G \cite{3gpp}.

This survey is meant to overview these recent developments regarding  both the performance guarantees (in particular on weight enumerator and  capacity) and the decoding algorithms for RM codes.


\subsection{Outline}
The organization of this survey is as follows. We start in Section \ref{sect:defs} with the main definitions and basic properties of RM codes (their recursive structure, distance, duality, symmetry group and local properties). We then cover the bounds on their weight enumerator in Section \ref{sec:wt-dist}. In Section \ref{sect:capacity}, we cover results tackling the capacity achievability, using results on weight enumerator, thresholds of monontone Boolean functions and connections to polarization theory. We then cover various decoding algorithms in Section \ref{sect:drm}, providing pseudo-codes for them, and conclude in Section\ref{sect:open} with a selection of open problems.   

\section{Definitions and basic properties} \label{sect:defs}
\subsection{Definition and parameters}
Codewords of binary Reed-Muller codes consist of
the evaluation vectors of
multivariate polynomials over the binary field $\bF_2$. The encoding procedure of RM codes maps the information bits stored in the polynomial coefficients to the polynomial evaluation vector.
Consider the polynomial ring $\bF_2[x_1,x_2,\dots,x_m]$ with $m$ variables. For a polynomial $f\in \bF_2[x_1,x_2,\dots,x_m]$ and a binary vector $\original{z}=(z_1,z_2,\dots,z_m) \in \bF_2^m$, let $\Eval_{\original{z}}(f):=f(z_1,z_2,\dots,z_m)$ be the evaluation of $f$ at the vector $\original{z}$, and let $\Eval(f):=(\Eval_{\original{z}}(f): \original{z}\in\bF_2^m)$ be the evaluation vector of $f$ whose coordinates are the evaluations of $f$ at all $2^m$ vectors in $\bF_2^m$. 
Reed-Muller codes with parameters $m$ and $r$ consist of all the evaluation vectors of polynomials with $m$ variables and degree no larger than $r$.

\begin{definition}
 The $r$-th order (binary) Reed-Muller code $\mathrm{RM}(m,r)$ code is defined as the following set of binary vectors 
 $$
 \mathrm{RM}(m,r):=\{\Eval(f): f\in\bF_2[x_1,x_2,\dots,x_m], \deg (f)\le r  \}\;.
 $$
\end{definition}

Note that in later sections, we might use $\Eval(f)$ and $f$ interchangeably to denote the codeword of RM codes.
For a subset $A\subseteq[m]:=\{1,2,\dots,m\}$, we use the shorthand notation $x_A:=\prod_{i\in A}x_i$.
Notice that we always have $x^n=x$ in $\bF_2$ for any integer $n\ge 1$, so we only need to consider the polynomials in which the degree of each $x_i$ is no larger than $1$. 
All such polynomials with degree no larger than $r$ are linear combinations of the following set of monomials
$$
\{x_A:A\subseteq[m], |A|\le r \}.
$$
There are $\sum_{i=0}^r\binom{m}{i}$ such monomials, and the encoding procedure of $\mathrm{RM}(m,r)$ maps the coefficients of these monomials to their corresponding evaluation vectors. Therefore, $\mathrm{RM}(m,r)$ is a linear code with code length $n=2^m$ and code dimension  $\sum_{i=0}^r\binom{m}{i}$.
Moreover, the evaluation vectors $\{\Eval(x_A):A\subseteq[m], |A|\le r \}$ form a generator matrix of $\mathrm{RM}(m,r)$. Here we give a few examples of generator matrices for RM codes with code length $8$:
\begin{align*}
& \setlength\arraycolsep{2pt}
 \mathrm{RM}(3,0): 
\begin{bmatrix}
\Eval(1)
\end{bmatrix}
= \begin{bmatrix} 
1&1&1&1 & 1&1&1&1
\end{bmatrix}
\quad \quad \quad \quad \quad
\mathrm{RM}(3,1): 
\begin{bmatrix}
\Eval(x_1) \\
\Eval(x_2) \\
\Eval(x_3) \\
\Eval(1) 
\end{bmatrix}
= \begin{bmatrix} 
1&1&1&1 & 0&0&0&0 \\
1&1&0&0 & 1&1&0&0 \\
1&0&1&0 & 1&0&1&0 \\
1&1&1&1 & 1&1&1&1 
\end{bmatrix}      \\
& \setlength\arraycolsep{2pt}
 \mathrm{RM}(3,2): 
\begin{bmatrix}
\Eval(x_1x_2) \\
\Eval(x_1x_3) \\
\Eval(x_2x_3) \\
\Eval(x_1) \\
\Eval(x_2) \\
\Eval(x_3) \\
\Eval(1) 
\end{bmatrix}
= \begin{bmatrix} 
1&1&0&0 & 0&0&0&0 \\
1&0&1&0 & 0&0&0&0 \\
1&0&0&0 & 1&0&0&0 \\
1&1&1&1 & 0&0&0&0 \\
1&1&0&0 & 1&1&0&0 \\
1&0&1&0 & 1&0&1&0 \\
1&1&1&1 & 1&1&1&1 
\end{bmatrix}  
\quad \quad
\mathrm{RM}(3,3): 
\begin{bmatrix}
\Eval(x_1x_2x_3) \\
\Eval(x_1x_2) \\
\Eval(x_1x_3) \\
\Eval(x_2x_3) \\
\Eval(x_1) \\
\Eval(x_2) \\
\Eval(x_3) \\
\Eval(1) 
\end{bmatrix}
= \begin{bmatrix} 
1&0&0&0 & 0&0&0&0 \\
1&1&0&0 & 0&0&0&0 \\
1&0&1&0 & 0&0&0&0 \\
1&0&0&0 & 1&0&0&0 \\
1&1&1&1 & 0&0&0&0 \\
1&1&0&0 & 1&1&0&0 \\
1&0&1&0 & 1&0&1&0 \\
1&1&1&1 & 1&1&1&1 
\end{bmatrix}
\end{align*}
From this example, we can see that $\mathrm{RM}(m,0)$ is the repetition code, and $\mathrm{RM}(m,m)$ consists of all the binary vectors of length $n=2^m$, i.e., the evaluation vectors $\{\Eval(x_A):A\subseteq[m]\}$ form a basis of $\bF_2^n$.

Another equivalent way of defining RM codes is the Plotkin $(u, u+ v)$ construction, which we discuss in detail below in Section~\ref{sect:rscd}.
We also note that RM codes can be defined as geometry codes and refer to \cite{Lin01,Macwilliams77} for further details.

\subsection{Recursive structure and distance}  \label{sect:rscd}
For any polynomial $f\in\bF_2[x_1,x_2,\dots,x_m]$, we can always decompose it into two parts, one part containing $x_m$ and the other not containing $x_m$:
\begin{equation} \label{eq:dcp}
f(x_1,x_2,\dots,x_m)=g(x_1,x_2,\dots,x_{m-1})+x_m h(x_1,x_2,\dots,x_{m-1}).
\end{equation}
Here we use the fact that $x_m^n=x_m$ in $\bF_2$ for any integer $n\ge 1$.

We can also decompose the evaluation vector $\Eval(f)$ into two subvectors, one subvector consisting of the evaluations of $f$ at all $\original{z}=(z_1,\dots,z_m)$'s with $z_m=0$ and the other subvector consisting of the evaluations of $f$ at all $\original{z}=(z_1,\dots,z_m)$'s with $z_m=1$. We denote the first subvector as $\Eval^{[z_m=0]}(f)$ and the second one as $\Eval^{[z_m=1]}(f)$. We also define their sum over $\bF_2$ as 
$\Eval^{[/z_m]}(f):=\Eval^{[z_m=0]}(f)+\Eval^{[z_m=1]}(f)$.
Note that all three vectors $\Eval^{[z_m=0]}(f), \Eval^{[z_m=1]}(f)$ and $\Eval^{[/z_m]}(f)$ have length $2^{m-1}$, and their coordinates are indexed by $(z_1,z_2,\dots,z_{m-1})\in\bF_2^{m-1}$.

By \eqref{eq:dcp}, $\Eval^{[z_m=0]}(f)$ is the evaluation vector of $g(x_1,x_2,\dots,x_{m-1})$, and  $\Eval^{[/z_m]}(f)$ is the evaluation vector of $h(x_1,x_2,\dots,x_{m-1})$.
Now assume that $\Eval(f)\in\mathrm{RM}(m,r)$, or equivalently, assume that $\deg(f)\le r$. Then we have $\deg(g) \le r$ and $\deg(h) \le r-1$. Therefore, $\Eval^{[z_m=0]}(f) \in \mathrm{RM}(m-1,r)$ and 
$\Eval^{[/z_m]}(f)\in \mathrm{RM}(m-1,r-1)$.
This is called the Plotkin $(u, u+ v)$ construction of RM codes, meaning that if we take a codeword $c\in\mathrm{RM}(m,r)$, then we can always divide its coordinates into two subvectors $u$ and $u+v$ of length $2^{m-1}$, where $u\in\mathrm{RM}(m-1,r)$, $v\in\mathrm{RM}(m-1,r-1)$
and $c=(u,u+v)$.

A consequence of this recursive structure is that the code distance of $\mathrm{RM}(m,r)$ is $d=2^{m-r}$. We prove this by induction. It is easy to establish the induction basis. For the inductive step, suppose that the claim holds for $m-1$ and all $r\le m-1$; then we only need to show that the Hamming weight of the vector $(u,u+v)$ is at least $2^{m-r}$ for any $u\in\mathrm{RM}(m-1,r)$ and $v\in\mathrm{RM}(m-1,r-1)$, i.e., we only need to show that $w(u)+w(u+v)\ge 2^{m-r}$, where $w(\cdot)$ is the Hamming weight of a vector. Since $w(u+v)\ge w(v)-w(u)$, we have $w(u)+w(u+v)\ge w(v)$. By the inductive hypothesis, $w(v)\ge 2^{m-r}$, so $w(u)+w(u+v)\ge 2^{m-r}$. This completes the proof of the code distance.

The Plotkin construction also implies a recursive relation between the generator matrices of RM codes. More precisely, let $G(m-1,r-1)$ be a generator matrix of $\mathrm{RM}(m-1,r-1)$ and let $G(m-1,r)$ be a generator matrix of $\mathrm{RM}(m-1,r)$. Then we can obtain a generator matrix $G(m,r)$ of $\mathrm{RM}(m,r)$ using the following relation:
$$
G(m,r)=\begin{bmatrix}
G(m-1,r) & G(m-1,r) \\
0 & G(m-1,r-1)
\end{bmatrix},
$$
where $0$ denotes the all-zero matrix with the same size as $G(m-1,r-1)$.

\begin{table}
\centering
\large
\begin{tabular}{ |c|c|c|c|c| } 
 \hline
 code & code length & code dimension & code distance & dual code \\   \hline
 $\mathrm{RM}(m,r)$ & $n=2^m$ & $k = \sum_{i=0}^r \binom{m}{i}$ & $d = 2^{m-r}$ & $\mathrm{RM}(m,m-r-1)$ \\ 
 \hline
\end{tabular}
 \caption{Important parameters of $\mathrm{RM}(m,r)$.}
 \label{tb:imp}
\end{table}

\subsection{Duality}
The dual code of a binary linear code $\cC\subseteq \bF_2^n$ is defined as\footnote{As we work over $\F_2$ all calculations are done in that field.}
$$
\cC^{\perp}:=\{x\in\bF_2^n: \langle x,c \rangle=0 \quad \forall c\in\cC\},
\text{~~~where~} \langle x,c \rangle=\sum_{i=1}^n x_i c_i .
$$
By definition, the dual code $\cC^{\perp}$ is also a linear code,  and we have 
\begin{equation}\label{eq:sumc}
\dim(\cC)+\dim(\cC^{\perp})=n.
\end{equation}

Next we will show that the dual code of $\mathrm{RM}(m,r)$ is $\mathrm{RM}(m,m-r-1)$.
First, observe that the Hamming weight of every codeword in $\mathrm{RM}(m,m-1)$ is even, i.e., for every $f\in\bF_2[x_1,\dots,x_m]$ with $\deg(f)\le m-1$, we have $\sum_{\original{z}\in\bF_2^m}\Eval_{\original{z}}(f)=0$, where the summation is over $\bF_2$. This is because the Hamming weight of $\Eval(x_A)$ is $2^{m-|A|}$, so $\sum_{\original{z}}\Eval_{\original{z}\in\bF_2^m}(x_A)=0$ for all subsets $A$ with size $|A|\le m-1$. For every $f\in\bF_2[x_1,\dots,x_m]$ with $\deg(f)\le m-1$, we can write it as $f=\sum_{A\subset[m],|A|\le m-1}u_A x_A$. Therefore,
$$
\sum_{\original{z}\in\bF_2^m}\Eval_{\original{z}}(f)=
\sum_{\original{z}\in\bF_2^m}\Eval_{\original{z}}(\sum_{A\subset[m],|A|\le m-1}u_A x_A)
=\sum_{A\subset[m],|A|\le m-1}  \Big( u_A
\sum_{\original{z}\in\bF_2^m}\Eval_{\original{z}}(x_A) \Big) = 0 .
$$

Suppose that $\Eval(f)$ is a codeword of $\mathrm{RM}(m,r)$ and $\Eval(g)$ is a codeword of $\mathrm{RM}(m,m-r-1)$. Then $\deg(f)\le r$ and $\deg(g)\le m-r-1$.
Notice that $\langle \Eval(f), \Eval(g) \rangle = \sum_{\original{z}\in\bF_2^m} \Eval_{\original{z}}(f)\Eval_{\original{z}}(g)
=\sum_{\original{z}\in\bF_2^m} \Eval_{\original{z}}(fg)$.
Since $\deg(fg)\le m-1$, we have
$\langle \Eval(f), \Eval(g) \rangle 
=\sum_{\original{z}\in\bF_2^m} \Eval_{\original{z}}(fg)=0$.
Therefore, every codeword of $\mathrm{RM}(m,m-r-1)$ belongs to the dual code of $\mathrm{RM}(m,r)$, i.e., $\mathrm{RM}(m,m-r-1)\subseteq\mathrm{RM}(m,r)^{\perp}$.

Since
$$
\dim(\mathrm{RM}(m,m-r-1))=\sum_{i=0}^{m-r-1}\binom{m}{i}
=\sum_{i=0}^{m-r-1}\binom{m}{m-i}
=\sum_{i=r+1}^{m}\binom{m}{i},
$$
we have
$$
\dim(\mathrm{RM}(m,r))+\dim(\mathrm{RM}(m,m-r-1))
=\sum_{i=0}^{r}\binom{m}{i} +
\sum_{i=r+1}^{m}\binom{m}{i}
=\sum_{i=0}^{m}\binom{m}{i} = 2^m = n.
$$
Combining this with \eqref{eq:sumc}, we know that $\dim(\mathrm{RM}(m,m-r-1))=\dim(\mathrm{RM}(m,r)^{\perp})$. Thus we conclude that 
$$
\mathrm{RM}(m,m-r-1)=\mathrm{RM}(m,r)^{\perp} .
$$
This in particular tells us that the parity check matrix of $\mathrm{RM}(m,r)$ is the generator matrix of $\mathrm{RM}(m,m-r-1)$. The important parameters of $\mathrm{RM}(m,r)$ are summarized in Table~\ref{tb:imp}.

\subsection{Affine-invariance}
\label{sect:aip}
The automorphism group of a code $\cC$ is the set of permutations under which $\cC$ remains invariant. More precisely, the automorphism group of a code $\cC$ with code length $n$ is defined as $\cA(\cC):=\{\pi\in S_n: \pi(\cC)=\cC\}$, where $\pi(\cC):=\{\pi(c):c\in\cC\}$, and $\pi(c)$ is vector obtained from permuting the coordinates of $c$ according to $\pi$. It is easy to verify that $\cA(\cC)$ is always a subgroup of the symmetric group $S_n$. 

RM codes are affine-invariant in the sense that $\cA(\mathrm{RM}(m,r))$ contains a subgroup isomorphic to the affine linear group.
More specifically, since the codewords of RM codes are evaluation vectors and they are indexed by the vectors $\original{z}\in\bF_2^m$, the affine linear transform $g_{A,\original{b}}: \original{z}\mapsto A\original{z}+\original{b}$ gives a permutation on the coordinates of the codeword when $A$ is an $m\times m$ invertible matrix over $\bF_2$ and $\original{b}\in \bF_2^m$. Next we show that such a permutation indeed belongs to $\cA(\mathrm{RM}(m,r))$.
For any codeword $c \in \mathrm{RM}(m,r)$, there is a polynomial $f\in \bF_2[x_1,\dots,x_m]$ with $\deg(f)\le r$ such that $c=\Eval(f)$. 
Since $g_{A,\original{b}}(c)=\Eval(f \circ g_{A,\original{b}})$ and $\deg(f \circ g_{A,\original{b}})= \deg(f)\le r$, we have $g_{A,\original{b}}(c)\in \mathrm{RM}(m,r)$. Therefore, $g_{A,\original{b}} \in \cA(\mathrm{RM}(m,r))$, and RM codes are affine-invariant.

Recall that in Section~\ref{sect:rscd} we showed that $\Eval^{[/z_m]}(f)\in \mathrm{RM}(m-1,r-1)$ if $\Eval(f)\in \mathrm{RM}(m,r)$. Using the affine-invariant property, we can replace $z_m$ in this statement with any linear combination of $z_1,\dots,z_m$. More specifically, for any $\ell=b_1z_1+\dots+b_mz_m$ with nonzero coefficient vector $\original{b}=(b_1,\dots,b_m)\neq\original{0}$, we define $\Eval^{[\ell=0]}(f),  \Eval^{[\ell=1]}(f), \Eval^{[/\ell]}(f)$ in the same way as 
$\Eval^{[z_m=0]}(f),\Eval^{[z_m=1]}(f), \linebreak[4]
\Eval^{[/z_m]}(f)$.
For any such $\ell$,
one can always find an affine linear transform mapping $z_m$ to $\ell$. Since RM codes are invariant under such affine transforms, we have $\Eval^{[/\ell]}(f)\in \mathrm{RM}(m-1,r-1)$ whenever $\Eval(f)\in \mathrm{RM}(m,r)$.
This observation will be used in several decoding algorithms in Section~\ref{sect:drm}.

\subsection{General finite fields and locality}
The definition of binary RM codes above can be naturally extended to more general finite fields $\bF_q$.
Let us consider the polynomial ring $\mathbb{F}_q[x_1,x_2,\dots,x_m]$ of $m$ variables. 
For a polynomial $f\in\bF_q[x_1,x_2,\dots,x_m]$, we again use  $\Eval(f):=(\Eval_{\original{z}}(f): \original{z}\in\bF_q^m)$ to denote the evaluation vector of $f$. 
Since $x^q=x$ in $\bF_q$, we only need to consider the polynomials in which the degree of each $x_i$ is no larger than $q-1$, and the degree of such polynomials is no larger than $m(q-1)$.
\begin{definition}
Let $n:=q^m$ and $r\le m(q-1)$.
 The $r$-th order $q$-ary Reed-Muller code $\mathrm{RM}_q(m,r)$ code is defined as the following set of vectors in $\mathbb{F}_q^n$:
$$
\mathrm{RM}_q(m,r) := \{\Eval(f): f\in\bF_q[x_1,x_2,\dots,x_m], \deg (f)\le r  \}.
$$
\end{definition}

A locally decodable code (LDC) is an error-correcting code that allows a single bit of the original message to be decoded with high probability by only examining (or querying) a small number of bits of a possibly corrupted codeword.
RM codes over large finite fields are the oldest and most basic family of LDC.
When RM codes are used as LDC, the order $r$ of RM codes is typically set to be smaller than the field size $q$.
At a high level, local decoding of RM codes requires us
to efficiently correct the evaluation of a multivariate polynomial at a given point $\original{z}$ from the evaluation of the same polynomial at a small number of other points.
The decoding algorithm chooses a set  of points on an affine line that passes through $\original{z}$. It then queries the codeword for the evaluation of the polynomial on the points in this set and interpolates that polynomial to obtain the evaluation at $\original{z}$.
 We refer the readers to \cite{Yekhanin12} for more details on this topic.

\subsection{Notations}
We summarize here a few notations and parameters used in the paper.
We use $n=2^m$ to denote the code length (or blocklength) and $k=\sum_{i=0}^r\binom{m}{i}$ to denote the code dimension of RM codes.
We use $\mathbf{1}[\cdot]$ for the indicator function.
We also use the  notation $W:X\to Y$ to denote a communication channel that maps input random variable $X$ to output random variable $Y$, with the channel transition probability $W(y|x)=\mathbb{P}(Y=y|X=x)$.
We use $|\cdot|$ to denote the Hamming weight of a binary vector, i.e., the number of $1$'s in this vector, and we use $\wt{\cdot}$ to denote its relative Hamming weight, i.e., Hamming weight divided by the length of the  vector.
We use $h_z(\cdot)$ to denote the extrinsic information transfer (EXIT) function and 
$H_2(x):=-x\log_2(x)-(1-x)\log_2(1-x)$ to denote the binary entropy function.
Finally, in Section~\ref{sect:drm}, if two vectors belong to $\bF_2^m$, then by default their sum is over $\bF_2$ unless mentioned otherwise.

\section{Weight enumerator}\label{sec:wt-dist}

	In this section we survey known results on the weight distribution of binary Reed-Muller codes. We  first give the basic definition and discuss known results,  then we explain the ideas of the proofs of some of the main results. In Section \ref{sec:wt-to-cap}, we explain how these results can be used to prove that Reed-Muller codes achieve capacity for the BEC and the BSC for a certain range of parameters.\\

	The \emph{weight enumerator} of a code measures how many codewords of a given weight are there in the code. 
	\begin{definition}[Weight, Bias, Weight enumerator]\label{def:wt}\sloppy
		For a codeword $f$ we denote its hamming weight by $|f| = |\{z \in \bF_2^m \mid f(z)=1\}|$, and its relative weight, which we refer to simply as weight, with $\weight{f} = |f|/2^m = \E_Z[f(Z)]= \bP_Z[f(Z) =1]$ where $Z$ is drawn uniformly at random in $\bF_2^m$. When no distribution is mentioned, we draw the underlying random variable, in this case $Z$, according to the uniform distribution. We also denote the bias of a codeword by $\bias{f} = \E_Z[(-1)^{f(Z)}]$.    
		For  $\beta \in [0,1]$ we define $\weightdistribution{r}{m}{\beta} \triangleq \abs{\set{f \in \polynomials{r}{m} \mid \weight{f} = \beta}}$ and 
		$\weightdistribution{r}{m}{\leq\beta} \triangleq \abs{\set{f \in \polynomials{r}{m} \mid \weight{f}\leq \beta}}$.
	\end{definition}

	Thus, $\wdist{r}{m}{\leq \beta}$ counts the number of codewords of (relative) weight at most $\beta$. In particular, for $\beta < 2^{-r}$, $\weightdistribution{r}{m}{\leq\beta}=1$.
	The weight enumerator is one of the most useful measures for proving that a code achieves capacity as there are formulas that relate the distribution of weights  in a code to the probability of correcting random erasures or random errors (see Section~\ref{sec:wt-to-cap}). Intuitively, if the weight enumerator behaves similarly to that of a random code then we can expect the code to achieve capacity in a similar manner to random codes. Clearly, RM codes are quite different than random codes. 
	In particular, for a random code we expect that besides the zero codeword, every other codeword will have weight roughly $(1/2\pm \epsilon)$ (where $\epsilon$ is a constant depending on the rate), whereas RM codes contain many codewords of small weight. Nevertheless, as we shall see, if one can show that the weight enumerator drops quickly for $\beta<1/2$ then this may be sufficient for proving that the code achieves capacity. Thus, proving strong upper bound on the weight enumerator for weights slightly smaller than $1/2$ is an interesting and in some cases also a fruitful approach to proving that RM codes achieve capacity.
	
	\subsection{Results}\label{sec:weight:result}
	
	Computing the weight enumerator of RM codes is a well known problem that is open in most ranges of parameters. In 1970 Kasami and Tokura \cite{kasami1970weight} characterized all codewords of weight less than twice the minimum distance. This was  later improved in \cite{kasami1976weight} to all codewords of weight less than $2.5$ times the minimal distance. For degrees larger than $3$ they obtained the following result.
	\begin{theorem}[Theorem 1 of \cite{kasami1976weight}]\label{thm:wt<2.5}
		If $r \ge 3$ and $f\in \rmrm$ satisfies $\wt{f} < 2.5\cdot 2^{-r}$ then, up to an invertible linear transformation, 
		$$f = x_1 g(x_3,\ldots,x_m) + x_2 h(x_3,\ldots,x_m) + x_1 x_2 k(x_3,\ldots,x_m)\;,$$
		where $\deg(g)=\deg(h)=r-1$ and $\deg(k)=r-2$.
	\end{theorem}
By counting the number of such representations one can get a good estimate on the number of codewords of such weight.
	
	No significant progress was then made for over thirty years until the work of Kaufman, Lovett and Porat \cite{Kaufman12}  gave, for any constant degree $r=O(1)$,  asymptotically tight bounds on the weight enumerator of RM codes of degree $r$. 
	\begin{theorem}[Theorem 3.1 of  \cite{Kaufman12}]\label{thm:klp}
		Let $1\leq \ell\leq r-1$ and $0<\epsilon\leq 1/2$. It holds that
		$$\left(1/\epsilon\right)^{c_r \cdot m^{r-\ell}} \leq \wdist{r}{m}{\leq (1-\epsilon)2^{-\ell}} \leq \left(1/\epsilon\right)^{C_r \cdot m^{r-\ell}} \;,$$
		where $c_r,C_r>0$ are constants that depend only on $r$.
	\end{theorem}	
	Note that the relative weight of a codeword in $\rmrm$ is between $2^{-r}$ and $1$, hence we consider only $1\leq \ell\leq r-1$ in the statement of the theorem. 
	Unfortunately, as the degree gets larger, the estimate in Theorem~\ref{thm:klp} becomes less and less tight. Building on the techniques of \cite{Kaufman12},  Abbe, Shpilka and Wigderson \cite{Abbe15} managed to get better bounds for degrees up to $m/4$, which they used to show that RM codes achieve capacity for the BEC and the BSC for degrees $r=o(m)$. 
		\begin{theorem}[Theorem 3.3 of  \cite{Abbe15}]\label{thm:asw}
		Let $1\leq \ell\leq r-1<m/4$ and $0<\epsilon\leq 1/2$. Then,\footnote{We use the notation  ${n\choose \leq k} \triangleq \sum_{i=0}^{k}{n\choose i}$.}
		$$\wdist{r}{m}{\leq (1-\epsilon)2^{-\ell}} \leq \left(1/\epsilon\right)^{O\left( \ell^4 {m-\ell \choose \leq r-\ell}\right)} \;.$$
	\end{theorem}

	Sberlo and Shpilka \cite{Sberlo18} polished the techniques of \cite{Abbe15} and managed to obtain good estimates for every degree.
		\begin{theorem}[Theorem 1.2 of \cite{Sberlo18}]
		\label{thm:sberlo}
		Let $\gamma = r/m$. Then, for every integer $\ell$,
		$$ \wdist{r}{m}{\leq 2^{-\ell}}\leq 2^{\per{O(m^4) + 17(c_{\gamma}\ell+d_{\gamma})\gamma^{\ell-1}\binom{m}{\leq r}}} \;,$$
		where $c_{\gamma} = \frac{1}{1-\gamma}$ and $d_{\gamma} = \frac{2-\gamma}{(1-\gamma)^2}$.
	\end{theorem}
	To better understand the upper bound, note that the leading term in the exponent is $O\per{\ell \gamma^{\ell-1}\binom{m}{\leq r}}$ (when $0<\gamma<1$ is a constant)  whereas, if we were to state Theorem~\ref{thm:asw} in the same way, then its leading term would be $O\per{\ell^4 \gamma^{\ell-1}\binom{m}{\leq r}}$.
	
		Recently, Samorodnitsky \cite{Samorod18} proved a remarkable general result regarding the weight enumerator of codes that either they or their dual code achieve capacity for the BEC. His techniques are completely different than the techniques of \cite{Kaufman12,Abbe15,Sberlo18}.
	\begin{theorem}[Proposition 1.6 in \cite{Samorod18}]\label{thm:sam}
		Let $C\subseteq\{0,1\}^n$ a linear code of rate $R$. Let $(a_0,\ldots,a_n)$ be the distribution of hamming weights\footnote{I.e., there are exactly $a_i$ codewords whose hamming weight is exactly $i$ in $C$.} in $C$.
		For $0 \leq k \leq n$,
let $k^* = \text{min}\{k, n-k\}$. Let $\theta=R^{2\ln 2}$.
\begin{enumerate}
    \item If $C^\perp$ achieves capacity on the BEC then for all $0\leq k\leq n$
    $$ a_k \leq 2^{o(n)}\cdot \left(\frac{1}{1-R}\right)^{k^*\cdot 2\ln 2}\;.$$
    \item \label{alex:tight} If $C$ achieves capacity on the BEC then for all $0\leq k\leq n$
    $$a_k \leq 2^{o(n)}\cdot \left\{ 
    \begin{array}{ll}
    \frac{|C|}{(1-\theta)^{k^*}(1+\theta)^{n-k^*}} & 0\leq k^*\leq \frac{1-\theta}{2}\cdot n\\
    |C|\cdot \frac{{n\choose k^*}}{2^n} & \text{otherwise}
    \end{array}
    \right.$$
    \end{enumerate}
	\end{theorem}
	Observe that Item~\ref{alex:tight} in Theorem~\ref{thm:sam} says that for $k/n\in [\frac{1}{2}\pm R^{2\ln 2}]$, the weight distribution of a capacity achieving code is (up to the $2^{o(n)}$ term) the same as that of a random code.

	As Kudekar et al. \cite{Kudekar17} proved that all codes with parameters $\reedmuller{m/2\pm O(\sqrt{m})}{m}$ achieve  capacity for the BEC,  Theorem~\ref{thm:sam} gives strong upper bound on the weight enumerator of RM codes for linear weights, some are as tight as possible. 
	
	On the other hand, for hamming weight $k=o(n)$, the theorem fails to give meaningful bounds as the $2^{o(n)}$ term becomes too large to ignore and in fact, it dominates the entire estimate, and in particular it is a weaker bound than the one given in \Cref{thm:sberlo}. In addition, when the rate $R$ approaches zero (e.g. when $r<(1/2-\epsilon)m$) Theorem~\ref{thm:sam} does not give a meaningful estimate, again due to the $2^{o(n)}$ term, whereas  \Cref{thm:sberlo} works for such degrees as well. Thus, \Cref{thm:sam} gives very strong bounds for constant rate and constant relative weight, while \Cref{thm:sberlo} gives better bounds for small weight or small rate.

	So far we mostly discussed results on the weight distribution for weights that are a constant factor smaller than $1/2$. However, one expects that a random codeword will have weight roughly $1/2$. Thus, an interesting question to understand is the concentration of the weight, in particular it is interesting to know how many codewords have weight smaller than $(1-\epsilon)/2$ for small $\epsilon$, or, in terms of bias, how many codewords have bias at least $\epsilon$.  To the best of our knowledge, for other range of parameters, the first such result was obtained by Ben Eliezer, Hod and Lovett \cite{ben2012random} who proved the following.
	\begin{theorem}[Lemma 2 in \cite{ben2012random}]\label{thm:BHL}
		Let $m,r \in \N$ and $\delta > 0$ such that $r \leq (1-\delta)m$. Then there exist positive constants $c_1,c_2$ (which depend solely on $\delta$) such that,
		$$\weightdistribution{r}{m}{\leq \frac{1-2^{-c_1 \frac{m}{r}}}{2}} \leq 2^{\per{(1-c_2) \binom{m}{\leq r}}} \;,$$
		where the probability is over a uniformly random polynomial with $m$ variables and degree $\leq r$.
	\end{theorem}
This result was later extended to  other prime fields in  \cite{beame2018bias}. Note that for RM codes of constant rate, Theorem~\ref{thm:sam} gives much sharper estimates, but it does not extend to other ranges of parameters. 

Observe that when the degree $r$ is linear in $m$ Theorem~\ref{thm:BHL} only applies to weights that are some constant smaller than $1/2$ and does not give information about the number of polynomials that have bias $o(1)$. Such a result was obtained by Sberlo and Shpilka \cite{Sberlo18}, and it played an important role in their results on the capacity of RM codes. They first proved a result for the case that $r< m/2$ and then for the general case (with a weaker bound).
		\begin{theorem}[Theorem 1.4 of \cite{Sberlo18}]
		\label{thm:sberlo-bias}
		Let $\ell,m\in \N$ and let $0 < \gamma(m) < 1/2 - \Omega\per{\sqrt{\frac{\log m}{m}}}$ be a parameter (which may be constant or depend on $m$) such that $\frac{\ell+\log\frac{1}{1-2\gamma}}{(1-2\gamma)^2} = o(m)$. Then, 
		$$\weightdistribution{\gamma m}{ m}{\leq \frac{1-2^{-\ell}}{2}} \leq 2^{\per{O(m^4)+\per{1-2^{-c(\gamma,\ell)}}\binom{m}{\leq r}}} \;,$$
		where  $c(\gamma,\ell) = O\per{\frac{\gamma^2 \ell  +  \gamma \log(1/1-2\gamma)}{1-2\gamma} + \gamma}$.
	\end{theorem}
	
	As the form of the bound is a bit complicated, the following remark was made in \cite{Sberlo18}.
	\begin{remark}
		To make better sense of the parameters in the theorem we note the following.
		\begin{itemize}
			\item When $\gamma<1/2$ is a constant, $c(\gamma,\ell) =O(\ell)$.
			\item The bound is meaningful up to degrees $\left(\frac{1}{2} - \Omega\left(\frac{\sqrt{\log m}}{\sqrt m}\right) \right)m$, but falls short of working for constant rate RM codes.
			
			\item
			For $\gamma$ which is a constant the upper bound is applicable to $\ell = o(m)$ (in fact it is possible to push it all the way to some $\ell = \Omega(m)$). For $\gamma$ approaching $1/2$, i.e $\gamma = 1/2 - o(1)$, there is a trade-off between how small the $o(1)$ is and the largest $\ell$ for which the bound is applicable to. Nevertheless, even if $\gamma = 1/2 - \Omega\per{\sqrt{\frac{\log m}{m}}}$ the lemma still holds for $\ell = \Omega(\log m)$ (i.e, for a polynomially small bias).
		\end{itemize}
	\end{remark}
	

	We see that for linear degrees ($r=\Omega(m)$) Theorem~\ref{thm:BHL}  gives a  bound on the number of polynomials (or codewords) that have at least some constant bias, whereas \Cref{thm:sberlo-bias} holds for a wider range of parameters and in particular can handle bias which is nearly exponentially small in $m$. 
For general degrees, Sberlo and Shpilka obtained the following result.	
	\begin{theorem}[Theorem 1.7 of \cite{Sberlo18}]\label{thm:sberlo-bias-general}
		\label{thm : weak concentration of bias for all gamma}
		Let $r \leq m\in \N$ and $\epsilon > 0$. Then,
		$$\wdist{r}{m}{\leq \frac{1-\epsilon}{2}} \leq \exp\per{-\frac{2^{r}\epsilon^2}{2}} \cdot 2^{{\binom{m}{\leq r}}} \;.$$
	\end{theorem}
	
	Observe that the main difference between \Cref{thm:BHL} and \Cref{thm:sberlo-bias-general} is that in \Cref{thm:BHL} the bias ($\epsilon=2^{-c_1 m/r}$) is directly linked to the degree. In particular, for every $\epsilon$ \Cref{thm:BHL} allows for polynomials of degree $r = O(m/\log(1/\epsilon))$ whereas \Cref{thm:sberlo-bias-general} allows for $r$ and $\epsilon$ to be arbitrary. 
	When  $r= O(m/\log(1/\epsilon))$   the bound in  \Cref{thm:BHL} is stronger than the one given in \Cref{thm:sberlo-bias-general}. However, when $\epsilon=o(1)$ (as a function of $m$) then the result of \Cref{thm:BHL} is meaningful only for $r=o(m)$ whereas \Cref{thm:sberlo-bias-general} gives strong bounds for polynomials of every degree. 
	
To complete the picture we state two lower bounds on the weight enumerator. The first is a fairly straightforward observation.

\begin{observation}\label{obs:wd-lb}
$$ \frac{1}{2} \cdot 2^{{m-\ell+1 \choose \leq r - \ell +1}} \leq \wdist{r}{m}{\leq 2^{-\ell}}\;.$$
\end{observation}

Comparing to \Cref{thm:sberlo} and \Cref{thm:sam}, we see that for $\ell=O(1)$ the  lower bound in Observation~\ref{obs:wd-lb} is roughly of the form $\frac{1}{2}2^{(r/m)^{\ell-1}\cdot {m\choose \leq r}}$. Thus, for $r=m/2$ it is similar to what \Cref{thm:sam} gives, except that it has a smaller constant in the exponent: $2$ versus $4\ln 2$. For $r=\gamma m$ it gives a much smaller constant in the exponent compared to \Cref{thm:sberlo}: $2$ versus $17(c_{\gamma}\ell+d_{\gamma})$. The main difference is for very small weights, say $\ell = (1-\epsilon)r$. Then, the estimate in  Observation~\ref{obs:wd-lb} is much smaller than what \Cref{thm:sberlo} gives. This is mainly due to the fact that \Cref{thm:sberlo} heavily relies on estimates of binomial coefficients that become less and less good as $(r-\ell) \rightarrow 0$.

The second lower bounds was given in \cite{Sberlo18} and it concerns weights around  $1/2$. 
	\begin{theorem}[Theorem 1.8 of \cite{Sberlo18}]
		\label{thm:biaslb}
		Let $20\leq r \leq m\in \N$. Then for any integer $\ell <r/3$ and sufficiently large $m$ it holds that
		$$\frac{1}{2}\cdot  2^{\per{\sum_{j=1}^{\ell-1}\binom{m-j}{\leq r-1}}}\leq \wdist{r}{m}{\leq \frac{1-2^{-\ell}}{2}}   \;.$$
	\end{theorem}

	Comparing the upper bound in \Cref{thm:sberlo-bias} to \Cref{thm:biaslb} we see that there is a gap between the two bounds. Roughly, the lower bound on the number of polynomials that have bias at least $\epsilon$ matches the upper bound corresponding to bias at least $\sqrt{\epsilon}$. This may be a bit difficult to see when looking at \Cref{thm:sberlo} but we refer to Remark 3.16 in \cite{Sberlo18} for a qualitative comparison.

\subsection{Proof strategy}

In this section we explain the basic ideas behind the proofs of the theorems stated in Section~\ref{sec:weight:result}. The most common proof strategy is based on the approach of Kaufman et al. \cite{Kaufman12}, which, following \cite{Sberlo18}, we call the $\epsilon$-net approach. Theorem~\ref{thm:sam} is proved using a completely different approach. \\

\subsubsection{The $\epsilon$-net approach for upper bounding the weight enumerator (Theorems ~\ref{thm:klp},~\ref{thm:asw},~\ref{thm:sberlo},~\ref{thm:sberlo-bias},~\ref{thm:sberlo-bias-general})} \label{sec:eps-net}

	The main idea in the  work of \cite{Kaufman12}, which was later refined in \cite{Abbe15} and \cite{Sberlo18}, is that in order to upper bound the number of polynomials of certain weight we should find a relatively small set, in the space of all functions $\F_2^m \rightarrow \F_2$, such that all low weight polynomials are contained in balls of (relative) radius at most $\epsilon$ around the elements of the set, with respect to the hamming distance. We call such a set an $\epsilon$-net for $\reedmuller{r}{m}$. 
	
	Why is this approach useful? Assuming we have found such a net, we can upper bound the number of low weight codewords by the number of codewords in each ball times the size of the net. The crux of the argument is to note that the number of codewords in each ball is upper bounded by $\wdist{r}{m}{2\epsilon}$. This gives rise to a recursive approach whose base case is when the radius of the ball is smaller than half the minimum distance (and then there is at most one codeword in the ball). Formally, this idea is captured by the next simple claim.

	\begin{lemma}
		\label{lem:basic counting argument}
		Let $S \subseteq \polynomials{m}{r}$ be a subset of polynomials with an $\epsilon$-net $\mathcal{N}$. Then,
		$$\abs{S} \leq \abs{\mathcal{N}}\cdot \weightdistribution{r}{m}{2\epsilon} \;.$$
	\end{lemma}

		Thus, to get strong upper bounds on the number of low weight/bias polynomials we would like the $\epsilon$-net to be as effective as possible. This means that on the one hand we would like the $\epsilon$-net to be small and on the other hand that no ball around an element of the net should contain too many codewords. 
	
	Before explaining how to get such an $\epsilon$-net we first explain how this approach was developed in the papers \cite{Kaufman12},\cite{Abbe15},\cite{Sberlo18}.	To prove Theorem~\ref{thm:klp} Kaufman et al. \cite{Kaufman12} constructed an $\epsilon$-net such that each ball contains at most one low weight polynomial. To achieve this they picked $\epsilon=2^{-r-1}$. However, since they insisted on having at most one low weight polynomial in every ball, this resulted in a relatively large net. To prove Theorem~\ref{thm:asw}, Abbe et al. observed that one can bound the number of codewords in each ball using the weight enumerator at smaller weights. This allowed them to pick a larger $\epsilon$ and use recursion. \cite{Abbe15} used the same approach as \cite{Kaufman12} to construct the $\epsilon$-net,  though with different parameters and with tighter analysis on the net size. \cite{Sberlo18} improved further on  \cite{Abbe15} by observing that the $\epsilon$-net approach was used only to bound the weight enumerator at weights which are somewhat smaller than $1/2$. The reason for that is that the calculations performed in \cite{Kaufman12,Abbe15} were not tight enough and stopped working as the weights got closer to $1/2$ or as the degree got larger. 
	Thus, \cite{Sberlo18} first improved the calculations upper bounding the size of the $\epsilon$-net and then, using the improved calculations, obtained results for the weight enumerator also for weights close to $1/2$, and for all degrees, as stated in Theorems~\ref{thm:sberlo},~\ref{thm:sberlo-bias} and~\ref{thm:sberlo-bias-general}.

We now explain the main idea of Kaufman et al. for constructing the $\epsilon$-net. For this we will need the notion of \emph{discrete derivative}. The discrete derivative of a function $f:\F_2^m\to \F_2$ at direction $y\in \F_2^m$ is the function 
\begin{equation}  \label{eq:disder}
\Delta_y f: x \mapsto \Delta_y f (x) \triangleq f(x+y)+f(x)\;.
\end{equation}
It is not hard to see that if $f$ is a degree $r$ polynomial then, for every $y$, $\Delta_y f$ has degree at most $r-1$. 
Another basic observation is that if a function $f:\F_2^m\to \F_2$ has weight $\beta$, then, for each $x$,
\[
\bP_Y[\Delta_Y f(x)=f(x)]=1-\beta\;,
\]
where by $Y$ we mean a random variable that is distributed uniformly over $\F_2^m$.
Thus, if $\beta<1/2$ then $\Delta_y f(x)$ gives a good estimate for $f(x)$. Hence, if we consider $t$ directions $y_1,\ldots,y_t \in \F_2^m$ and define 
\[\label{eq:F0t0maj}
F_{y_1,\ldots,y_t}(x) \triangleq \text{Majority}\left(\Delta_{y_1} f(x),\ldots,\Delta_{y_t} f(x)\right)
\]
then we get from the Chernoff-Hoeffding bound that
\[ 
\E_{X,Y_1,\ldots,Y_t} [\mathbf{1}[F_{Y_1,\ldots,Y_t}(X)\neq f(X)]] \approx \exp(-t) \;.
\]
Picking $t=O(\log 1/\epsilon)$ we see that for each polynomial $f\in \reedmuller{r}{m}$ there is some $F_{y_1,\ldots,y_t}$ at hamming distance at most $\epsilon$. Thus, the set of all such $F_{y_1,\ldots,y_t}$ forms an $\epsilon$-net for polynomials of weight at most $\beta$ in $\reedmuller{r}{m}$. All that is left to do is to count the number of such functions $F_{y_1,\ldots,y_t}$ to obtain a bound on the size of our net. 

In fact, one can carry the same approach further and rather than approximating $f$ by its first order derivatives, use instead higher order derivatives. For a set of $k$ directions $\cY=\{y_1,\ldots,y_k\} \in \left(\F_2^m\right)^k$ we define 
\[ \Delta_{\cY}f(x) = \Delta_{y_k}\Delta_{y_{k-1}}\cdots\Delta_{y_1}f(x)\;.
\]
The following version of Lemma 2.2 of \cite{Kaufman12} appeared in \cite{Sberlo18}.

	\begin{lemma}
		\label{lem:KLP}
		Let $f : \F_2^m \rightarrow \F_2$ be a function such that $\weight{f}\leq 2^{-k}$ for $k \geq 2$ and let $\epsilon > 0$. Then, there exist directions $\cY_1,\ldots, \cY_t \in (\F_2^m)^{k-1}$ such that
		\[\prob{X}{f(X) \neq \text{Majority}\left(\Delta_{\cY_1}{f}(X),\ldots,\Delta_{\cY_t}{f}(X)\right)}\leq \epsilon \;,\]
		where  $t = \lceil 17\log(1/\epsilon) \rceil$.
	\end{lemma}

	Denote\footnote{Actually, we have to take a weighted majority, but for sake of clarity we ignore this detail in our presentation.} 
	\[F_{\cY_1,\ldots,\cY_t} \triangleq \text{Majority}\left(\derivative{\cY_1}{f},\ldots,\derivative{\cY_t}{f}\right)\]
	and
		\[ \mathcal{N}_{k,t} \triangleq  \set{F_{\cY_1,\ldots,\cY_t}
		 : \cY_1,\ldots, \cY_t \in (\F_2^m)^{k} \;,\; f \in \reedmuller{r}{m} \text{ and } \weight{f}\leq 2^{-k-1} }\;. \]
		
		Combining Lemma~\ref{lem:basic counting argument} with recursive applications of Lemma~\ref{lem:KLP}, we obtain the following bound on the weight enumerator.
		
			\begin{corollary}
		\label{cor:recursion for weight}
		Let $r,m,\ell \in \N$ such that $r \leq m$. Then,
		\[\weightdistribution{r}{m}{\leq 2^{-\ell}} \leq \abs{\mathcal{N}_{\ell-1,t}}\cdot\weightdistribution{r}{m}{2^{-\ell-1}}\;,\]
		where $t = 17(\ell+2)$.
	Consequently, 
		\[\weightdistribution{r}{m}{\leq 2^{-\ell}} \leq \prod_{j=\ell}^{r} \abs{\mathcal{N}_{j-1,17(j+2)}} \;.\]
	\end{corollary}

The way that Kaufman et al. bounded the size of $\mathcal{N}_{k,t}$ was simply to say that each $F_{\cY_1,\ldots,\cY_t}$ is an explicit function of $t$ polynomials of degree $r-k$ and hence the size of the net is at most $\abs{\reedmuller{r-k}{m}}^t$. One idea in the improvement of \cite{Abbe15} over \cite{Kaufman12} is that derivatives of polynomials can be represented as polynomials in fewer variables. Specifically, one can think of $\Delta_{\cY} f$ as a polynomial defined on the vector space $\text{span}(\cY)^\perp$. This allows for some saving in the counting argument, namely, 
\[ \abs{\mathcal{N}_{k,t}} \leq 2^{mkt}\cdot \abs{\reedmuller{r-k}{m-k}}^t \;,\]
where the term $2^{mkt}$ comes from the fact that now we need to explicitly specify the sets $\cY_1,\ldots,\cY_t$.
\cite{Sberlo18} further improved the upper bound by noting that different derivatives contain information about each other. I.e., they share monomials. This allowed them to get a better control of the amount of information encoded in the list of derivatives and as a result to obtain a better bound on the size of the net. This proved significant for bounding the number of codewords having small bias.

The discussion above relied on Lemma~\ref{lem:KLP} that works for weights at most $1/4$. For weights closer to $1/2$ a similar approach is taken except that this time Kaufman et al. noted that one can pick a subspace of dimension $t$ and consider all $2^t-1$ non-trivial first order derivatives, according to directions in the subspace, to obtain a good approximation for the polynomial. Since the directions according to which we take derivatives are no longer independent they could no longer use the Chernoff-Hoeffding bound in their argument. Instead they observed that the directions are $2$-wise independent and could therefore use the Chebyshev bound instead to bound $t$. Specifically, Lemma 2.4 of \cite{Kaufman12} as stated in \cite{Sberlo18} gives:

	\begin{lemma}
		[Lemma 2.4 in \cite{Kaufman12}]
		\label{KPL2}
		Let $f : \F_2^n \rightarrow \F_2$ be a function such that $\bias{f} \geq \delta > 0$ and let $\epsilon>0$. Then, for  $t = \lceil \log(1/\epsilon) + 2\log(1/\delta) + 1 \rceil$, there exist directions $y_1,\ldots,y_t \in \F_2^m$ such that,
		$$\prob{X}{f(X) = \text{Majority}\left(\derivative{\sum_{i \in I}y_i}{f}(X) : \emptyset \neq I \subseteq [t]\right)} \geq 1 - \epsilon \;.$$
	\end{lemma}
	\begin{corollary}\label{cor:def-bt}
		For any $t \in \N$ define,
		$$\mathcal{B}_{t} = \set{\text{Majority}\left(\derivative{\sum_{i \in I}y_i}{f}(x) : \emptyset \neq I \subseteq [t]\right) : f \in \reedmuller{r}{m} \;,\; y_1,\ldots,y_t \in \F_2^m } \;.$$
		Then, for  $t = \lceil \log(1/\epsilon) + 2\log(1/\delta) + 1 \rceil$, $\mathcal{B}_{t}$ is an $\epsilon$-net for $\set{f \in \reedmuller{r}{m} : \bias{f} \geq \delta}$.
	\end{corollary}

To conclude, the $\epsilon$-net approach works as follows. We first show that each polynomial of weight at most $\beta$ can be approximated by an explicit function of some lower order derivatives. We then count the number of such possible representations and then continue recursively to bound the number of codewords that are close to each such function. \\

\subsubsection{Connection to hypercontractivity}

Samordnitsky proved Theorem~\ref{thm:sam} as a corollary of a more general result concerning the behavior of Boolean functions under noise. We shall only give a high level description of his approach. 

For a Boolean function $\phi: \{0,1\}^n\to \R$ and a noise parameter $\eta$ let us denote 
\[\phi_\eta(x) \triangleq \sum_{y\in\{0,1\}^n} \eta^{\dist{x}{y}}(1-\eta)^{n-\dist{x}{y}} \phi(y) \;,\]
where $\dist{x}{y}$ is the hamming distance between $x$ and $y$. Thus, $\phi_\eta(x)$ averages the value of $\phi$ around $x$ according to the $\eta$-biased measure. As $\phi_\eta$ is a convex combination of many shifted copies of $\phi$, it's $\ell_q$ norm cannot be larger than $\phi$'s. Samordnitsky's main theorem quantifies that loss in norm. 
To state his main result we shall need the following notation. For $0\leq \lambda\leq 1$, let $T \sim \lambda$ denote a random subset $T$ of $\{1,\ldots,n\}$ in which each element is chosen independently with probability $\lambda$. Let $\E(\phi|T)$ be the conditional expectation of $\phi$ with respect to $T$, that is, $\E(\phi |T)(x) \triangleq \E_{Y:Y|_T =x|_T} \phi(Y)$.

\begin{theorem}[Theorem 1.1 of \cite{Samorod18}]\label{thm:sam-norm}
Let $\phi:\{0,1\}^n\to\R$ be non negative. Then, for any $q>1$ 
\[\log_2 \| \phi_\eta \|_q \leq \E_{T\sim\lambda}\log_2 \| \E(\phi | T)\|_q \;,\]
with $\lambda= (1-2\eta)^{a(q)}$, where
\[ a(q) = \left\{ \begin{array}{ll}
\frac{1}{2 \ln 2} \frac{2^{3-q}\left( 2^{q-1} -1 \right)}{q-1} & 1 < q \leq 2 \\
\frac{1}{2 \ln 2} \frac{q}{q-1} & q\leq 2
\end{array} \right. \;.
\]
\end{theorem}
Results quantifying the decrease in norm due to noise are called hypercontractive inequalities, and Samorodnitsky's proof  follows the proof of a hypercontractive inequality due to Gross \cite{gross1975logarithmic}.
As Samorodnitsky writes, the idea of the proof is to ``view both sides of the corresponding inequality as functions of $\eta$ (for a fixed $q$) and compare the derivatives. Since noise operators form a semigroup it suffices to compare the derivatives at zero, and this is done via an appropriate logarithmic Sobolev inequality.''
While Theorem~\ref{thm:sam-norm} does not seem related to RM codes it turns out that if one takes $\phi$ to be the characteristic function of the code then this gives a lot of information about the weight distribution of the code. 
For a linear code $C\subseteq \{0,1\}^n$ with generating matrix $G$, let $\phi_C = \frac{2^n}{|C|}\mathds{1}_C$. For $T\subseteq \{1,\ldots,n\}$ let $\rank_C(T)$ denote the rank of the submatrix of $G$ generated by the rows whose indices are in $T$. The following is a consequence of Theorem~\ref{thm:sam-norm}.

\begin{theorem}[Proposition 1.3,1.4 of \cite{Samorod18} specialized to $q=2$]\label{thm:sam-rank}
Let $(a_0,\ldots,a_n)$ be the distribution of hamming weights in $C$ and $(b_0,\ldots,b_n)$ the distribution of hamming weights in $C^\perp$.
For  $0\leq \lambda \leq 1$ let $\eta=\frac{1 - \lambda^{\ln 2}}{2}$. Then, for any such $\lambda$, and $\theta=\lambda^{2\ln 2}$, it holds that
\begin{align*}
    \lambda n - \E_{T\sim \lambda} \rank_{C^\perp}(T)  \geq \log_2\left( \E [{\phi_C}_\eta^2]\right) &= \log_2 \left( \sum_{i=0}^{n} b_i \lambda^{2i\cdot \ln 2}\right) \\ &= \log_2\left(\frac{1}{|C|}\cdot \sum_{k=0}^{n}a_k(1-\theta)^k(1+\theta)^{n-k} \right)\;.
\end{align*} 
\end{theorem}

Samordnitsky then observed that if $C^\perp$ achieves capacity for the BEC then for $\lambda=\text{Rate}(C^\perp)+o(1)$ it holds that
$\lambda n - \E_{T\sim \lambda} \rank_{C^\perp}(T) =o(n)$. Similar result follows if $C$ achieves capacity, for the appropriate $\lambda$. Combining this with Theorem~\ref{thm:sam-rank} he obtained his main result on the weight enumerator of codes that either they or their duals  achieve capacity.\\

\subsubsection{Lower bounds on the weight enumerator}

The proofs of both Observation~\ref{obs:wd-lb} and Theorem~\ref{thm:biaslb} are based on exhibiting a large set of polynomials having the claimed weight. 

Observation~\ref{obs:wd-lb} follows from the simple fact that for a random polynomial $g(x_{\ell},\ldots,x_n)$, of degree $r-\ell+1$, the degree $r$ polynomial $f(x_1,\ldots,x_n)=x_1\cdot x_2 \cdots x_{\ell-1}\cdot g$ will have weight at most $2^{-\ell}$ with  probability $1/2$ (as half of the polynomials $g$ have weight at most $1/2$).

To prove Theorem~\ref{thm:biaslb} we consider all polynomials of the form $$g(x_1,\ldots,x_m) = \sum_{i=1}^{\ell} x_i f_i(x_{i+1},\ldots,x_{m-i})\;,$$ where $f_i\in\reedmuller{r-1}{m-i}$. It is not hard to prove that with high probability, when we choose the $f_i$'s at random we get that $\bias{g} \geq 2^{-\ell+1}$. A simple counting argument then gives the claimed lower bound.

\section{Capacity results}  \label{sect:capacity}

\subsection{RM codes achieve capacity at low rate \cite{Abbe15,Sberlo18}}\label{sec:wt-to-cap}

	The results on the weight enumerator that were described in Section~\ref{sec:wt-dist} can be used to show that RM codes having low rate achieve capacity for the BEC and the BSC. Note that in the case of rates close to $0$ achieving capacity means that, for the BEC we can correct a fraction $p\geq 1-R(1+\epsilon)$ of random erasures (with high probability), and for the BSC we can correct a fraction $p$ of random errors (with high probability) for $p$ satisfying $h(p)\geq 1-R(1+\epsilon)$. See \cite{Abbe15} for a discussion of achieving capacity in extremal regimes.
	In the following subsections we shall show how to relate the error of the natural decoder for each channel to the weight distribution of the code. Such results were first prove in \cite{Abbe15} and later strengthened in \cite{Sberlo18} using essentially the same approach so we only give the state of the art results of \cite{Sberlo18}.\\
	
	\subsubsection{BEC}\label{sec:erasures}
	
	In this section we relate the weight distribution of a RM code to the probability of correctly decoding from random erasures. In \cite{Abbe15,Sberlo18} this was used to conclude that RM codes achieve capacity for the BEC at certain parameter range.

	Denote by $\err(p,m,r)$ the probability that $\reedmuller{r}{m}$  cannot recover from random erasures with parameter $p$ (i.e. when the erasure probability is $p$). To upper bound $\err(p,m,r)$ we need to understand when can we correct a codeword $f$ from some erasure pattern, where the erasure pattern is the set of coordinates erased in $f$. Thus, given a codeword $f$ and an erasure pattern $S \subseteq \F_2^m$, the corresponding corrupted codeword is the evaluation vector of $f$ with the evaluations over the set $S$ erased. The basic idea is based on the simple fact that an erasure pattern can be recovered from if and only if no codeword is supported on the pattern. By that we mean that there is no codeword $f$ that obtains nonzero values only for points in $S$, and therefore there is no other codeword that agrees with the transmitted one on the unerased section. 
	Thus, to bound the probability of failing to correct we need to bound the probability that a codeword is supported on a random erasure pattern (where each coordinate is selected with probability $p$). 

	\begin{lemma}
		\label{lem:erase-recover}
		For $ r\leq m \in \N$. For $R=\text{Rate}\left(\reedmuller{r}{rm}\right)$ it holds that,
		\begin{equation}\label{eq:erasure-prob}
		    \err(p,m,r) \leq  \sum_{\beta} (1-\beta)^{(1-p)2^m} \cdot \weightdistribution{r}{m}{\beta} \;,
		\end{equation}
		where $\beta$ ranges over all possible weights.
	\end{lemma}
	
	\begin{proof}[Proof sketch]
		To simplify matters assume that the erasure pattern consists of exactly $p\cdot 2^m=pn$ erasures. We now wish to bound that probability that  a random set $S\subset \F_2^m$ of size $|S|=p\cdot 2^m$ contains the support of a codeword. To bound this probability, fix a codeword $f$. Then, this probability is exactly equivalent to the probability that the complement of the support of $f$ contains the complement of $S$. Thus, assuming we only care about error patterns of size $p\cdot 2^m$ we get
		\[\err(p ,m,r) \approx 	\sum_{0 \neq f \in \polynomials{m}{r}} \frac{\binom{(1-\weight{f})2^m}{(1-p)\cdot 2^m }}{\binom{2^m}{(1-p)\cdot 2^m}} \leq
		\sum_{0 \neq f \in \polynomials{m}{r}} (1-\weight{f})^{(1-p)\cdot 2^m} \;.\]
	\end{proof}
		Thus, if we want to use Lemma~\ref{lem:erase-recover} to prove that RM codes achieve capacity for random erasures then we basically have to show that $\weightdistribution{r}{m}{\leq\beta}$ decays faster than $(1-\beta)^{(1-p)2^m}$, for $p = 1-(1+o(1))R$. In order to estimate the sum in \eqref{eq:erasure-prob}, \cite{Sberlo18} partitions the summation over $\beta$ to the dyadic intervals $[2^{-k-1}, 2^{-k}]$ and shows that each such interval sums to a small quantity. For $k\geq 2$ they use the estimate given in Theorem~\ref{thm:sberlo}. To handle the interval $[1/4,1/2]$ they partition it further to smaller dyadic intervals. Specifically, they start with polynomials of bias $\epsilon$, for some sub-constant $\epsilon$, and double the bias until they reach bias $1/2$ (equivalently, weight $1/4$). To estimate the contribution of each such sub-interval to the sum they use Theorem~\ref{thm:sberlo-bias}. This gives the following result from \cite{Sberlo18}.

\begin{theorem}[Theorem 1.9 from \cite{Sberlo18}]
\label{thm:bec-cap-low}
For any $r \leq m/50$,   $\reedmuller{r}{m}$ achieves capacity on the BEC.
\end{theorem}

In addition, they proved that for higher degrees, that lead to codes of rate $1/\poly(\log n)$, RM codes can correct a fraction $1 - o(1)$ of random erasures. The proof idea is essentially the same.

\begin{theorem}
\label{thm:bec-noise}
For any $r < m/2 - \Omega\per{{\sqrt{m\log m}}}$, $\reedmuller{r}{m}$ can efficiently correct a fraction of $1-o(1)$ random erasures (as $m$ increases).
\end{theorem}

While this result is not enough to deduce that the code achieves capacity, it shows that it is very close to doing that.\\

	\subsubsection{BSC}\label{sec:errors}
	
	Similarly to what we did in Section~\ref{sec:erasures} we shall relate the probability that we fail to correct random errors to the weight distribution of the code. For simplicity we shall assume that instead of flipping each bit with probability $p$ the channel flips exactly $p\cdot 2^m=pn$ locations at random. Consider the decoder that returns the closest codeword to the received word. A bad error pattern $z \in \F_2^n$ is one for which there exists another error pattern $z' \in \F_2^n$, of weight $s=pn$, such that $z+z'$ is a codeword in $\reedmuller{r}{m}$.
	Note that since both $z$ and $z'$ are different and have the same weight, the weight of $z+z'$ must be even and in $\{d,\dots,2pn\}$. As both $z+z'$ and the all $1$ vector  are codewords, we also have that the weight of $z+z'$ is at most $n-d$, hence $|{z+z}|\in  \{d,\dots,n-d\}$.
 Therefore, counting the number of bad error patterns is equivalent to counting the number of weight $s$ vectors that can be obtained by ``splitting'' codewords whose (un-normalized) weight is an even number in $\{d,\dots,n-d\}$. 
For a codeword $y$ of hamming weight $|{y}|=w$, there are $w/2$ choices for the support of $z$ inside\footnote{By $\text{support}(y)$ we mean the set of nonzero coordinates of $y$.} $\text{support}(y)$ and $pn-w/2$ choices outside the codeword's support ($z$ and $z'$ must cancel each other outside the support of $y$ and hence have the same weight inside $\text{support}(y)$). 
Denote by $\mathcal{B}$ the set of bad error patterns, the union bound then gives
 \begin{equation*}
 \bP\{\mathcal{B}\}  \leq \sum_{w \in \{d,\dots,n-d\}} \weightdistribution{r}{m}{w/n} \frac{{w \choose w/2}{n-w \choose pn -w/2}}{{n \choose pn}} =  \sum_{\beta \in \frac{1}{n}\cdot\{d,\dots,n-d\}} \weightdistribution{r}{m}{ \beta}\frac{{\beta n \choose \beta n/2}{(1-\beta)n \choose (p-\beta/2)n}}{{n \choose pn}}\;.   \end{equation*}
As before, Sberlo and Shpilka partition the sum to small interval around weight $1/2$ and then to dadic interval of the form $[2^{-k-1},2^{-k}]$ and bound the contribution of each interval separately to the sum using Theorems~\ref{thm:sberlo-bias} and \ref{thm:sberlo}, respectively.
They thus were able to prove the following result.

\begin{theorem}[Theorem 1.10 of \cite{Sberlo18}]
\label{thm:bsc-cap-low}
For any $r \leq m/70$,  $\reedmuller{r}{m}$ achieves capacity on the BSC.
\end{theorem}

Similarly to the BEC case, we can show that up to degrees close to $m/2$, RM codes can correct a fraction $1/2 - o(1)$ of random errors. 

\begin{theorem}
\label{thm:bsc-noise}
For any $r < m/2 - \Omega\per{{\sqrt{m\log m}}}$ the  maximum likelihood decoder for $\reedmuller{r}{m}$ can correct a fraction of $1/2-o(1)$ random errors.
\end{theorem}

As in Theorem~\ref{thm:bec-noise}, this is not enough to show that RM codes achieve capacity at this range of parameters (as the $o(1)$ term is not the correct one), but it gives a good indication that it does.

\subsection{RM codes achieve capacity on the BEC at high rate \cite{Abbe15}}
\label{sect:hj}

In this section we explain the high level idea of the proof of \cite{Abbe15} that RM codes of very high degree achieve capacity for the BEC. Similar to the case of $R\to 0$, we say that a family of codes of rate $R\to 1$ achieves capacity for the BEC if it can correct (with high probability) a fraction $p\geq (1-\epsilon)(1-R)$ of erasures. Thus, for such a code to achieve capacity for the BEC  it must hold that,  with high probability,  $\left(1-(1-\epsilon)(1-R)\right)n $ random rows of the generating span  the row space. This is equivalent to saying that if we consider the parity check matrix of the code, then a random subset of $(1-\epsilon)(1-R)n$ columns is full rank (i.e. the columns are linearly independent) with high probability. 

Since we will be dealing with very high rates (i.e. very high degrees) it will be more convenient for us to think of our code as $\reedmuller{m-r-1}{m}$ (i.e. RM code of degree $m-r-1$). 
As the parity check matrix of $\reedmuller{m-r-1}{m}$ is the generating matrix of $\reedmuller{r}{m}$, the discussion above gives rise to the question that we discuss next. 

For an input $z\in\F_2^m$ and degree $r$ denote with $z^r$ the column of $R_n$ corresponding to the evaluation point $z$ (recall Equation~\eqref{eq:def-Rn}). In other words, $z^r$ contains the evaluation of all multilinear monomials of degree at most $r$ at $z$. Thus, the code $\reedmuller{m-r-1}{m}$ achieves capacity for the BEC if it holds with high probability that a random subset $Z\subset F_2^m$, of size $|Z|=(1-\epsilon){m\choose \leq r}$, satisfies that the set $\{z^r \mid \\z\in Z\}$ is linearly independent. 

Abbe et al. \cite{Abbe15} show that this is indeed the case for small enough $r$.
\begin{theorem}\cite{Abbe15}
For $r=o(\sqrt{m/\log m})$, $\reedmuller{m-r-1}{m}$ achieves capacity on the BEC. \end{theorem}

The idea of the proof is to view the random $Z$ as if the points were picked one after the other, and to prove that as long as $r$ is not too large, the relevant set remains independent with high probability. An important observation made in \cite{Abbe15} connects this problem to a question about common zeros of degree $r$ polynomials.

\begin{claim}[Lemmas 4.6-4.9 of \cite{Abbe15}]\label{cla:common-zero}
Let $z_1,\ldots,z_s\in\F_2^m$ be such that the vectors $\{z_i^r\}$ are linearly independent. Let ${\cal I} = \{f\in \reedmuller{r}{m} \mid f(z_i)=0 \text{ for every } 1\leq i\leq s\}$. Then, for $z\in\F_2^m$, $z^r$ is not in the span of $\{z_i^r\}$ if and only if there is some $f\in {\cal I}$ such that $f(z)\neq 0$. 
\end{claim}

\begin{proof}[Proof Sketch]
Observe that a vector $c\in\F_2^{m\choose \leq r}$ is perpendicular to $z^r$ (in the sense that the dot product is zero) 
if and only if the polynomial whose coefficients are given by the coordinates of $c$ vanishes at $z$. Thus, the claim says that for a vector $z^r$ to be linearly independent of a set of vectors $Z$ it must be the case that the dual space to $Z\cup\{z^r\}$ is smaller than the dual space of $Z$. 
\end{proof}

Thus, we have to understand what is the probability that a new point is a common zero of the polynomials in $\cal I$. By dimension arguments, $\dim(\text{span}(Z))=s$ and hence $\dim(\text{span}({\cal I}))={m\choose \leq r} - s$. It remains to prove that if $r=o(\sqrt{m/\log m})$ and $s=(1-\epsilon){m\choose \leq r}$ then the number of common zeroes of the polynomials in $\cal I$ is $o(2^m)$. This will guarantee that a random point $z$ is, with high probability, not a common zero and by Claim~\ref{cla:common-zero} this will ensure that $z^r$ is linearly independent of $Z$ as desired. 

Denote with $V\subset\F_2^m$ the set of common zeroes of the polynomials in $\cal I$. That is, $V=\{z\in\F_2^m \mid \forall f\in {\cal I} \;, f(z)=0\}$.
To prove that $V$ is not too large, \cite{Abbe15} show that if $V$ was large then there would be many degree $r$ polynomials whose evaluation vectors on the points in $V$ are linearly independent. Indeed, by dimension arguments again we have that the dimension of the evaluation vectors of degree $r$ polynomials, restricted to the points in $V$, is at most ${m\choose \leq r} -\dim(\text{span}({\cal I})) = {m\choose \leq r} - \left({m\choose \leq r} - s\right) =s$. Thus, if one could show that must be more than that many linearly independent polynomials when $V$ is large, then one would get a contradiction. This is captured in the next lemma.

\begin{lemma}[Lemma 4.10 of \cite{Abbe15}]\label{lem:number_of_ind_polynomials}
Let $V\subseteq \F_2^m$ such that $|V| > 2^{m-t}$. Then there are more than ${m-t \choose \leq r}$ linearly independent polynomials of degree $\leq r$ that are defined on $V$.
\end{lemma}

The proof relies on the notion of generalized hamming weight and on a theorem of Wei \cite{Wei}. To present this result we need some definitions. Let $C\subseteq \F_2^n$ be a linear code and $D\subseteq C$ a linear subcode. We denote $\text{supp}(D) = \{ i : \exists y \in D, \text{ such that } y_i \neq 0\}$.

\begin{definition}[Generalized Hamming weight]\label{dfn:GHW_1}
For a code $C$ of length $n$ and an integer $a$ we define 
$$d_a(C) = \min \{ \text{supp}(D) \mid D\subseteq C \text{ is a linear subcode with } \dim(D) = a\}.$$ 
\end{definition}
The following lemma follows from arguments similar to what we described above. 
\begin{lemma}\label{lem:GHW_2}
For a code $C$ of length $n$ and an integer $a$ we have that 
$$d_a(C) = \max \{ b \mid \forall |S|< b \text{ we have that }\dim(C[S^c]) > \dim(C)-a\},$$
where for a set of coordinates $S$, $C[S^c]$ is the linear code obtained by projecting each codeword in $C$ to the coordinate set $S^c$, the complement of the set $S$.
\end{lemma}
Wei's theorem then gives the following identity.
\begin{theorem}[\cite{Wei}]\label{thm:GHW}
Let $0\leq a\leq {m\choose \leq r}$ be an integer. Then, 
$d_a(\reedmuller{r}{m}) = \sum_{i=1}^{\ell} 2^{m_i}$, where $a = \sum_{i=1}^{\ell}{m_i \choose \leq r_i}$ is the unique representation of $a$.\footnote{Wei also proved that  every $0\leq a\leq {m\choose \leq r}$ has a unique representation as $a = \sum_{i=1}^{\ell}{m_i \choose \leq r_i}$, where $m_i-r_i = m-r -i+1$.}
\end{theorem}

Setting parameters carefully, we can prove Lemma~\ref{lem:number_of_ind_polynomials}.

\begin{proof}[Proof of Lemma~\ref{lem:number_of_ind_polynomials}]
For $a = \sum_{i=1}^{t}{m-i\choose \leq r-1}$, Theorem~\ref{thm:GHW} implies that $d_a(\reedmuller{r}{m}) = \sum_{i=1}^{t} 2^{m-i} = 2^m - 2^{m-t}$.
Thus, if $|V| > 2^m - d_a(\reedmuller{r}{m})  = 2^{m -t}$ then there are more than ${m\choose \leq r} -  a = {m\choose \leq r} -  \sum_{i=1}^{t}{m-i \choose \leq r-1}={m-t \choose \leq r}$ many linearly independent degree-$r$ polynomials defined on $Z$. 
\end{proof}

The result showing that $\reedmuller{m-r-1}{m}$ achieves capacity for $r=o(\sqrt{m/\log m})$ now follows from the argument above by careful calculations.

\subsection{RM codes achieve capacity on the BEC at constant rate \cite{Kudekar16STOC}}

In \cite{Kudekar16STOC,Kudekar17}, Kudekar et al. made use of the classic results about sharp thresholds of monotone boolean functions \cite{Kahn88,Bourgain92,Talagrand94,Friedgut96} to prove that RM codes achieve capacity on erasure channels.  More precisely, define
$$
Q(x):=\frac{1}{\sqrt{2\pi}} \int_x^{\infty}
e^{-t^2/2} dt .
$$
\begin{theorem}\label{thm:bec-cap}
For any constant $p\in(0,1)$, 
the sequence of RM codes $\{\mathrm{RM}(m,r_m)\}$ with
$$
r_m=\max\Big\{ \Big\lfloor \frac{m}{2}+\frac{\sqrt{m}}{2}
Q^{-1}(p) \Big\rfloor,0 \Big\}
$$
has code rates approaching $1-p$ and achieves capacity over BEC with erasure probability $p$.
\end{theorem}
While their proof applies to (Generalized) RM codes with larger alphabet $\bF_q$ over the $q$-ary erasure channel, in this section we  focus on the proof for the binary RM codes over BEC.

We denote the output alphabet of BEC as $\{0,1,*\}$, where $*$ is the erasure symbol.
Consider the transmission of a codeword $x=(x_{\original{z}}:\original{z}\in\bF_2^m)\in\mathrm{RM}(m,r)$ through $n=2^m$ copies of BEC channels, and the channel output vector is denoted as $y=(y_{\original{z}}:\original{z}\in\bF_2^m)$.
The probabilistic model is set up as follows: Let $X=(X_{\original{z}}:\original{z}\in\bF_2^m)$ be a random codeword uniformly chosen from $\mathrm{RM}(m,r)$ and let $Y(p)$ be the corresponding random output vector after transmitting $X$ through $n=2^m$ copies of BEC channels with erasure probability $p$.
Since the conditional entropy $H(X_{\original{z}}|Y(p)=y)=0$ for every channel output vector $y$ with $y_{\original{z}}\neq *$,
we have
\begin{equation}\label{eq:ber}
H(X_{\original{z}}|Y(p))
=\bP(Y_{\original{z}}=*)
H(X_{\original{z}}|Y_{\sim\original{z}}(p))
=p H(X_{\original{z}}|Y_{\sim\original{z}}(p)),
\end{equation}
where $Y_{\sim\original{z}}(p):=(Y_{\original{z}'}(p):\original{z}'\in\bF_2^m,\original{z}'\neq\original{z})$.
An important property of the BEC is that for any index $\original{z}\in\bF_2^m$ and any output vector $y$, the conditional entropy $H(X_{\original{z}}|Y(p)=y)$ and $H(X_{\original{z}}|Y_{\sim\original{z}}(p)=y_{\sim\original{z}})$ can only take two values, either $0$ or $1$, i.e., given these outputs, each codeword coordinate $X_{\original{z}}$ can either be exactly reconstructed or 
$X_{\original{z}}$ is equally likely to be $0$ or $1$. In the latter case,
no useful information is provided about $X_{\original{z}}$, and we say that $X_{\original{z}}$ is ``erased". Further, since $H(X_{\original{z}}|Y(p))=\sum_{y} H(X_{\original{z}}|Y(p)=y)\bP (Y(p)=y)=\sum_{y: H(X_{\original{z}}|Y(p)=y)=1} \bP (Y(p)=y)$, we have that $H(X_{\original{z}}|Y(p))$ gives the probability that $Y(p)$ takes on values such that $X_{\original{z}}$ is equally likely to be $0$ or $1$, i.e., such that $X_{\original{z}}$ is ``erased". Therefore, the channel mapping $X_{\original{z}}$ to $Y(p)$ is equivalent to a BEC with erasure probability $H(X_{\original{z}}|Y(p))$ (and similarly for $H(X_{\original{z}}|Y_{\sim\original{z}}(p))$). When the channel output vector is $y$, the bit-MAP decoder decodes $X_{\original{z}}$ as
$$
\argmax_{x\in\{0,1\}} 
\bP(X_{\original{z}}=x|Y(p)=y),
$$
and if the maximum is not unique, i.e., if $\bP(X_{\original{z}}=0|Y(p)=y)=\bP(X_{\original{z}}=1|Y(p)=y)$, then the decoder declares an erasure and this takes place with probability $H(X_{\original{z}}|Y(p))$.
From now on, we write 
$$
h_{\original{z}}(p):=H(X_{\original{z}}|Y_{\sim\original{z}}(p)).
$$
This function is called the extrinsic information transfer (EXIT) function and it was introduced by ten Brink \cite{ten99} in the context of turbo decoding.
According to \eqref{eq:ber}, analyzing the behavior of $h_{\original{z}}(p)$ is equivalent to analyzing the decoding error probability of the bit-MAP decoder.
We also define the average EXIT function as 
$$
h(p):=\frac{1}{n}\sum_{\original{z}\in\bF_2^m}
h_{\original{z}}(p).
$$
Since RM codes are affine invariant, for any two coordinates $\original{z}$ and $\original{z}'$, we can always find a permutation in the automorphism group of RM codes mapping from $\original{z}$ to $\original{z}'$. As a consequence, one can easily show that the value of $h_{\original{z}}(p)$ is independent of the index $\original{z}$, i.e., 
$$
h(p)=h_{\original{z}}(p) \text{~~for every~} \original{z}\in\bF_2^m .
$$
Furthermore, we also have that $0\le h(p)\le 1$ for any $0\le p\le 1$. 
The average EXIT functions of some
rate-$1/2$ Reed-Muller codes are shown in Fig.~\ref{fig:shptrans}.
The following three properties of the (average) EXIT function allow us to prove that RM codes achieve capacity under bit-MAP decoding:
\begin{enumerate}
    \item $h(p)$ is a strictly increasing function of $p$.
    \item $h(p)$ has a sharp transition from $0$ to $1$. More precisely, for any $0<\epsilon<1/2$, we have 
    $$
    h^{-1}(1-\epsilon)-h^{-1}(\epsilon)\to 0
    $$ 
    as the code length $n\to\infty$.
    \item The average EXIT function satisfies the {\em area theorem}
    \begin{equation}\label{eq:area}
    \int_0^1 h(p)dp=\frac{k}{n},
    \end{equation}
    where $k$ is the dimension of the RM code and $n$ is code length.
\end{enumerate}

\begin{figure}
    \centering
    \includegraphics[width=0.45\textwidth]{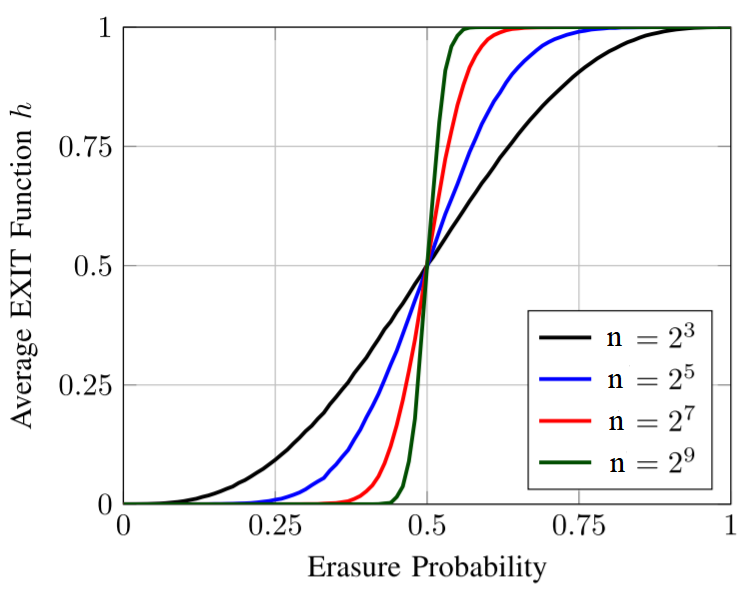}
    \caption{The average EXIT function of the rate-$1/2$ Reed-Muller code with
    blocklength $n$. This is originally Fig.~1 in \cite{Kudekar17}.}
    \label{fig:shptrans}
\end{figure}

Property 2) here simply means that as the blocklength increases, the transition width of the
average EXIT function decreases, as illustrated in Fig.~\ref{fig:shptrans}.

The first two properties imply that the function $h(p)$ has a threshold $p_*$ and a small transition width $w\to 0$ as $n\to\infty$. When $p<p_*-w/2$, we have $h(p)\approx 0$; when $p$ is around $p_*$, i.e., when $p_*-w/2\le p \le p_*+w/2$, the function $h(p)$ increases very fast from $0$ to $1$; when $p>p_*+w/2$, we have $h(p)\approx 1$. More precisely, for a given $\epsilon\in (0,1/2)$, we define the threshold as $p_*(\epsilon)=\frac{1}{2}h^{-1}(1-\epsilon)+\frac{1}{2}h^{-1}(\epsilon)$ and the transition width as $w(\epsilon)=h^{-1}(1-\epsilon)-h^{-1}(\epsilon)$. For $p<p_*(\epsilon)-w(\epsilon)/2$, we have $h(p)<\epsilon$; then the function $h(p)$ increases from $\epsilon$ to $1-\epsilon$ between $p_*(\epsilon)-w(\epsilon)/2<p<p_*(\epsilon)+w(\epsilon)/2$; for $p>p_*(\epsilon)+w(\epsilon)/2$, we have $h(p)>1-\epsilon$. Property 2) above tells us that for any given $\epsilon\in (0,1/2)$, the transition width $w(\epsilon)\to 0$ as $n\to\infty$.
Finally, property 3) above allows us to locate the threshold $p_*(\epsilon)$: Let us think of the extreme case where the transition width $w(0)=0$. Under this assumption, we have $h(p)=0$ for all $p<p_*$ and $h(p)=1$ for all $p>p_*$. Then $\int_0^1 h(p)dp=1-p_*$, so by \eqref{eq:area} we have that $p_*=1-\frac{k}{n}$.
Now back to the case where we only have $w(\epsilon)\to 0$ for any given $\epsilon$ instead of $w(0)=0$, we can still use a similar calculation to show that $p_*(\epsilon)\to 1-\frac{k}{n}$ for any given $\epsilon\in(0,1/2)$ when $n\to\infty$.
This in particular means that for any $\epsilon>0$, there exists $\delta>0$ such that $h(p)<\epsilon$ for all $p<1-\frac{k}{n}-\delta$ and all large enough $n$.
Since the decoding error of bit-MAP decoder is  $ph(p)$, we have shown that the decoding error goes to $0$ when the code rate approaches $1-p$, the capacity of BEC, i.e., we have shown that RM codes achieve capacity of BEC under bit-MAP decoder.

Next we will explain how to prove properties 1)--3).
Property 3) first appeared in \cite{Ashikhmin04}, and it follows directly from the following equality
\begin{equation} \label{eq:yz}
H(X|Y(1))-H(X|Y(0))=\int_0^1 \left(\sum_{\original{z}\in\bF_2^m}
h_{\original{z}}(p)\right) dp
\end{equation}
Indeed, $H(X|Y(1))=k$ and $H(X|Y(0))=0$, so the left-hand side of \eqref{eq:yz} is $k$. By definition, the right-hand side of \eqref{eq:yz} is $n\int_0^1 h(p)dp$. Therefore, \eqref{eq:area} follows immediately from \eqref{eq:yz}.
Equation \eqref{eq:yz} itself is implied by the results of both \cite{Ashikhmin04} and \cite{Measson08}, and it can be proved using the chain rule of derivatives as follows: We consider a slightly different model where different coordinates in the codeword are transmitted through BEC channels with (possibly) different erasure probabilities. More specifically, assume that $X_{\original{z}}$ is transmitted through a BEC with erasure probability $p_{\original{z}}$; previously we assumed that $p_{\original{z}}=p$ for all $\original{z}\in\bF_2^m$.
For a parametrized path $(p_{\original{z}}(t):\original{z}\in\bF_2^m)$ defined for
$t \in [0, 1]$, one finds that
$$
H(X|Y(1))-H(X|Y(0))=\int_0^1 \left(\sum_{\original{z}\in\bF_2^m}
\frac{\partial H(X|Y(p_{\original{z}}:\original{z}\in\bF_2^m))}{\partial p_{\original{z}}} \frac{d p_{\original{z}}}{dt} \right) dt.
$$
One can further show (e.g. using \eqref{eq:ber}) that 
$$
H(X_{\original{z}}|Y_{\sim\original{z}}(p_{\sim\original{z}}))=
\frac{\partial H(X|Y(p_{\original{z}}:\original{z}\in\bF_2^m))}{\partial p_{\original{z}}}.
$$
Therefore,
$$
H(X|Y(1))-H(X|Y(0))=\int_0^1 (\sum_{\original{z}\in\bF_2^m}
H(X_{\original{z}}|Y_{\sim\original{z}}(p_{\sim\original{z}}))
\frac{d p_{\original{z}}}{dt} ) dt.
$$
Equation \eqref{eq:yz} follows immediately by
taking the parametrized path $(p_{\original{z}}(t):\original{z}\in\bF_2^m)$ to be $p_{\original{z}}(t)=t$ for all $\original{z}\in\bF_2^m$ in the above equation.

Both property 1) and property 2) are established by connecting $h(p)$ to monotone and symmetric boolean functions and making use of the classic results on such functions \cite{Kahn88,Bourgain92,Talagrand94,Friedgut96}.
In order to connect $h(p)$ to boolean functions, we use the fact that $h(p)=h_{\original{z}}(p)$ for every $\original{z}\in\bF_2^m$. Below we fix an index $\original{z}$ and associate an erasure pattern $e_{\sim\original{z}}(y_{\sim\original{z}})=(e_{\original{z}'}(y_{\sim\original{z}}):\original{z}'\in\bF_2^m,\original{z}'\neq\original{z}) \in\{0,1\}^{n-1}$ with each channel output vector $y_{\sim\original{z}}$ as follows:
$$
e_{\original{z}'}(y_{\sim\original{z}})=\mathbf{1}[y_{\original{z}'}=*]
\text{~for all~}
\original{z}'\in(\bF_2^m\setminus\{\original{z}\}),
$$
i.e., $e_{\sim\original{z}}(y_{\sim\original{z}})$ is an indicator vector of the erasures in $y_{\sim\original{z}}$.
We further define the set $\Omega_{\original{z}}\subseteq\{0,1\}^{n-1}$ as
$$
\Omega_{\original{z}}:=
\{e_{\sim\original{z}}(y_{\sim\original{z}}): H(X_{\original{z}}|
Y_{\sim\original{z}}(p)=y_{\sim\original{z}}) =1\} .
$$
This is the set of erasure patterns that prevent the recovery of $X_{\original{z}}$ from $Y_{\sim\original{z}}(p)$. Note that the set $\Omega_{\original{z}}$ is completely determined by the RM code construction and it is independent of the channel erasure probability $p$.
We say that an erasure pattern $e_{\sim\original{z}}$ covers another erasure pattern $\tilde{e}_{\sim\original{z}}$
if $e_{\original{z}'}\ge \tilde{e}_{\original{z}'}$ for all $\original{z}'\in(\bF_2^m\setminus\{\original{z}\})$. The set
$\Omega_{\original{z}}$ has the following {\em monotone} property: if  (1) $\tilde{e}_{\sim\original{z}}\in \Omega_{\original{z}}$ and (2) $e_{\sim\original{z}}$ covers  $\tilde{e}_{\sim\original{z}}$, then $e_{\sim\original{z}}\in \Omega_{\original{z}}$.

For each erasure pattern $e_{\sim\original{z}}$, we use $|e_{\sim\original{z}}|$ to denote its Hamming weight. 
We write the probability of erasure pattern $e_{\sim\original{z}}$ under BEC with channel erasure probability $p$ as $\mu_p(e_{\sim\original{z}})$, and it is given by 
$$
\mu_p(e_{\sim\original{z}})
=p^{|e_{\sim\original{z}}|}(1-p)^{n-1-|e_{\sim\original{z}}|} .
$$
Using this notation,
$$
\mu_p(\Omega_{\original{z}})
=\sum_{e_{\sim\original{z}}\in\Omega_{\original{z}}} p^{|e_{\sim\original{z}}|}(1-p)^{n-1-|e_{\sim\original{z}}|}.
$$
By definition, we have
\begin{equation}\label{eq:opw}
h(p)=h_{\original{z}}(p)=\mu_p(\Omega_{\original{z}}).
\end{equation}

Next we define the ``boundary" of $\Omega_{\original{z}}$. For an erasure pattern $e_{\sim\original{z}}\in\{0,1\}^{n-1}$, we define $\flip(e_{\sim\original{z}},\original{z}')$ to be another erasure pattern obtained by flipping the $\original{z}'$ coordinate in $e_{\sim\original{z}}$.
The ``boundary" of $\Omega_{\original{z}}$ across the $\original{z}'$ coordinate is defined as 
$$
\partial_{\original{z}'}\Omega_{\original{z}}
:= \{e_{\sim\original{z}}:
\mathbf{1}[e_{\sim\original{z}}\in\Omega_{\original{z}}]
\neq \mathbf{1}[\flip(e_{\sim\original{z}},\original{z}')\in\Omega_{\original{z}}]\}.
$$
Notice that $\partial_{\original{z}'}\Omega_{\original{z}}$ contains the ``boundary" erasure patterns both inside and outside of $\Omega_{\original{z}}$.
This is because by definition if $e_{\sim\original{z}}\in\partial_{\original{z}'}\Omega_{\original{z}}$, then $\flip(e_{\sim\original{z}},\original{z}')\in \partial_{\original{z}'}\Omega_{\original{z}}$.
The probability measure $\mu_p(\partial_{\original{z}'}\Omega_{\original{z}})$ is called the {\em influence} of the $\original{z}'$ coordinate,
and one further defines the {\em total influence} as
$$
I^{(p)}(\Omega_{\original{z}})
:=\sum_{\original{z}'\in(\bF_2^m\setminus\{\original{z}\})} \mu_p(\partial_{\original{z}'}\Omega_{\original{z}})  ;
$$
see \cite{Margulis74,Russo82,Boucheron13,oDonnell14} for discussions and properties of these functions.

An important consequence of the monotone property is the following equality \cite{Margulis74,Russo82,Boucheron13} connecting the derivative of $\mu_p(\Omega_{\original{z}})$ with the total influence
$$
\frac{d\mu_p(\Omega_{\original{z}})}{dp}
=  I^{(p)}(\Omega_{\original{z}}) .
$$
Property 1) of $h(p)$ follows immediately from this equality and \eqref{eq:opw}.

As for Property 2) of $h(p)$, observe that the set $\Omega_{\original{z}}$ also has a symmetric property which follows from the affine-invariant property of RM codes. More precisely, for any $\original{z}',\original{z}''\in(\bF_2^m\setminus\{\original{z}\})$, the influences of the $\original{z}'$ and $\original{z}''$ coordinates are always the same for all $0\le p\le 1$, i.e., we always have
$\mu_p(\partial_{\original{z}'}\Omega_{\original{z}})
=\mu_p(\partial_{\original{z}''}\Omega_{\original{z}})$.
This is because for any triple $\original{z},\original{z}',\original{z}''\in\bF_2^m$ we can always find an invertible affine linear transform over $\bF_2^m$ that fixes $\original{z}$ and maps $\original{z}'$ to $\original{z}''$.

This symmetric property of $\Omega_{\original{z}}$ allows us to use the following classic result on Boolean functions: 
\begin{theorem}[\cite{Kahn88,Bourgain92,Talagrand94,Friedgut96}]
Let $\Omega\subseteq\{0,1\}^M$ be a monotone set and suppose that for all $0\le p\le 1$, the influences of all bits under the measure $\mu_p$ are equal. Then for all $0\le p\le 1$,
$$
\frac{d\mu_p(\Omega)}{dp}
= I^{(p)}(\Omega)\ge 
\ln(M) \mu_p(\Omega)
(1-\mu_p(\Omega)) .
$$
\end{theorem}

Applying this theorem to $\Omega_{\original{z}}$, we obtain that for all $0<p<1$,
$$
\frac{d\mu_p(\Omega_{\original{z}})}{dp}
= I^{(p)}(\Omega_{\original{z}})\ge 
\ln(n-1) \mu_p(\Omega_{\original{z}})
(1-\mu_p(\Omega_{\original{z}}))  .
$$
Combining this with \eqref{eq:opw}, we have
\begin{equation}\label{eq:lgdt}
\frac{dh(p)}{dp} \ge 
\ln(n-1) h(p)
(1-h(p)).
\end{equation}
This inequality tells us that for any given $\epsilon\in(0,1/2)$, we always have $\frac{dh(p)}{dp}\to\infty$ as $n\to\infty$ for all $p$ satisfying $h(p)\in(\epsilon,1-\epsilon)$.
Therefore, $h(p)$ has a sharp transition from $\epsilon$ to $1-\epsilon$ for arbitrarily small $\epsilon>0$ when the code length $n$ is large enough, and this proves property 2) of $h(p)$.

Now this shows that RM codes achieve capacity of BEC under the bit-MAP decoder. In \cite{Kudekar16STOC,Kudekar17}, it was also proved that RM codes achieve capacity of BEC under the block-MAP decoder. The proof uses the same framework and relies on a refinement of \eqref{eq:lgdt} in \cite{Bourgain97}. In particular,
Bourgain and Kalai \cite{Bourgain97} showed that there exists a universal constant $C$ such that
$$
\frac{dh(p)}{dp} \ge 
C \ln(\ln(n-1))
\ln(n-1) h(p) (1-h(p)).
$$
The proof for block-MAP decoder then follows the same line of arguments as the proof for bit-MAP decoder but requires a few more technical details which we will omit here.

\begin{remark}
For BEC, the conditional entropy $H(X_{\original{z}}|Y_{\sim\original{z}}(p)=y_{\sim\original{z}})$ can only be either $0$ or $1$, and it is independent of the channel erasure probability $p$.
This property allows us to connect $h(p)$ with boolean functions and use the classic results to analyze $h(p)$. 
Unfortunately, this property no longer holds for general communication channels such as BSC. Now suppose instead that $Y(p)$ is the  random output vector after transmitting $X$ through $n=2^m$ copies of BSC with channel crossover probability $p$. In this case, the conditional entropy $H(X_{\original{z}}|Y_{\sim\original{z}}(p)=y_{\sim\original{z}})$ will vary with $p$ and its range is clearly not limited to $\{0,1\}$.
Moreover, the monotone property of the set $\Omega_{\original{z}}$ is also a consequence of transmitting codewords through erasure channels, and it does not hold for general communication channels like BSC.
In other words, for BSC, more errors do not necessarily lead to larger conditional entropy
$H(X_{\original{z}}|Y_{\sim\original{z}}(p)=y_{\sim\original{z}})$. For example, since the all-one vector is a codeword of RM codes, we always have the equality
$$
H(X_{\original{z}}|Y_{\sim\original{z}}(p)=\mathbf{0})
=H(X_{\original{z}}|Y_{\sim\original{z}}(p)=\mathbf{1}),
$$ 
where $\mathbf{0}$ and $\mathbf{1}$ are the all-zero and all-one vectors, respectively. Now assume that we transmit the all-zero codeword. Then this equality means that the conditional entropy given an output with no errors is the same as the conditional entropy given an output that is erroneous in every bit. This clearly does not satisfy the monotone property. Finally, the area theorem (see \eqref{eq:area}) for EXIT function only holds for BEC. For general channels, one needs to work with the generalized EXIT (GEXIT) function introduced in \cite{Measson09}. The GEXIT function is similar in many respects to the EXIT function: The GEXIT function satisfies the area theorem for general channels and it is neither boolean nor monotonic. New ideas are certainly required to analyze such functions in order to generalize the method of \cite{Kudekar17,Kudekar16STOC} for the cases of BSC or more general communication channels.
\end{remark}

\begin{remark}
Note that the approach in this subsection only works for the constant rate regime and does not extend to the extremal rate regimes discussed in Section~\ref{sec:wt-to-cap} and Section~\ref{sect:hj}.
\end{remark}

\subsection{RM codes polarize and Twin-RM codes achieve capacity on any BMS \cite{AY18}}
\label{sect:RMpolar}

\begin{figure}
\centering
\begin{subfigure}{0.48\textwidth}
\centering
\begin{tikzpicture}
\node [sblock, align=center] at (3,1.6)  (y1) { $U_1$ \\[0.5em]  $U_2$  \\[0.5em]  \vdots \\[0.5em]  $U_n$ };
\node [sblock, align=center] at (5,1.6)  (y2) {$G_n$: $n\times n$\\tensor\\ product
 \\ matrix};
\node [sblock, align=center] at (7,1.6)  (y3) { $X_1$ \\[0.5em] $X_2$  \\[0.5em] \vdots  \\[0.5em]  $X_n$ };
\node [block] at (8.3, 2.7) (w1) {$W$};
\node [block] at (8.3, 2) (w2) {$W$};
\node  at (8.3, 1.3)  {\vdots};
\node [block] at (8.3, 0.6) (w3) {$W$};

\node at (9.5, 2.7) (z1) {$Y_1$};
\node at (9.5, 2) (z2) {$Y_2$};
\node  at (9.5, 1.3) {\vdots};
\node at (9.5, 0.6) (z3) {$Y_n$};

\draw[->,thick] (y1)--(y2);
\draw[->,thick] (y2)--(y3);

\draw[->,thick] (w1)--(z1);
\draw[->,thick] (w2)--(z2);
\draw[->,thick] (w3)--(z3);

\draw[->,thick] (7.4, 2.7)--(w1);
\draw[->,thick] (7.4, 2)--(w2);
\draw[->,thick] (7.4, 0.6)--(w3);
\end{tikzpicture}
\caption{Polarization framework}
\label{fig:bj1}
\end{subfigure}
\hfill
\begin{subfigure}{0.48\textwidth}
\centering
\begin{tikzpicture}
\node [sblock, align=center] at (3,1.6)  (y1) { $U_{A_1}$ \\[0.5em]  $U_{A_2}$  \\[0.5em]  \vdots \\[0.5em]  $U_{A_n}$ };
\node [sblock, align=center] at (5,1.6)  (y2) {$R_n: n \times n$
\\generator\\matrix of\\ $\mathrm{RM}(m,m)$};
\node [sblock, align=center] at (7,1.6)  (y3) { $X_1$ \\[0.5em] $X_2$  \\[0.5em] \vdots  \\[0.5em]  $X_n$ };
\node [block] at (8.3, 2.7) (w1) {$W$};
\node [block] at (8.3, 2) (w2) {$W$};
\node  at (8.3, 1.3)  {\vdots};
\node [block] at (8.3, 0.6) (w3) {$W$};

\node at (9.5, 2.7) (z1) {$Y_1$};
\node at (9.5, 2) (z2) {$Y_2$};
\node  at (9.5, 1.3) {\vdots};
\node at (9.5, 0.6) (z3) {$Y_n$};

\draw[->,thick] (y1)--(y2);
\draw[->,thick] (y2)--(y3);

\draw[->,thick] (w1)--(z1);
\draw[->,thick] (w2)--(z2);
\draw[->,thick] (w3)--(z3);

\draw[->,thick] (7.4, 2.7)--(w1);
\draw[->,thick] (7.4, 2)--(w2);
\draw[->,thick] (7.4, 0.6)--(w3);
\end{tikzpicture}
\caption{Polarization framework for RM codes}
\label{fig:bj2}
\end{subfigure}
\caption{Polarization framework: The input vector $U^n$ are i.i.d. Bernoulli($1/2$) random variables. Given the channel output vector $Y^n$, we use successive decoder to decode the input vector $U^n$ one by one from top to bottom.}
\end{figure}
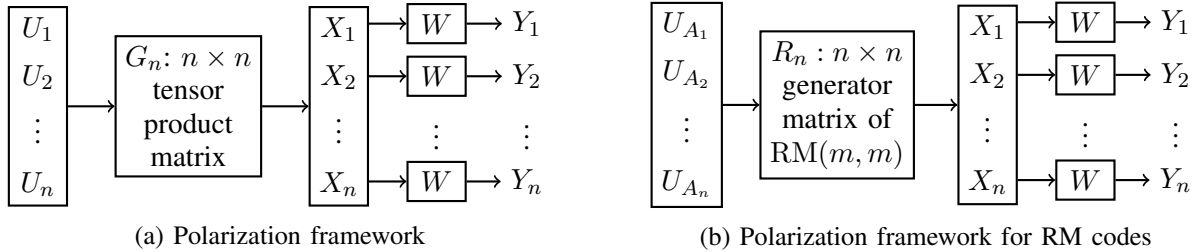


Recall that a binary memoryless symmetric (BMS) channel is a channel $W:\{0,1\}\to\cY$ such that there is a permutation $\pi$ on the output alphabet $\cY$ satisfying i) $\pi^{-1}=\pi$ and ii) $W(y|1)=W(\pi(y)|0)$ for all $y\in\cY$. In particular, the binary erasure channel (BEC), the binary symmetric channel (BSC), and the binary input additive white Gaussian noise channel (BIAWGN) are all BMS channels.

We now consider an arbitrary BMS  channel $W$.
We use the communication model in Fig.~\ref{fig:bj1} to transmit information over $W$. The input vector $U^n$ consists of $n$ i.i.d. Bernoulli($1/2$) random variables. We encode $U^n$ by multiplying it with an $n\times n$ invertible matrix $G_n$ and denote the resulting vector as $(X_1,\dots,X_n)=(U_1,\dots,U_n) G_n$. Then we transmit each $X_i$ through an independent copy of $W$. Given the channel output vector $Y^n$, our task is to recover the input vector $U^n$, and we use a successive decoder to do so. The successive decoder decodes the input vector bit by bit from $U_1$ to $U_n$. When decoding $U_i$, it makes use of all the channel outputs $Y^n$ and all the previously decoded\footnote{Assuming no decoding errors up to $U^{i-1}$.} inputs $U^{i-1}$.
The conditional entropy 
$$
H_i:=H(U_i|U^{i-1},Y^n)
$$
indicates whether $U_i$ is noisy or noiseless under the successive decoder: If $H_i\approx 0$, then $U_i$ is (almost) noiseless and can be correctly decoded with high probability. If $H_i$ is bounded away from $0$, then so is the decoding error probability of decoding $U_i$.

Informally, we say that the matrix $G_n$ polarizes if almost all $H_i,1\le i\le n$ are close to either $0$ or $1$, or equivalently, if almost all $U_i,i\le i\le n$ become either noiseless or completely noisy.
An important consequence of polarization is that every polarizing matrix automatically gives us a capacity-achieving code under the successive decoder. Indeed, if $G_n$ polarizes, then we can construct the capacity achieving code by putting all the information in the $U_i$'s whose corresponding $H_i$ is close to $0$ and freezing all the other $U_i$'s to be $0$, i.e., we only put information in the (almost) noiseless $U_i$'s.
Let $\cG$ be the set of indices of the noiseless $U_i$'s. Then the generator matrix of this code is the submatrix of $G_n$ obtaining by retaining only the rows whose indices belong to $\cG$.
To show that this code achieves capacity, we only need to argue that $|\cG|$ is asymptotic to $nI(W)$, and this directly follows from
\begin{align}
    \sum_{i=1}^n H_i=H(U^n|Y^n)=H(X^n|Y^n)
=nH(X_i|Y_i)=n(1-I(W)), \label{bal}
\end{align}
where the last equality relies on the assumption that $W$ is symmetric. Since almost all $H_i$'s are close to either $0$ or $1$, by the equation above we know that the number of $H_i$'s that are close to 1 is asymptotic to $n(1-I(W))$, so the number of $H_i$'s that are close to $0$ is asymptotic $nI(W)$, i.e., $|\cG|$ is asymptotic to $nI(W)$.

In his influential  paper \cite{Arikan09}, Ar{\i}kan gave an explicit construction of a polarizing matrix
$$
G_n:=\begin{bmatrix}
1 & 0 \\
1 & 1
\end{bmatrix}^{\otimes m} ,
$$
where $\otimes$ is the Kronecker product and $n=2^m$.
For example
$$
G_4:=\begin{bmatrix}
1 & 0 \\
1 & 1
\end{bmatrix}^{\otimes 2}
=\begin{bmatrix}
1 & 0 & 0 & 0 \\
1 & 1 & 0 & 0 \\
1 & 0 & 1 & 0 \\
1 & 1 & 1 & 1
\end{bmatrix} .
$$
Polar codes are simply the capacity-achieving codes constructed from $G_n$. More precisely:
\begin{theorem}[Polarization for $G_n$]
For any BMS channel $W$ and any $0<\epsilon<1/2$, 
\begin{align}
\left| \left\{i\subseteq[2^m]:  H_i \in (\epsilon, 1-\epsilon) \right\}
\right| = o(2^m). \label{pol}
\end{align}
\end{theorem}
In particular, the above still holds for some choices of $\epsilon=\epsilon_n$ that are $o(1/n)$ (in fact, $\epsilon_n$ can even decay exponentially with roughly the square-root of $n$).
Therefore, the polar code retaining only  rows $i$ of $G_n$ such that $H_i \le \epsilon_n$ has a block error probability that is upper-bounded by $|G_n| \epsilon_n \le n \epsilon_n$,
and if 
$\epsilon_n=o(1/n)$, the block error probability is upper-bounded by $n \epsilon_n=o(1)$, and the code achieves capacity.

The encoding procedure of polar codes amounts to finding $\cG$, the set of noiseless (or ``good'') bits, and efficient algorithms for finding these were proposed in \cite{Arikan09,Tal13}.
In \cite{Arikan09}, Ar{\i}kan also showed that the successive decoder for polar codes allows for an $O(n\log n)$ implementation. Later in \cite{Tal15}, a list decoding version of the successive decoder was proposed, and its performance is nearly the same as the Maximum Likelihood (ML) decoder of polar codes for a wide range of parameters.

In \cite{AY18}, the authors develop a similar polarization framework to analyze RM codes; see Fig.~\ref{fig:bj2} for an illustration.
More precisely, the monomials $x_{A_1},\dots,x_{A_n}$ defined by the $n:=2^m$ subsets $A_1,\dots,A_n$ of $[m]$, are used in replacement to the increasing integer index $i$ in $[n]$.
We arrange these subsets in the following order: Larger sets always appear before smaller sets; for sets with equal size, we use the lexicographic order. More precisely, if $i<j$, we always have $|A_i|\ge |A_j|$, and we have $A_i < A_j$ --where $<$ denotes the lexicographic order-- if $|A_i|=|A_j|$.
Define the matrix 
\begin{equation}\label{eq:def-Rn}
R_n:=\begin{bmatrix}
\Eval(x_{A_1}) \\
\Eval(x_{A_2}) \\
\vdots   \\
\Eval(x_{A_n}) 
\end{bmatrix} 
\end{equation}
whose row vectors are arranged according to the order of the subsets. By definition, $R_n$ is a generator matrix of $\mathrm{RM}(m,m)$.
Here we give a concrete example of the order of sets and $R_n$ for $m=3$ and $n=2^m=8$:
$$
\begin{array}{c}
    A_1=\{1,2,3\}  \\
    A_2=\{1,2\}   \\
    A_3=\{1,3\}   \\
    A_4=\{2,3\}   \\
    A_5=\{1\}     \\
    A_6=\{2\}     \\
    A_7=\{3\}     \\
    A_8=\emptyset
\end{array}
\quad\quad
R_8=\begin{bmatrix}
1 & 0 & 0 & 0 & 0 & 0 & 0 & 0  \\
1 & 1 & 0 & 0 & 0 & 0 & 0 & 0  \\
1 & 0 & 1 & 0 & 0 & 0 & 0 & 0  \\
1 & 0 & 0 & 0 & 1 & 0 & 0 & 0  \\
1 & 1 & 1 & 1 & 0 & 0 & 0 & 0  \\
1 & 1 & 0 & 0 & 1 & 1 & 0 & 0  \\
1 & 0 & 1 & 0 & 1 & 0 & 1 & 0  \\
1 & 1 & 1 & 1 & 1 & 1 & 1 & 1
\end{bmatrix}   .
$$
Note that $R_n$ is a row permutation of $G_n$.
Let $U_{A_1},\dots,U_{A_n}$ be the (random) coefficients of the monomials $x_{A_1},\dots,x_{A_n}$. Multiplying the coefficient vector $(U_{A_1},\dots,U_{A_n})$ with $R_n$ gives us a mapping from the coefficient vector to the evaluation vector $X^n:=(U_{A_1},\dots,U_{A_n}) R_n$. 
Then we transmit each $X_i$ through an independent copy of $W$ and get the channel output vector $Y^n$. We still use the successive decoder to decode the coefficient vector bit by bit from $U_{A_1}$ to $U_{A_n}$.
Similarly to $H_i$, we  define the conditional entropy 
$$
H_{A_i}:=H(U_{A_i}|U_{A_1},\dots,U_{A_{i-1}},Y^n),
$$
and by the chain rule we also have the balance equation 
\begin{align}
    \sum_{i=1}^n H_{A_i}=n(1-I(W)), \label{bal-rm}
\end{align}

In \cite{AY18}, a polarization result for $H_{A_i}$ is obtained, showing that with this ordering too, almost all $H_{A_i}$ are close to either $0$ or $1$. More precisely:
\begin{theorem}[Polarization of RM codes]
\label{thm:PRM}
For any BMS channel $W$ and any $0<\epsilon<1/2$, 
\begin{align}
\left| \left\{i \in [2^m]:  H_{A_i} \in (\epsilon, 1-\epsilon) \right\}
\right| = o(2^m). \label{rm-pol}
\end{align}
\end{theorem}
In particular, the above still holds for some choices of $\epsilon=\epsilon_n$ that decay faster than $1/n$. Therefore the code obtained by retaining only the monomials $x_{A_i}$ in $R_n$ corresponding to $A_i$'s such that $H_{A_i} \le \epsilon_n$ has a vanishing block error probability and achieves capacity; this follows from the same reasoning as in polar codes. We call this code the Twin-RM code as it is not necessarily the RM code. In fact, if the following implication were true,
\begin{align}
|A| > |B| \stackrel{?}{\implies}     H_A \ge H_B, \label{equi}
\end{align}
then the Twin-RM would be exactly the RM code, and the latter would also achieve capacity on any BMS. The same conclusion would hold if \eqref{equi} held true for most sets; it is nonetheless conjectured in \cite{AY18} that \eqref{equi} holds in the strict sense. To further support this claim, \cite{AY18} provides two partial results:
\begin{enumerate}
    \item[(i)] Partial order: 
    \begin{align}
        A \supseteq B \implies H_{A} \ge H_{B}, \label{equi2}
    \end{align}
    more generally, the implication is shown to hold if there exists $\tilde{B}$ s.t.\ $A \supseteq \tilde{B}$, $|\tilde{B}|=|B|$ and $\tilde{B}$ is less than $B$ and each component of $\tilde{B}$ is smaller than or equal to the corresponding component of $B$ (as integers).     

    \item[(ii)] For the BSC, \eqref{equi} is proved up to $2^m=16$, and numerically verified for some larger block lengths.
\end{enumerate}

It is also shown in \cite{AY18} that it suffices to check \eqref{equi} for specific subsets. It is useful at this point to introduce the  division of the input bits of $\mathrm{RM}(m,r)$ into $m+1$ layers, where the $j$th layer corresponds to the subsets of $[m]$ with size $j$, and the range of $j$ is from $0$ to $m$. Therefore, the $0$th layer only has one bit $U_1$, and the first layer has $m$ bits $U_2,U_3,\dots,U_{m+1}$. In general, the $i$th layer has $\binom{m}{i}$ bits.
It is shown in \cite{AY18} that it suffices to check \eqref{equi} for subsets $A,B$ that are respectively the last and first subsets in consecutive layers, as these are shown to achieve respectively the largest and least entropy within layers. We now present the proof technique for Theorem \ref{thm:PRM}.

{\bf Proof technique for Theorem \ref{thm:PRM}.} When $m=1$ and $n=2$, the polar and RM matrices are the same:   
$$
G_2=R_2=\begin{bmatrix}
1 & 0 \\
1 & 1
\end{bmatrix}
.$$
Define $H(W):=1-I(W)$. The balance equation gives
$$2H(W)=H(U_1|Y_1,Y_2) +H(U_2|Y_1,Y_2,U_1),$$
and since $H(U_2|Y_1,Y_2,U_1) \le H(U_2|Y_2)=H(W)$, 
we can create two synthetic channels $W^-: U_1 \to (Y_1,Y_2)$ and $W^+: U_2 \to (Y_1,Y_2,U_1)$ such that for $\Delta:= H(W^+)-H(W)$,
\begin{align}
    H(W^-)= H(W)+\Delta/2 \ge H(W) \ge H(W^+)= H(W)-\Delta/2, \label{sym} 
\end{align}
 and the above inequalities are strict unless $\Delta=H(W^+)-H(W)=0$, which is equivalent to $H(W)\in \{0,1\}$,
i.e., the channel is already extremal. One can write a quantitative version of this, i.e., for any binary input symmetric output channel $W$, there exists a positive continuous function $\delta$ on $[0,1/2)$ that vanishes only at $0$, such that for any $\epsilon \in (0,1/2)$,
\begin{align}
    H(W) \in (\epsilon,1-\epsilon) \,\,\, \implies \,\,\,
    \Delta=H(W^+)-H(W) \ge \delta(\epsilon) \label{extreme}.
\end{align}
Therefore, unless the initial channel $W$ was already close to extremal, the synthesized channel $W^+$ is strictly better by a bounded amount. 

If we move to $m=2$, then we still have $G_4=R_4$, and we can create four synthetic channels $W^{--}, W^{-+}, W^{+-}, W^{++}$ corresponding to the channels mapping $U_i$ to $(Y^4,U^{i-1})$ where the binary expansion of $i$ (mapping $0$ to $-$ and $1$ to $+$) gives the channel index. Note that the behavior of these channels at $m=2$ can be related to that at $m=1$, as suggested by the notation $W^{s_1 s_2}$, $s_1,s_2 \in \{-,+\}$. Namely, the two channels $(W^{*-},W^{*+})$ are the synthesized channels obtained by composing two independent copies of $W^*$, $* \in \{-,+\}$ with the transformation $G_2$, as done for $m=1$.

In the case of polar codes, one intentionally preserves this induction. Namely, after obtaining $n$ synthetic channels $W^{s}$, $s \in \{-,+\}^m$, one produce twice more channels with the + and - versions of these, such that for each one:
\begin{align}
    H(W^{s-})= H(W^s)+\Delta_s/2 \ge H(W^s) \ge H(W^{s+})= H(W^s)-\Delta_s/2, \label{sym2} 
\end{align}
with $\Delta_s:= H(W^{s+})-H(W^s)$, $s \in \{-,+\}^m$.
This follows by the inductive property of $G_{2n}$:
$$G_{2n}  =\begin{bmatrix}
G_n & 0 \\
G_n & G_n
\end{bmatrix}.$$
The polarization result is then a consequence of this recursive process: if one tracks the entropies of the $2^m$ channels at level $m$, at the next iteration, i.e., at level $m+1$, one breaks symmetrically each of the previous entropies into one strictly lower and one strictly larger value, as long as the produced entropies are not extremal, i.e., not 0 or 1. Therefore, extremal configurations are the only stable points of this process, and most of the values end up at these extremes. This can be deduced by using the martingale convergence theorem,\footnote{It is sufficient to check that the  variance of the $2^m$ entropies at time $m$ decreases when $m$ increases.} and the fact that the only fixed point of the $-,+$ transform are at the extremes, or conversely, that for values that are not  close to the extremes, a strict movement takes place as in \eqref{extreme}. One needs however a stronger condition than \eqref{extreme}. In fact, \eqref{extreme} holds for any $W$, but the function $\delta$ could depend on the channel $W$, and we need to apply \eqref{extreme} to the channels $W^s$ that have an output alphabet that grows as $s$ grows. We therefore need a universal function $\delta$ that applies to all the channels $W^s$, so that we can lower-bound $\Delta_s$ universally when $H(W^s) \in (\epsilon,1-\epsilon)$. Note that if $W$ is a BMS channels, then each $W^s$ is still a BMS channel, so we can restrict ourselves to this class of channels. The existence of a universal function $\delta$ follows then from the following inequality,\footnote{This strong inequality may not be needed to obtain a universal function $\delta$; weaker bounds and functions can be obtained, as for example in \cite{Blasiok18} for non-binary alphabets.} which can be proved as a consequence of the so-called Mrs.\ Gerber's Lemma\footnote{\eqref{gerber} is a convexity argument: $H(U_1+U_2|T_1,T_2) - H(U_2|T_2)
=\sum_{t_1,t_2}[H_2(p_{t_1}\star p_{t_2}) - H_2(p_{t_2})]P(T_2=t_2)P(T_1=t_1)$, where $p_{t} =H(U_1|T_1=t)$, $H_2$ is the binary entropy function, and one can introduce the expectation inside $H_2$ due to the convexity property of Mrs.\ Gerber's Lemma \cite{Wyner73}.}:
\begin{align}
H(U_1+U_2|T_1,T_2) - H(U_2|T_2) \ge H(V_1+V_2)-H(V_2) \label{gerber}
\end{align}
where $(U_1,T_1)$, $(U_2,T_2)$ are i.i.d.\ with $U_1$ binary uniform, $T_1$ is arbitrary discrete valued, and $V_1,V_2$ are i.i.d.\ binary uniform such that $H(U_2|T_2)=H(V_2)$. Therefore, the entropy spread $H(W^{s+})-H(W^s)$ for any $s$ is as large as the entropy spread of a simple BSC that has a matching entropy, since $1-H(V_2)$ is a BSC capacity and $H(U_2|T_2)=H(V_2)$.
The universal function $\delta$ can then be found  explicitly by inspecting the right hand side of \eqref{gerber}, and for any $\epsilon \in (0,1/2)$, there exists $\delta(\epsilon) >0$ such that for any $m \ge 1$, $s \in \{-,+\}^m$, 
\begin{align}
    H(W^s) \in (\epsilon,1-\epsilon) \,\,\, \implies \,\,\,
    \Delta_s=H(W^{s+})-H(W^s) \ge \delta(\epsilon) \label{eq:spread}.
\end{align}
For RM codes, the inductive argument is broken. In particular, we no longer have a symmetric break of the conditional entropies as in \eqref{sym2}, hence no obvious martingale argument.  Nonetheless, we next argue that the loss of the inductive/symmetric structure takes place in a favorable way, i.e., the spread in \eqref{sym2} tends to be greater for RM codes than for polar codes. In turn, we claim that the conditional entropies polarize faster in the RM code ordering (see \cite{AY18}). We next explain this and show how one can take a short-cut to show that RM codes polarize using increasing chains of subsets, exploiting the Plotkin recursive structure of RM codes and known inequalities from polar codes. We first need to define increasing chains. 
\begin{definition}[Increasing chains]
We say that $\emptyset=B_0 \subseteq B_1 \subseteq B_2 \subseteq \dots \subseteq B_m=[m]$ is an increasing chain if $|B_i|=i$ for all $i=0,1,2,\dots,m$.
\end{definition}

As for polar codes, we will make use of the recursive structure of RM codes, i.e., the fact that $\mathrm{RM}(m+1,m+1)$ can be decomposed into two independent copies of $\mathrm{RM}(m,m)$. In order to distinguish $H_A$'s for RM codes with different parameters, we add a superscript to the notation, writing $H_A^{(m)}$ instead of $H_A$.
A main step in our argument consist in proving the following theorem:
\begin{theorem}[RM polarization on chains]\label{thm:m}
For every BMS channel $W$, every positive $m$ and every increasing chain $\{B_i\}_{i=0}^m$, we have
\begin{align}
H_{B_0}^{(m)} \le H_{B_1}^{(m)} \le H_{B_2}^{(m)} \le \dots \le H_{B_m}^{(m)}. \label{thm:m1}
\end{align}
Further, for any $\epsilon\in (0,1/2)$, there is a constant $D(\epsilon)$ such that for every positive $m$ and every increasing chain $\{B_i\}_{i=0}^m$,
\begin{align}
    \left|\left\{ i\in\{0,1,\dots,m\}:\epsilon < H_{B_i}^{(m)} < 1-\epsilon \right\} \right| \le D(\epsilon). \label{thm:m2}
\end{align}
\end{theorem}
Note that $D(\epsilon)$ does not depend on $m$ here. This theorem relies strongly on the following interlacing property over chains. 
\begin{lemma}[Interlacing property]\label{interlacing}
For every BMS channel $W$, every positive $m$ and every increasing chain $\{B_i\}_{i=0}^m$, we have
\begin{equation} \label{eq:jjw1}
 H_{B_i}^{(m+1)} \le H_{B_i}^{(m)} \le H_{B_{i+1}}^{(m+1)}  \quad \forall i\in\{0,1,\dots,m\}.
\end{equation}
\end{lemma}

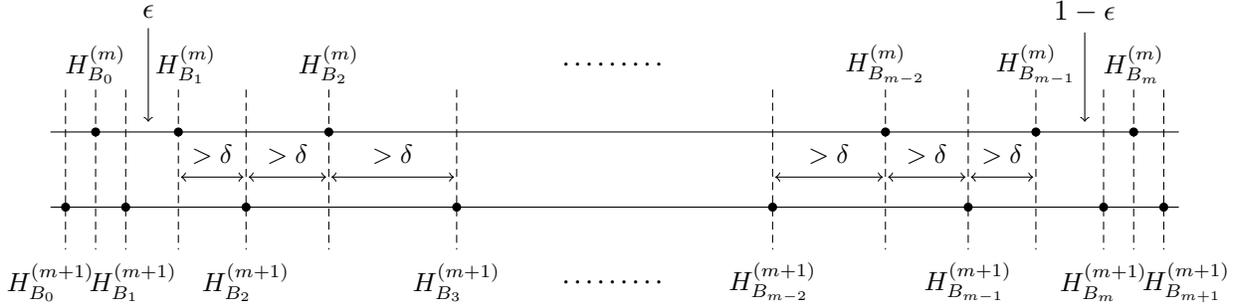
\begin{figure}
\centering
\begin{tikzpicture}
\draw (0,0) -- (15,0);
\draw (0,1) -- (15,1);
\draw [fill] (0.2,0) circle [radius=1.5pt];
\draw [fill] (1,0) circle [radius=1.5pt];
\draw [fill] (2.6,0) circle [radius=1.5pt];
\draw [fill] (5.4,0) circle [radius=1.5pt];
\draw [fill] (9.6,0) circle [radius=1.5pt];
\draw [fill] (12.2,0) circle [radius=1.5pt];

\draw [fill] (14,0) circle [radius=1.5pt];
\draw [fill] (14.8,0) circle [radius=1.5pt];
\draw [fill] (0.6,1) circle [radius=1.5pt];
\draw [fill] (1.7,1) circle [radius=1.5pt];
\draw [fill] (3.7,1) circle [radius=1.5pt];
\draw [fill] (11.1,1) circle [radius=1.5pt];
\draw [fill] (13.1,1) circle [radius=1.5pt];

\draw [fill] (14.4,1) circle [radius=1.5pt];
\draw
node at (7.5,-1) [] {\dots\dots\dots}
node at (7.5,1.9) [] {\dots\dots\dots}
node at (1.3,2.6) [] (e1) {$\epsilon$}
node at (1.3,1) [] (e2) {}
node at (13.75,2.6) [] (e3) {$1-\epsilon$}
node at (13.75,1) [] (e4) {}
node at (0.2,1.7) [] (x1) {}
node at (0,-1) []  {\small $H_{B_0}^{(m+1)}$}
node at (0.6,1.7) [] (x2) {}
node at (0.6,1.9) []  {\small $H_{B_0}^{(m)}$}
node at (1,1.7) [] (x3) {}
node at (1.1,-1) []  {\small $H_{B_1}^{(m+1)}$}

node at (2.15,0.7) [] {\small $>\delta$}
node at (1.7,1.7) [] (x6) {}
node at (1.8,1.9) []  {\small $H_{B_1}^{(m)}$}
node at (2.6,1.7) [] (x7) {}
node at (2.6,-1) []  {\small $H_{B_2}^{(m+1)}$}
node at (3.15,0.7) [] {\small $>\delta$}
node at (3.7,1.7) [] (x8) {}
node at (3.7,1.9) []  {\small $H_{B_2}^{(m)}$}
node at (4.55,0.7) [] {\small $>\delta$}
node at (5.4,1.7) [] (x9) {}
node at (5.4,-1) []  {\small $H_{B_3}^{(m+1)}$}
node at (9.6,1.7) [] (x10) {}
node at (9.6,-1) []  {\small $H_{B_{m-2}}^{(m+1)}$}
node at (10.35,0.7) [] {\small $>\delta$}
node at (11.1,1.7) [] (x11) {}
node at (11.1,1.9) []  {\small $H_{B_{m-2}}^{(m)}$}
node at (11.65,0.7) [] {\small $>\delta$}
node at (12.2,1.7) [] (x12) {}
node at (12.2,-1) []  {\small $H_{B_{m-1}}^{(m+1)}$}
node at (12.65,0.7) [] {\small $>\delta$}
node at (13.1,1.7) [] (x13) {}
node at (13.1,1.9) []  {\small $H_{B_{m-1}}^{(m)}$}

node at (14,1.7) [] (x16) {}
node at (14,-1) []  {\small $H_{B_m}^{(m+1)}$}
node at (14.4,1.7) [] (x17) {}
node at (14.4,1.9) []  {\small $H_{B_m}^{(m)}$}
node at (14.8,1.7) [] (x18) {}
node at (15.1,-1) []  {\small $H_{B_{m+1}}^{(m+1)}$}
node at (0.2,-0.7) [] (y1) {}
node at (0.6,-0.7) [] (y2) {}
node at (1,-0.7) [] (y3) {}

node at (1.6,0.4) [] (m1) {}
node at (2.7,0.4) [] (m2) {}
node at (2.5,0.4) [] (m3) {}
node at (3.8,0.4) [] (m4) {}
node at (3.6,0.4) [] (m5) {}
node at (5.5,0.4) [] (m6) {}

node at (1.7,-0.7) [] (y6) {}
node at (2.6,-0.7) [] (y7) {}
node at (3.7,-0.7) [] (y8) {}
node at (5.4,-0.7) [] (y9) {}
node at (9.6,-0.7) [] (y10) {}
node at (11.1,-0.7) [] (y11) {}
node at (12.2,-0.7) [] (y12) {}
node at (13.1,-0.7) [] (y13) {}

node at (9.5,0.4) [] (m11) {}
node at (11.2,0.4) [] (m12) {}
node at (11,0.4) [] (m13) {}
node at (12.3,0.4) [] (m14) {}
node at (12.1,0.4) [] (m15) {}
node at (13.2,0.4) [] (m16) {}

node at (14,-0.7) [] (y16) {}
node at (14.4,-0.7) [] (y17) {}
node at (14.8,-0.7) [] (y18) {};
\draw [densely dashed] (x1) -- (y1);
\draw [densely dashed] (x2) -- (y2);
\draw [densely dashed] (x3) -- (y3);

\draw [densely dashed] (x6) -- (y6);
\draw [densely dashed] (x7) -- (y7);
\draw [densely dashed] (x8) -- (y8);
\draw [densely dashed] (x9) -- (y9);
\draw [densely dashed] (x10) -- (y10);
\draw [densely dashed] (x11) -- (y11);
\draw [densely dashed] (x12) -- (y12);
\draw [densely dashed] (x13) -- (y13);

\draw [densely dashed] (x16) -- (y16);
\draw [densely dashed] (x17) -- (y17);
\draw [densely dashed] (x18) -- (y18);
\draw [<->] (m1) -- (m2);
\draw [<->] (m3) -- (m4);
\draw [<->] (m5) -- (m6);

\draw [<->] (m11) -- (m12);
\draw [<->] (m13) -- (m14);
\draw [<->] (m15) -- (m16);
\draw [->] (e1) -- (e2);
\draw [->] (e3) -- (e4);
\end{tikzpicture}
\caption{Illustration of the interlacing property in \eqref{eq:jjw1} used in the proofs of Theorem~\ref{thm:m}.}
\label{fig:bvd}
\end{figure}

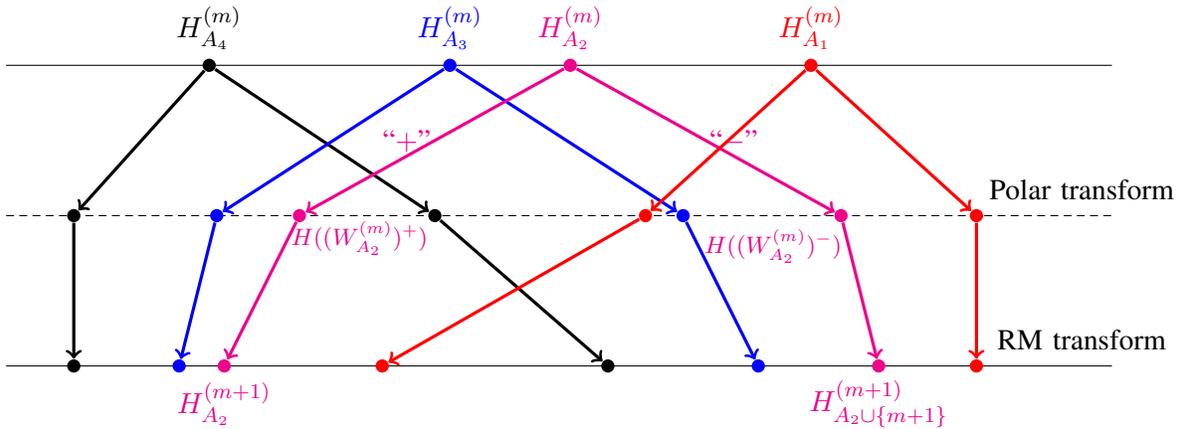
\begin{figure}
\centering
\begin{tikzpicture}[every node/.style={circle,inner sep=0pt, minimum size= 5pt}]
\draw (0.3,0) -- (15,0);
\draw [densely dashed] (0.3, 2) -- (15, 2);
\draw (0.3, 4) -- (15, 4);

\node at (14.6, 2.35) {Polar transform};
\node at (14.6, 0.35) {RM transform};

\node at (3, 4.5) {$H_{A_4}^{(m)}$};
\node at (6.2, 4.5) [color=blue] {$H_{A_3}^{(m)}$};
\node at (7.8, 4.5) [color=magenta] {$H_{A_2}^{(m)}$};
\node at (11, 4.5) [color=red] {$H_{A_1}^{(m)}$};

\node at (3.2, -0.5) [color=magenta] {$H_{A_2}^{(m+1)}$};
\node at (11.9, -0.5) [color=magenta] {$H_{A_2\cup\{m+1\}}^{(m+1)}$};

\node at (5, 1.7) [color=magenta] {\footnotesize $H((W_{A_2}^{(m)})^+)$};
\node at (10.5, 1.6) [color=magenta] {\footnotesize $H((W_{A_2}^{(m)})^-)$};

\node at (3,4) [fill]  (bk0) {};
\node at (6.2,4) [fill= blue]  (m0) {};
\node at (7.8,4) [fill= magenta]  (bl0) {};
\node at (11,4) [fill= red]  (br0) {};

\node at (1.2,2) [fill]  (bk1) {};
\node at (6,2) [fill] (bk2) {};

\node at (3.1,2) [fill, blue] (m1) {};
\node at (9.3,2) [fill, blue] (m2) {};

\node at (4.2,2) [fill, magenta] (bl1) {};
\node at (11.4,2) [fill, magenta] (bl2) {};

\node at (8.8,2) [fill, red] (br1) {};
\node at (13.2,2) [fill, red] (br2) {};

\node at (1.2, 0) [fill] (bk3) {};
\node at (8.3, 0) [fill] (bk4) {};

\node at (2.6,0) [fill, blue] (m3) {};
\node at (10.3,0) [fill, blue] (m4) {};

\node at (3.2,0) [fill, magenta] (bl3) {};
\node at (11.9, 0) [fill, magenta] (bl4) {};

\node at (5.3,0) [fill, red] (br3) {};
\node at (13.2,0) [fill, red] (br4) {};

\draw[very thick,->](bk0) -- node {}(bk1);
\draw[very thick,->](bk0) -- node {}(bk2);
\draw[very thick,->](bk1) -- node {}(bk3);
\draw[very thick,->](bk2) -- node {}(bk4);

\draw[very thick,->, color=blue](m0) -- node {}(m1);
\draw[very thick,->, color=blue](m0) -- node {}(m2);
\draw[very thick,->, color=blue](m1) -- node {}(m3);
\draw[very thick,->, color=blue](m2) -- node {}(m4);

\draw[very thick,->, color=magenta](bl0) -- node [left] {``$+$"}(bl1);
\draw[very thick,->, color=magenta](bl0) -- node [right] {``$-$"} (bl2);
\draw[very thick,->, color=magenta](bl1) -- node {}(bl3);
\draw[very thick,->, color=magenta](bl2) -- node {}(bl4);

\draw[very thick,->, color=red](br0) -- node {}(br1);
\draw[very thick,->, color=red](br0) -- node {}(br2);
\draw[very thick,->, color=red](br1) -- node {}(br3);
\draw[very thick,->, color=red](br2) -- node {}(br4);
\end{tikzpicture}
\caption{The fast polar transform with block size $4$.
The dots on the second line are the results of the standard polar transform, and the dots on the third line are the results of fast polar transform. In the fast polar transform, the bit-channel obtained by adding a monomial gets even worse (i.e., has even more entropy) than what would the classical polar $-$ transform produce on that bit-channel, and similarly,  the better bit-channel obtained by not adding the monomial is even better (less entropic) that what the polar $+$ transform would produce. Therefore, the gap between $H_{A_i\cup\{m+1\}}^{(m+1)}$ and $H_{A_i}^{(m+1)}$ is always larger than the gap between $H((W_{A_i}^{(m)})^-)$ and $H((W_{A_i}^{(m)})^+)$. Intuitively, this explains why RM codes polarize and do so even faster than polar codes (although the formal proof uses the increasing chain argument).}
\label{fig:gp}
\end{figure}

The proof of Theorem \ref{thm:m} relies mainly on the previous lemma and the following polar-like inequality: for any $\epsilon \in (0,1/2)$, there is $\delta(\epsilon)>0$ such that for any increasing chain and any $i\in\{0,1,\dots,m\}$,
\begin{equation} \label{eq:jjw2}
H_{B_i}^{(m)} \in (\epsilon,1-\epsilon)
\implies
 H_{B_i}^{(m)} -  H_{B_i}^{(m+1)} > \delta(\epsilon)
\text{~and~} 
 H_{B_{i+1}}^{(m+1)} - H_{B_i}^{(m)} > \delta(\epsilon).
\end{equation}
Note that \eqref{eq:jjw2} is analogous to \eqref{eq:spread} except that (i) it does not involve the $+,-$ transform of polar codes but the augmentation or not of a monomial with a new element, (ii) the resulting spread is not necessarily symmetrical as in \eqref{sym2}. This is  where it can be seen that RM codes have a bigger spread of polarization than polar codes, as further discussed below. 

We first note that \eqref{thm:m1} follows directly from the interlacing property \eqref{eq:jjw1}; see Fig.~\ref{fig:bvd} for an illustration of this. To prove \eqref{thm:m2}, we  combine \eqref{thm:m1} with the polar-like inequality \eqref{eq:jjw2}.
Indeed, by \eqref{eq:jjw2} we know that as long as $H_{B_i}^{(m)}>\epsilon$ and $H_{B_{i+1}}^{(m)}<1-\epsilon$, we have $H_{B_{i+1}}^{(m)} - H_{B_i}^{(m)}>2\delta$; see Fig.~\ref{fig:bvd} for an illustration.
Let $j$ be the smallest index such that $H_{B_j}^{(m)}>\epsilon$, and let $j'$ be the largest index such that $H_{B_{j'}}^{(m)}<1-\epsilon$. Then
$$
\left|\left\{ i\in\{0,1,\dots,m\}:\epsilon < H_{B_i}^{(m)} < 1-\epsilon \right\} \right| = j'-j+1. 
$$
Since $H_{B_i}^{(m)}$ increases with $i$, we have 
$
H_{B_{j'}}^{(m)}- H_{B_j}^{(m)}=\sum_{i=j}^{j'-1} (H_{B_{i+1}}^{(m)} - H_{B_i}^{(m)}) >2(j'-j)\delta
$, and since $H_{B_{j'}}^{(m)}- H_{B_j}^{(m)}$ is upper bounded by $1$, we have $j'-j<\frac{1}{2\delta}$. Therefore,
$$
\left|\left\{ i\in\{0,1,\dots,m\}:\epsilon < H_{B_i}^{(m)} < 1-\epsilon \right\} \right| < \frac{1}{2\delta}+1. 
$$
Thus we have proved \eqref{thm:m2} with the choice of $D(\epsilon)= 1/(2\delta(\epsilon))+1$.

Now we are left to explain how to prove \eqref{eq:jjw1}--\eqref{eq:jjw2}.
In Fig.~\ref{fig:bj2}, we define the bit-channel $W_{A_i}^{(m)}$ as the binary-input channel that takes $U_{A_i}^{(m)}$ as input and $Y^n,(U_{A_j}^{(m)}:j<i)$ as outputs, i.e., $W_{A_i}^{(m)}$ is the channel seen by the successive RM decoder when decoding $U_{A_i}^{(m)}$. By definition, we have $H_{A_i}^{(m)}=1-I(W_{A_i}^{(m)})$.
Making use of the fact that $\mathrm{RM}(m+1,m+1)$ can be decomposed into two independent copies of $\mathrm{RM}(m,m)$, one can show the larger spread of RM code split. More precisely, for every $A_i\subseteq[m]$, the bit-channel
$W_{A_i}^{(m+1)}$ is always better than the ``$+$" polar transform of $W_{A_i}^{(m)}$, and the bit-channel $W_{A_i\cup\{m+1\}}^{(m+1)}$ is always worse than the ``$-$" polar transform of $W_{A_i}^{(m)}$, i.e.,
$$
H_{A_i}^{(m+1)}
\le H((W_{A_i}^{(m)})^+)
\le  H_{A_i}^{(m)}
\le H((W_{A_i}^{(m)})^-)
\le H_{A_i\cup\{m+1\}}^{(m+1)} .
$$
Therefore, the gap between $H_{A_i\cup\{m+1\}}^{(m+1)}$ and $H_{A_i}^{(m)}$ is even larger than the gap between $H((W_{A_i}^{(m)})^-)$ and $H_{A_i}^{(m)}$. Similarly, the gap between $H_{A_i}^{(m)}$ and $H_{A_i}^{(m+1)}$ is even larger than the gap between $H_{A_i}^{(m)}$ and $H((W_{A_i}^{(m)})^+)$; see Fig.~\ref{fig:gp} for an illustration. Combining this with the polar inequality \eqref{eq:spread}, we have shown that \eqref{eq:jjw1}--\eqref{eq:jjw2} hold for any $B_{i+1}=B_i\cup\{m+1\}$.
Then by the symmetry of RM codes, one can show that $H_{B_i\cup\{j\}}^{(m+1)} \ge H_{B_i\cup\{m+1\}}^{(m+1)}$ for all $j\in[m]\setminus B_i$.
This proves \eqref{eq:jjw1}--\eqref{eq:jjw2} and Theorem \ref{thm:m}.

Theorem \ref{thm:m} proves a polarization on each increasing chain (See equation~\eqref{thm:m2}). In order to obtain the global polarization of RM codes (Theorem~\ref{thm:PRM}), we observe that there are in total $m!$ increasing chains, and a careful averaging argument over these $m!$ chains gives the result in Theorem~\ref{thm:PRM}.

\section{Decoding algorithms} \label{sect:drm}
We will survey various decoding algorithms for RM codes in this section. 
We divide these algorithms into three categories. The first category (Section~\ref{sect:Reed}) only consists of Reed's algorithm \cite{Reed54}: This is the first decoding algorithm for RM codes, designed for the worst-case error correction, and it can efficiently correct any error pattern with Hamming weight up to half the code distance.
The second category (Section~\ref{sect:ptre}) includes efficient algorithms designed for correcting random errors or additive Gaussian noise. These algorithms afford good practical performance in the short to medium code length regime or for low-rate RM codes. Yet due to complexity constraints, most of them are not efficient for decoding RM codes with long code length.
More specifically, we will cover the Fast Hadamard Transform decoder \cite{Green66,Be86} for first-order RM codes, Sidel'nikov-Pershakov algorithm and its variants \cite{Sidel92,Sakkour05}, Dumer's list decoding algorithm \cite{Dumer04,Dumer06,Dumer06a} and Recursive Projection-Aggregation algorithm \cite{YA18} as well as an algorithm based on minimum-weight parity checks \cite{Santi18}.
Finally, the last category (Section~\ref{sec:SSV}) again only consists of a single decoding algorithm---a Berlekamp-Welch type decoding algorithm proposed in \cite{Saptharishi17}.
This algorithm is designed for correcting random errors. Its performance guarantee (i.e., its polynomial run-time estimate) was established for decoding RM codes of degrees up to $r=o(\sqrt{m})$ while all the previous decoding algorithms discussed in this section only have performance guarantee for constant value of $r$ (i.e., we do not have polynomial upper bounds on their run time at other regimes of parameters).
In fact, this algorithm also gives interesting results for degrees $r=m-o(\sqrt{m/\log m})$.

\subsection{Reed's algorithm \cite{Reed54}: Unique decoding up to half the code distance}
\label{sect:Reed}
In this section, we recap
Reed's decoding algorithm \cite{Reed54} for $\mathrm{RM}(m,r)$. It can correct any error pattern with Hamming weight less than $2^{m-r-1}$, half the code distance.

For a subset $A\subseteq[m]$, we write $\overline{A}=[m]\setminus A$ and we use $V_A:=\{\original{z}\in\bF_2^m: z_i=0~\forall i\in\overline{A}\}$ to denote the $|A|$-dimensional subspace of $\bF_2^m$, i.e., $V_A$ is the subspace obtained by fixing all $z_i$'s to be $0$ for $i$ outside of $A$.
For a subspace $V_A$ in $\bF_2^m$, there are $2^{m-|A|}$ cosets of the form $V_A+\original{b}:=\{\original{z}+\original{b}:\original{z}\in V_A\}$, where $\original{b}\in\bF_2^m$.
For any $A\subseteq[m]$ and any $\original{b}\in\bF_2^m$, we always have
\begin{equation}\label{eq:tt1}
\sum_{\original{z}\in (V_A+\original{b})}
\Eval_{\original{z}}(x_A)  =1,
\end{equation}
and we also have that for any $A\not\subset B$,
\begin{equation}\label{eq:tt2}
\sum_{\original{z}\in (V_A+\original{b})}
\Eval_{\original{z}}(x_B)  =0.
\end{equation}
The sums in \eqref{eq:tt1}--\eqref{eq:tt2} are both over $\bF_2$.
To see \eqref{eq:tt1}, notice that  $\Eval_{\original{z}}(x_A)=1$ if and only if $z_i=1$ for all $i\in A$, and there is only one such $\original{z}\in (V_A+\original{b})$. 
To see \eqref{eq:tt2}: Since $A\not\subseteq B$, there is $i\in (A\setminus B)$. The value of $z_i$ does not affect the evaluation $\Eval_{\original{z}}(x_B)$. Therefore,
$
\sum_{\original{z}\in (V_A+\original{b}), z_i=0}
\Eval_{\original{z}}(x_B)
= \sum_{\original{z}\in (V_A+\original{b}), z_i=1}
\Eval_{\original{z}}(x_B),
$
and \eqref{eq:tt2} follows immediately.

Suppose that the binary vector $y=(y_{\original{z}}:\original{z}\in\bF_2^m)$ is a noisy version of a codeword $\Eval(f)\in\mathrm{RM}(m,r)$ such that $y$ and $\Eval(f)$ differ in less than $2^{m-r-1}$ coordinates. Reed's algorithm recovers the original codeword from $y$ by decoding the coefficients of the polynomial $f$.
Since $\deg(f)\le r$, we can always write $f=\sum_{A\subseteq[m], |A|\le r} u_A x_A$, where $u_A$'s are the coefficients of the corresponding monomials.
Reed's algorithm first decodes  the coefficients of all the degree-$r$ monomials, and then it decodes the coefficients of all the degree-$(r-1)$ monomials, so on and so forth, until it decodes all the coefficients.

To decode the coefficients $u_A$ for $|A|=r$, Reed's algorithm first calculates the sums
$\sum_{\original{z}\in (V_A+\original{b})} y_{\original{z}}$ over each of the $2^{m-r}$ cosets of the subspace $V_A$, and then it performs a majority vote among these $2^{m-r}$ sums: If there are more $1$'s than $0$'s, then we decode $u_A$ as $1$. Otherwise we decode it as $0$.
Notice that if there is no error, i.e., if $y=\Eval(f)$, then we have
$$
\sum_{\original{z}\in (V_A+\original{b})} y_{\original{z}}
= \sum_{\original{z}\in (V_A+\original{b})}
\Eval_{\original{z}}(\sum_{B\subseteq[m], |B|\le r} u_B x_B)
= \sum_{B\subseteq[m], |B|\le r} u_B
\sum_{\original{z}\in (V_A+\original{b})}
\Eval_{\original{z}}(x_B).
$$
According to \eqref{eq:tt1}--\eqref{eq:tt2}, for the subsets $B\subseteq[m]$ with $|B|\le r=|A|$, $\sum_{\original{z}\in (V_A+\original{b})}
\Eval_{\original{z}}(x_B)=1$ if and only if $B=A$.
Therefore, $\sum_{\original{z}\in (V_A+\original{b})} y_{\original{z}}
=u_A$ for all the $2^{m-r}$ cosets of the form $V_A+\original{b}$ if $y=\Eval(f)$.
Since we assume that $y$ and $\Eval(f)$ differ in less than $2^{m-r-1}$ coordinates, there are less than $2^{m-r-1}$ cosets for which $\sum_{\original{z}\in (V_A+\original{b})} y_{\original{z}}
\neq u_A$. After the majority voting among these $2^{m-r}$ sums, we will obtain the correct value of $u_A$.

After decoding all the coefficients of the degree-$r$ monomials, we can calculate 
$$
y'=y-\Eval(\sum_{B\subseteq[m], |B|= r} u_B x_B).
$$
This is a noisy version of the codeword
$\Eval(f-\sum_{B\subseteq[m], |B|= r} u_B x_B) \in\mathrm{RM}(m,r-1)$, and the number of errors in $y'$ is less than $2^{m-r-1}$ by assumption. Now we can use the same method to decode the coefficients of all the degree-$(r-1)$ monomials from $y'$. We then repeat this procedure until we decode all the coefficients of $f$.

\begin{theorem}
For a fixed $r$ and growing $m$, Reed's algorithm corrects any error pattern with Hamming weight less than $2^{m-r-1}$ in $O(n\log^r n)$ time when decoding $\mathrm{RM}(m,r)$.
\end{theorem}

Reed's algorithm is summarized below:

\begin{center}
\begin{algorithm}
\caption{Reed's algorithm for decoding $\mathrm{RM}(m,r)$}
{\bf Input}: Parameters $m$ and $r$ of the RM code, and a binary vector $y=(y_{\original{z}}:\original{z}\in\bF_2^m)$ of length $n=2^m$

{\bf Output}: A codeword $c\in\mathrm{RM}(m,r)$

\vspace*{0.05in}
\begin{algorithmic}[1]
\State{$t\gets r$}

\While{$t\ge 0$}

\For{each subset $A\subseteq[m]$ with $|A|=t$}

\State{Calculate $\sum_{\original{z}\in (V_A+\original{b})} y_{\original{z}}$ for all the $2^{m-t}$ cosets of $V_A$}

\State{$\texttt{num1}\gets$ number of cosets $(V_A+\original{b})$ such that
$\sum_{\original{z}\in (V_A+\original{b})} y_{\original{z}}=1$}

\State{$u_A \gets \mathbf{1}[\texttt{num1} \ge 2^{m-t-1}]$}

\EndFor

\State{$y\gets y-\Eval(\sum_{A\subseteq[m],|A|=t}u_A x_A)$}

\State{$t\gets t-1$}

\EndWhile

\State{$c\gets \Eval(\sum_{A\subseteq[m],|A|\le r}u_A x_A)$}

\State{Output $c$}
\end{algorithmic}
\end{algorithm}
\end{center}

\subsection{Practical algorithms for short to medium length RM codes} \label{sect:ptre}

\subsubsection{Fast Hadamard Transform (FHT) for first order RM codes \cite{Green66,Be86}} \label{sect:fht}

The dimension of the first order RM code $\mathrm{RM}(m,1)$ is $m+1$, so there are in total $2^{m+1}=2n$ codewords.
A naive implementation of the Maximum Likelihood (ML) decoder requires $O(n^2)$ operations. In this section we recap an efficient implementation of the ML decoder based on FHT which requires only $O(n\log n)$ operations. We will focus on the soft-decision version of this algorithm, and the hard-decision version can be viewed as a special case.

Consider a binary-input memoryless channel $W:\{0,1\}\to\cW$. The log-likelihood ratio (LLR) of an output symbol $x\in\cW$ is defined as
$$
\LLR(x):=\ln\Big(\frac{W(x|0)}{W(x|1)}\Big).
$$
We still use $y=(y_{\original{z}}: \original{z}\in \bF_2^m)$ to denote the noisy version of a codeword in $\mathrm{RM}(m,1)$. 
Given the channel output vector $y$, the ML decoder for first order RM codes aims to find $c\in\mathrm{RM}(m,1)$ to maximize
$
\prod_{\original{z}\in\bF_2^m} W(y_{\original{z}}| c_{\original{z}}).
$
This is equivalent to finding $c$ which maximizes the following quantity:
$$
\prod_{\original{z}\in\bF_2^m} \frac{W(y_{\original{z}}| c_{\original{z}})}
{\sqrt{W(y_{\original{z}}|0) W(y_{\original{z}}|1)}},
$$
which is further equivalent to maximizing
\begin{equation}\label{eq:llr}
\sum_{\original{z}\in\bF_2^m} \ln \Big(
\frac{W(y_{\original{z}}| c_{\original{z}})}
{\sqrt{W(y_{\original{z}}|0) W(y_{\original{z}}|1)}} \Big).
\end{equation}
As the codeword $c$ is a binary vector,
$$
\ln \Big( \frac{W(y_{\original{z}}| c_{\original{z}})}
{\sqrt{W(y_{\original{z}}|0) W(y_{\original{z}}|1)}} \Big)
=\left\{ \begin{array}{cc}
\frac{1}{2}\LLR(y_{\original{z}}) & \mbox{if~} c_{\original{z}}=0 \vspace*{0.05in} \\ 
-\frac{1}{2}\LLR(y_{\original{z}}) & \mbox{if~} c_{\original{z}}=1
\end{array} \right.  .
$$
From now on we will use the shorthand notation
$$
L_{\original{z}}:= \LLR(y_{\original{z}}),
$$
and the formula in \eqref{eq:llr} can be written as
\begin{equation} \label{eq:mxr}
\frac{1}{2} \sum_{\original{z}\in\bF_2^m} \Big( (-1)^{c_{\original{z}}} L_{\original{z}} \Big),
\end{equation}
so we want to find $c\in\mathrm{RM}(m,1)$ to maximize this quantity.

By definition, every $c\in\mathrm{RM}(m,1)$ corresponds to a polynomial in $\mathbb{F}_2[x_1,x_2,\dots,x_m]$ of degree one, so we can write every codeword $c$ as a polynomial $u_0+\sum_{i=1}^m u_i x_i$.
In this way, we have $c_{\original{z}}=u_0+\sum_{i=1}^m u_i z_i$, where $z_1,z_2,\dots,z_m$ are the coordinates of the vector $\original{z}$.
Now our task is to find $u_0,u_1,u_2,\dots,u_m\in \mathbb{F}_2$ to maximize
\begin{equation} \label{eq:fht}
\sum_{\original{z}\in\bF_2^m} \Big( (-1)^{u_0+\sum_{i=1}^m u_i z_i} L_{\original{z}} \Big)
=(-1)^{u_0} \sum_{\original{z}\in\bF_2^m} \Big( (-1)^{\sum_{i=1}^m u_i z_i} L_{\original{z}} \Big).
\end{equation}
For a binary vector $\original{u}=(u_1,u_2,\dots,u_m)\in\bF_2^m$, we define
$$
\hat{L}(\original{u}):= \sum_{\original{z}\in\bF_2^m} \Big( (-1)^{\sum_{i=1}^m u_i z_i} L_{\original{z}} \Big).
$$
To find the maximizer of \eqref{eq:fht}, we only need to compute $\hat{L}(\original{u})$ for all $\original{u}\in\bF_2^m$, but the vector $(\hat{L}(\original{u}):\original{u}\in\bF_2^m)$ is exactly the Hadamard Transform of the vector $(L_{\original{z}}:\original{z}\in\bF_2^m)$, so it can be computed using the Fast Hadamard Transform with complexity $O(n\log n)$.
Once we know the values of $(\hat{L}(\original{u}),\original{u}\in\bF_2^m)$, we can find $\original{u}^*=(u_1^*,u_2^*,\dots,u_m^*)\in\bF_2^m$ that maximizes $|\hat{L}(\original{u})|$.
If $\hat{L}(\original{u}^*)>0$, then the decoder outputs the codeword corresponding to $u_0^*=0,u_1^*,u_2^*,\dots,u_m^*$. Otherwise, the decoder outputs the codeword corresponding to $u_0^*=1,u_1^*,u_2^*,\dots,u_m^*$.
This completes the description of the soft-decision FHT decoder for first order RM codes.

The hard-decision FHT decoder is usually used for random errors, or equivalently, used for error corrections over BSC. For BSC, the channel output $y=(y_{\original{z}}: \original{z}\in \bF_2^m)$ is a binary vector.
Suppose that the crossover probability of BSC is $p<1/2$, then $L_{\original{z}}=\ln(\frac{1-p}{p})$ if $y_{\original{z}}=0$, and $L_{\original{z}}=-\ln(\frac{1-p}{p})$ if $y_{\original{z}}=1$.
Since rescaling the LLR vector by a positive factor does not change the maximizer of \eqref{eq:fht}, we can divide the LLR vector by $\ln(\frac{1-p}{p})$ when decoding RM codes over BSC$(p)$.
This is equivalent to setting $L_{\original{z}}=1$ for $y_{\original{z}}=0$ and $L_{\original{z}}=-1$ for $y_{\original{z}}=1$.
Then the rest of the hard-decision FHT decoding is the same as the soft-decision version.

\begin{theorem}
The FHT decoder finds the ML decoding result in $O(n\log n)$ time when decoding first order RM codes.
\end{theorem}

FHT can also be used for list decoding of first order RM codes. For list decoding with list size $s$, we find $s$ vectors $\original{u}^{(1)},\dots,\original{u}^{(s)}$ that give the largest values of $|\hat{L}(\original{u})|$ among all vectors in $\bF_2^m$.

As a final remark, we mention that
first order RM codes can also be decoded efficiently as geometry codes \cite{Tallini11}.

\begin{center}
\begin{algorithm}
\caption{FHT decoder for first order RM codes}
{\bf Input}: Code length $n=2^m$, and the LLR vector $(L_{\original{z}}:\original{z}\in\bF_2^m)$ of the received (noisy) codeword

{\bf Output}: A codeword $c\in\mathrm{RM}(m,1)$

\vspace*{0.05in}
\begin{algorithmic}[1]
\State{$(\hat{L}(\original{u}):\original{u}\in\bF_2^m)
\gets \FHT(L_{\original{z}}:\original{z}\in\bF_2^m)$}

\State{$\original{u}^*=(u_1^*,u_2^*,\dots,u_m^*) \gets \argmax_{\original{u}\in\bF_2^m} |\hat{L}(\original{u})|$}

\If{$\hat{L}(\original{u}^*)>0$}

\State{
$c\gets \Eval(\sum_{i=1}^m u_i^* x_i)$}

\Else

\State{
$c\gets \Eval(1+\sum_{i=1}^m u_i^* x_i)$}

\EndIf

\State{Output $c$}
\end{algorithmic}
\end{algorithm}
\end{center}

\subsubsection{Sidel'nikov-Pershakov algorithm \cite{Sidel92}
and its variant \cite{Sakkour05}}
\label{sect:sid}

In \cite{Sidel92}, Sidel'nikov and Pershakov proposed a decoding algorithm that works well for second order RM codes with short or medium code length (e.g. $\le 1024$). A version of their decoding algorithm also works for higher-order RM codes, but the performance is not as good as the one for second order RM codes.

In this section, we recap Sidel'nikov-Pershakov algorithm for second order RM codes.
Consider a polynomial $f\in\bF_2[x_1,\dots,x_m]$ with  $\deg(f)\le 2$:
$$
f(z_1,\dots,z_m)=\sum_{1\le i<j\le m}u_{i,j}z_iz_j
+\sum_{i=1}^m u_iz_i + u_0.
$$
For a vector $\original{b}=(b_1,\dots,b_m)\in\bF_2^m$, we have
\begin{align}
\Eval_{\original{z}+\original{b}}(f)
+\Eval_{\original{z}}(f)
=& \sum_{1\le i<j\le m}u_{i,j}(z_i+b_i) (z_j+b_j)
+ \sum_{1\le i<j\le m}u_{i,j}z_iz_j
+ \sum_{i=1}^m u_i b_i  \nonumber \\
=& \sum_{1\le i<j\le m} u_{i,j} z_i b_j
+ \sum_{1\le i<j\le m} u_{i,j} b_i z_j
+ \sum_{1\le i<j\le m} u_{i,j} b_i b_j
+ \sum_{i=1}^m u_i b_i \nonumber \\
=& \original{b} U \original{z}^T + \sum_{1\le i<j\le m} u_{i,j} b_i b_j + \sum_{i=1}^m u_i b_i ,
\label{eq:bs}
\end{align}
where the matrix $U$ is defined as
$$
U:=\left[ \begin{array}{ccccc}
0 & u_{1,2} & u_{1,3} & \dots & u_{1,m}  \\
u_{1,2} & 0 & u_{2,3} & \dots & u_{2,m} \\
u_{1,3} & u_{2,3} & 0 & \dots & u_{3,m} \\
\vdots & \vdots & \vdots & \vdots & \vdots \\
u_{1,m} & u_{2,m} & u_{3,m} & \dots & 0
\end{array}\right].
$$
Note that $\Eval_{\original{z}+\original{b}}(f)
+\Eval_{\original{z}}(f)$ is the coordinate of the discrete derivative of $f$ at direction $\original{b}$, as defined in \eqref{eq:disder}.

We first describe the decoder for BSC and then generalize it to other binary-input channels. Suppose that we transmit the codeword $c=\Eval(f) \in\mathrm{RM}(m,2)$ through some BSC, and we denote the channel output vector as $y\in\bF_2^n$.
For a fixed $\original{b}$, the vector $(y_{\original{z}+\original{b}}
+y_{\original{z}}: \original{z}\in\bF_2^m)$ is the noisy version of the codeword in $\mathrm{RM}(m,1)$ corresponding to the polynomial in \eqref{eq:bs}. 
Note that the vector $\original{b}U$ consist of the coefficients of all the degree-$1$ monomials in this polynomial.
Therefore, we can decode $\original{b}U$ from the noisy codeword $(y_{\original{z}+\original{b}}
+y_{\original{z}}: \original{z}\in\bF_2^m)$.
A naive way to do so is to decode each $\original{b}U$ separately using the FHT decoder for different vectors $\original{b}\in\bF_2^m$.
Sidel'nikov and Pershakov instead proposed to decode $\original{b}U$ for all $\original{b}\in\bF_2^m$ collectively:
The first step is to calculate $s$ candidates for $\original{b}U$ that have the largest posterior probability by decoding $(y_{\original{z}+\original{b}}
+y_{\original{z}}: \original{z}\in\bF_2^m)$ with the FHT list decoder described at the end of Section~\ref{sect:fht}. We denote these $s$ candidates as $D_{\original{b}}^{(1)},\dots,D_{\original{b}}^{(s)}$ and associate each of them with a reliability value initialized as its posterior probability.
Since $\original{b}U=\original{b}'U+(\original{b}+\original{b}')U$ for all $\original{b}'\in\bF_2^m$, the correct candidates $(D_{\original{b}}^*:\original{b}\in\bF_2^m)$  satisfy
$D_{\original{b}}^*=D_{\original{b}'}^*
+D_{\original{b}+\original{b}'}^*$.
In order to find $D_{\original{b}}^*$,
for each $i=1,\dots,s$, we check for all $\original{b}'\in\bF_2^m$ whether there are certain $i_1$ and $i_2$ such that 
$D_{\original{b}}^{(i)}=
D_{\original{b}'}^{(i_1)} +
D_{\original{b}+\original{b}'}^{(i_2)}$.
Each time when we find such $\original{b}'$ and $i_1,i_2$, we increase the reliability value of $D_{\original{b}}^{(i)}$ by some function\footnote{The choice of this function is somewhat ad hoc, and we omit the precise definition here.} of the reliability values of $D_{\original{b}'}^{(i_1)}$ and $D_{\original{b}+\original{b}'}^{(i_2)}$.
Finally, we set $D_{\original{b}}^*$ to be
$D_{\original{b}}^{(i)}$ with the largest reliability value among all $i=1,\dots,s$.

At this point, we have obtained 
$D_{\original{b}}^*$ for all $\original{b}\in\bF_2^m$.
 Notice that $D_{\original{b}}^*$ is the noisy  version of $\original{b}U$, and in particular, the first coordinate of $D_{\original{b}}^*$ is the noisy version of $u_{1,2}b_2+u_{1,3}b_3+\dots+u_{1,m}b_m$. Therefore, if we pick the first coordinate of $D_{\original{b}}^*$ for all $\original{b}\in\bF_2^m$, we will obtain the noisy version of a codeword from $\mathrm{RM}(m,1)$, and this codeword is the evaluation vector of the polynomial $u_{1,2}x_2+u_{1,3}x_3+\dots+u_{1,m}x_m$. After decoding this noisy codeword using the FHT decoder, we will obtain the coefficients $u_{1,2}, u_{1,3},\dots,u_{1,m}$, which form the first column of the matrix $U$.
 Similarly, we can also pick the $i$th coordinate of $D_{\original{b}}^*$ for all $\original{b}\in\bF_2^m$ and decode it with the FHT decoder. This will allow us to calculate the $i$th column of $U$.
 Once we decode all the entries in $U$, we have  the coefficients of all the degree-2 monomials in $f$. Then we use the FHT decoder again to decode all the other coefficients in $f$, which gives us the final decoding result.

For more general binary-input channels other than BSC, we are not able to calculate $y_{\original{z}+\original{b}}
+y_{\original{z}}$ since the two summands are not binary any more. We instead work with the LLRs
$L_{\original{z}}:= \LLR(y_{\original{z}})$.
Given $L_{\original{z}+\original{b}}$
and $L_{\original{z}}$, we want to estimate how likely $\Eval_{\original{z}+\original{b}}(f)
+\Eval_{\original{z}}(f)$ is $0$ or $1$.
The LLR of the sum
$\Eval_{\original{z}+\original{b}}(f)
+\Eval_{\original{z}}(f)$ can be calculated as
\begin{equation}\label{eq:tsi}
\ln \big( \exp \big( L_{\original{z}+\original{b}}
+L_{\original{z}} \big) +1 \big) - 
\ln \big(  \exp(L_{\original{z}})
+ \exp(L_{\original{z}+\original{b}}) \big).
\end{equation}
Once we replace $y_{\original{z}+\original{b}}
+y_{\original{z}}$ with this LLR, we can follow the decoding procedure described above for BSC to decode the output vector of more general binary-input channels.

In \cite{Sidel92}, Sidel'nikov and Pershakov showed that for second order RM codes $\mathrm{RM}(m,2)$, their algorithm can correct almost all error patterns with Hamming weight no more than
$(n-Cm^{1/4}n^{3/4})/2$ for any constant $C>\ln 4$ when the code length $n\to\infty$.

In \cite{Sakkour05}, Sakkour proposed a simplified and improved version of the Sidel'nikov-Pershakov algorithm for decoding second order RM codes. The main change in Sakkour's algorithm is to use a simple majority voting to obtain $D_{\original{b}}^*$ from $D_{\original{b}}$, replacing the more complicated procedure in Sidel'nikov-Pershakov algorithm. Such a simplification also leads to smaller decoding error probability. We summarize Sakkour's algorithm below in Algorithm~\ref{ag:sak} since it is simpler and has better performance:

\begin{center}
\begin{algorithm}
\caption{Sakkour's algorithm for decoding second order RM codes over BSC}
\label{ag:sak}

{\bf Input}: The code length $n=2^m$; the received (noisy) codeword $y=(y_{\original{z}}: \original{z}\in\bF_2^m)$;

{\bf Output}: 
A codeword $c\in\mathrm{RM}(m,2)$

\vspace*{0.05in}
\begin{algorithmic}[1]
\For{every $\original{b}\in\bF_2^m$}

\State{$D_{\original{b}} \gets \texttt{FHT\_Decoder}(y_{\original{z}+\original{b}}
+y_{\original{z}}: \original{z}\in\bF_2^m)$ 
\Comment{\texttt{FHT\_Decoder} for $\mathrm{RM}(m,1)$}}

\EndFor

\For{every $\original{b}\in\bF_2^m$}

\State{$D_{\original{b}}^* \gets \Majority
(D_{\original{b}+\original{b}'}+D_{\original{b}'}: 
\original{b}'\in\bF_2^m)$}

\State{\Comment{Majority function picks the vector occurring the largest number of times}}

\EndFor

\For{$i=1,2,\dots,m$}

\State{$E_i \gets (\text{the $i$-th coordiante of } D_{\original{b}}^*: \original{b}\in\bF_2^m)$ 
\Comment{$E_i$ is a vector of length $n$}}

\State{$\hat{E}_i \gets \texttt{FHT\_Decoder}(E_i)$ 
\Comment{\texttt{FHT\_Decoder} for $\mathrm{RM}(m,1)$}}

\State{$(u_{\{i,j\}}:j\in[m]\setminus\{i\})
\gets$ coefficients of the polynomial corresponding to $\hat{E}_i$}

\EndFor

\vspace*{0.04in}
\State{$y\gets y-\Eval(\sum_{A\subseteq[m],|A|=2}u_Ax_A)$}

\State{$\hat{D} \gets \texttt{FHT\_Decoder}(y)$ 
\Comment{\texttt{FHT\_Decoder} for $\mathrm{RM}(m,1)$}}

\State{$(u_A:A\subseteq[m],|A|\le 1)
\gets$ coefficients of the polynomial corresponding to $\hat{D}$}

\State{$c\gets \Eval(\sum_{A\subseteq[m],|A|\le 2}u_Ax_A)$}

\State{Output $c$}
\end{algorithmic}
\end{algorithm}
\end{center}

\subsubsection{Dumer's recursive list decoding \cite{Dumer04,Dumer06,Dumer06a}}
\label{sect:Dumer}

Dumer's recursive list decoding makes use of the Plotkin $(u,u+v)$ construction of RM codes (see Section~\ref{sect:rscd} for discussions), and it works well for short or medium length RM codes. More precisely, for RM codes with length $\le 256$, Dumer's recursive list decoding algorithm can efficiently approach the decoding error probability of the ML decoder.
For code length $512$ or $1024$, Dumer's algorithm works well for RM codes with  low code rates.
In \cite{Dumer06a}, Dumer and Shabunov also proposed to construct subcodes of RM codes that have better performance than RM codes themselves under the recursive list decoding algorithm.

We start with the basic version of Dumer's recursive decoding algorithm (without list decoding).
Suppose that we transmit a codeword $c=\Eval(f)\in\mathrm{RM}(m,r)$ through some binary-input channel $W$, and we denote the output vector as $y$ and the corresponding LLR vector as $L$.
The original codeword $\Eval(f)$ has two components $\Eval^{[z_m=0]}(f)$ and $\Eval^{[z_m=1]}(f)$.
We also devide the LLR vector $L$ into two subvectors $L^{[z_m=0]}$ and $L^{[z_m=1]}$ in the same way so that $L^{[z_m=0]}$ and $L^{[z_m=1]}$ are the LLR vectors of $\Eval^{[z_m=0]}(f)$ and $\Eval^{[z_m=1]}(f)$, respectively.
We then construct the LLR vector of $\Eval^{[/z_m]}(f)$  from $L^{[z_m=0]}$ and $L^{[z_m=1]}$ and we denote it as $L^{[/z_m]}$: Each coordinate  of $L^{[/z_m]}$ is obtained by combining the corresponding coordinates in $L^{[z_m=0]}$ and $L^{[z_m=1]}$ using formula \eqref{eq:tsi}.\footnote{In \cite{Dumer06a}, Dumer and Shabunov proposed to work with the quantity $W(x|1)-W(x|0)$ instead of LLR, but one can show that formula \eqref{eq:tsi} for LLR is equivalent to the combining method in \cite{Dumer06a} expressed in terms of $W(x|1)-W(x|0)$.}

We first decode $\Eval^{[/z_m]}(f)\in \mathrm{RM}(m-1,r-1)$ from $L^{[/z_m]}$ and denote the decoding result as $\hat{c}^{[/z_m]}$. Then we use $\hat{c}^{[/z_m]}$ together with $L^{[z_m=0]}$ and $L^{[z_m=1]}$ to form an updated LLR vector of $\Eval^{[z_m=0]}(f)$, which we denote as $\tilde{L}^{[z_m=0]}$. 
The updating rule is as follows:
For each $\original{z}=(z_1,\dots,z_{m-1})\in\bF_2^{m-1}$, if $\hat{c}_{\original{z}}^{[/z_m]}=0$, then we set $\tilde{L}_{\original{z}}^{[z_m=0]}=L_{\original{z}}^{[z_m=0]}+L_{\original{z}}^{[z_m=1]}$, and if $\hat{c}_{\original{z}}^{[/z_m]}=1$, then we set $\tilde{L}_{\original{z}}^{[z_m=0]}=L_{\original{z}}^{[z_m=0]}-L_{\original{z}}^{[z_m=1]}$.
As the next step, we decode $\Eval^{[z_m=0]}(f) \in \mathrm{RM}(m-1,r)$ from $\tilde{L}^{[z_m=0]}$ and denote the decoding result as $\hat{c}^{[z_m=0]}$.
Finally, we combine $\hat{c}^{[/z_m]}$ and $\hat{c}^{[z_m=0]}$ to form the final decoding result $\hat{c}=(\hat{c}^{[z_m=0]},\hat{c}^{[z_m=0]}+\hat{c}^{[/z_m]})$.

In this recursive decoding algorithm, we decompose the decoding of  $\mathrm{RM}(m,r)$ into two tasks: First, we decode a codeword from $\mathrm{RM}(m-1,r-1)$. After that, we decode another codeword from $\mathrm{RM}(m-1,r)$. Then the decoding of $\mathrm{RM}(m-1,r-1)$ and $\mathrm{RM}(m-1,r)$ are further decomposed into decoding another four codewords from RM codes with shorter code length and smaller order.
This decomposition procedure continues until we reach codewords from first order RM codes $\mathrm{RM}(i,1)$ for some $i$ or full RM codes $\mathrm{RM}(j,j)$ for some $j$.
For first order RM codes we use the FHT decoder, and for full RM codes we simply use the ML decoder.
The summary of the algorithm and an illustration of how it works for $\mathrm{RM}(6,2)$ and $\mathrm{RM}(6,3)$ are given in 
 Fig.~\ref{fig:bb}.


Next we briefly discuss the recursive list decoding algorithm. In the list decoding version, we usually stop at the zero order RM codes instead of the first order ones \footnote{In fact, stopping at first order RM codes in list decoding allows one to achieve smaller decoding error probability than stopping at zero order RM codes. In this paper, we only describe the version of list decoding stopping at zero order RM codes for two reasons: First, it is easier to describe; second, this is the version presented in Dumer and Shabunov's original paper \cite[Section~III]{Dumer06a}.}. Note that the zero order RM codes are simply repetition codes with dimension $1$.
We still go through the same procedure as illustrated in Fig.~\ref{fig:bb}, i.e., we keep decomposing the RM codes and eventually we only decode the RM codes on ``leaf nodes". In the list decoding algorithm, the codes on ``leaf nodes" are either repetition codes or full RM codes.
Each time when we decode a new leaf node, we examine several possible decoding results of this new node  for every candidate in the list: If this new leaf node is a repetition code, then there are only two possible decoding results--all zero or all one; if this new leaf node is a full RM code, then we take the $4$ most likely decoding results of it for every candidate in the list.
In this way, we will increase the list size by a factor of $2$ or $4$ at each step, depending on whether the new leaf node is a repetition code or a full RM code. We then calculate a reliability value for each candidate in the new list. When the list size is larger than some pre-specified value $\mu$, we prune the list down to size $\mu$ by only keeping the candidates with the largest reliability values.
Clearly, large $\mu$ leads to longer running time of the algorithm but smaller decoding error probability.

In \cite{Dumer06a}, a family of subcodes of RM codes were also proposed. The subcodes have smaller decoding probability under the recursive list decoding algorithm. The idea is quite natural: Each repetition code on the ``leaf nodes" only contain one information bit.
Some of these information bits are relatively noisy, and the others are relatively noiseless.
Dumer and Shabunov proposed to  set all the noisy bits to be $0$. In this way, one can get smaller decoding error at the cost of decreasing the code rate.


\begin{center}
\begin{algorithm}
\caption{Dumer's Algorithm $\Phi_r^m$ for decoding $\mathrm{RM}(m,r)$}
{\bf Input}: LLR vector $L=(L^{[z_m=0]},L^{[z_m=1]})$

{\bf Output}: $\hat{c}$

\begin{algorithmic}[1]
\If{$1<r<m$}

\State{
Calculate $L^{[/z_m]}$ from $L^{[z_m=0]}$ and $L^{[z_m=1]}$}

\State{
$\hat{c}^{[/z_m]}  \leftarrow
\Phi_{r-1}^{m-1}(L^{[/z_m]})$}

\State{
Calculate $\tilde{L}^{[z_m=0]}$ from $L^{[z_m=0]},L^{[z_m=1]}$ and $\hat{c}^{[/z_m]}$}

\State{
$\hat{c}^{[z_m=0]}  \leftarrow
\Phi_{r}^{m-1}(\tilde{L}^{[z_m=0]})$}

\State{
$\hat{c} \leftarrow
(\hat{c}^{[z_m=0]},\hat{c}^{[z_m=0]}+\hat{c}^{[/z_m]})$}

\ElsIf{$r=1$} 
\State{use FHT decoder}

\ElsIf{$r=m$} 
\State{use ML decoder}
\EndIf
\end{algorithmic}
\end{algorithm}
\end{center}

\begin{figure}
\centering

\begin{tikzpicture}
\node at (2,8) (x1) {$\mathrm{RM}(6,2)$};
\node at (1,7) (x2) {\textcolor{red}{$\mathrm{RM}(5,1)$}};
\node at (3,7) (x3) {$\mathrm{RM}(5,2)$};
\node at (2,6) (x4) {\textcolor{red}{$\mathrm{RM}(4,1)$}};
\node at (4,6) (x5) {$\mathrm{RM}(4,2)$};
\node at (3,5) (x6) {\textcolor{red}{$\mathrm{RM}(3,1)$}};
\node at (5,5) (x7) {$\mathrm{RM}(3,2)$};
\node at (4,4) (x8) {\textcolor{red}{$\mathrm{RM}(2,1)$}};
\node at (6,4) (x9) {\textcolor{red}{$\mathrm{RM}(2,2)$}};
\draw [->,thick] (x1)--(x2);
\draw [->,thick] (x1)--(x3);
\draw [->,thick] (x3)--(x4);
\draw [->,thick] (x3)--(x5);
\draw [->,thick] (x5)--(x6);
\draw [->,thick] (x5)--(x7);
\draw [->,thick] (x7)--(x8);
\draw [->,thick] (x7)--(x9);

\draw [->,densely dashed,thick,red] (0.7, 6.7) to [out=-90,in=180] (x4.west);

\draw [->,densely dashed,thick,red] (1.7, 5.7) to [out=-90,in=180] (x6.west);

\draw [->,densely dashed,thick,red] (2.7, 4.7) to [out=-90,in=180] (x8.west);

\draw [->,densely dashed,thick,red] (x8) to [out=-30,in=-150] (x9);
\end{tikzpicture}

\vspace*{0.2in}
\begin{tikzpicture}
\node at (8,8) (x1) {$\mathrm{RM}(6,3)$};
\node at (4,7) (x2) {$\mathrm{RM}(5,2)$};
\node at (12,7) (x3) {$\mathrm{RM}(5,3)$};
\node at (3,6) (x4) {\textcolor{red}{$\mathrm{RM}(4,1)$}};
\node at (5,6) (x5) {$\mathrm{RM}(4,2)$};
\node at (9,6) (x6) {$\mathrm{RM}(4,2)$};
\node at (15,6) (x7) {$\mathrm{RM}(4,3)$};
\node at (4,5) (x8) {\textcolor{red}{$\mathrm{RM}(3,1)$}};
\node at (6,5) (x9) {$\mathrm{RM}(3,2)$};
\node at (8,5) (x10) {\textcolor{red}{$\mathrm{RM}(3,1)$}};
\node at (10,5) (x11) {$\mathrm{RM}(3,2)$};
\node at (14,5) (x12) {$\mathrm{RM}(3,2)$};
\node at (16,5) (x13) {\textcolor{red}{$\mathrm{RM}(3,3)$}};
\node at (5,4) (x14) {\textcolor{red}{$\mathrm{RM}(2,1)$}};
\node at (7,4) (x15) {\textcolor{red}{$\mathrm{RM}(2,2)$}};
\node at (9,4) (x16) {\textcolor{red}{$\mathrm{RM}(2,1)$}};
\node at (11,4) (x17) {\textcolor{red}{$\mathrm{RM}(2,2)$}};
\node at (13,4) (x18) {\textcolor{red}{$\mathrm{RM}(2,1)$}};
\node at (15,4) (x19) {\textcolor{red}{$\mathrm{RM}(2,2)$}};
\draw [->,thick] (x1)--(x2);
\draw [->,thick] (x1)--(x3);
\draw [->,thick] (x2)--(x4);
\draw [->,thick] (x2)--(x5);
\draw [->,thick] (x3)--(x6);
\draw [->,thick] (x3)--(x7);
\draw [->,thick] (x5)--(x8);
\draw [->,thick] (x5)--(x9);
\draw [->,thick] (x6)--(x10);
\draw [->,thick] (x6)--(x11);
\draw [->,thick] (x7)--(x12);
\draw [->,thick] (x7)--(x13);
\draw [->,thick] (x9)--(x14);
\draw [->,thick] (x9)--(x15);
\draw [->,thick] (x11)--(x16);
\draw [->,thick] (x11)--(x17);
\draw [->,thick] (x12)--(x18);
\draw [->,thick] (x12)--(x19);

\draw [->,densely dashed,thick,red] (2.7, 5.7) to [out=-90,in=180] (x8.west);

\draw [->,densely dashed,thick,red] (3.7, 4.7) to [out=-90,in=180] (x14.west);

\draw [->,densely dashed,thick,red] (x14) to [out=-30,in=-150] (x15);

\draw [->,densely dashed,thick,red] (x15) to [out=90,in=-150] (x10.west);

\draw [->,densely dashed,thick,red] (x10) to [out=-90,in=150] (x16.west);

\draw [->,densely dashed,thick,red] (x16) to [out=-30,in=-150] (x17);

\draw [->,densely dashed,thick,red] (x17) to [out=-30,in=-150] (x18);

\draw [->,densely dashed,thick,red] (x18) to [out=-30,in=-150] (x19);

\draw [->,densely dashed,thick,red] (x19.east) to [out=30,in=-90] (x13);
\end{tikzpicture}
\caption{The recursive decoding algorithm $\Phi_r^m$ for $\mathrm{RM}(m,r)$ and an illustration of how it works for $\mathrm{RM}(6,2)$ and $\mathrm{RM}(6,3)$: We decompose $\mathrm{RM}(6,2)$ and $\mathrm{RM}(6,3)$ until we reach the leaf nodes, which are the first order or full RM codes marked in red. Eventually we only need to decode these RM codes on the leaf nodes using FHT or ML decoders. The order of decoding these leaf nodes is indicated by the red dashed arrows.}
\label{fig:bb}
\end{figure}
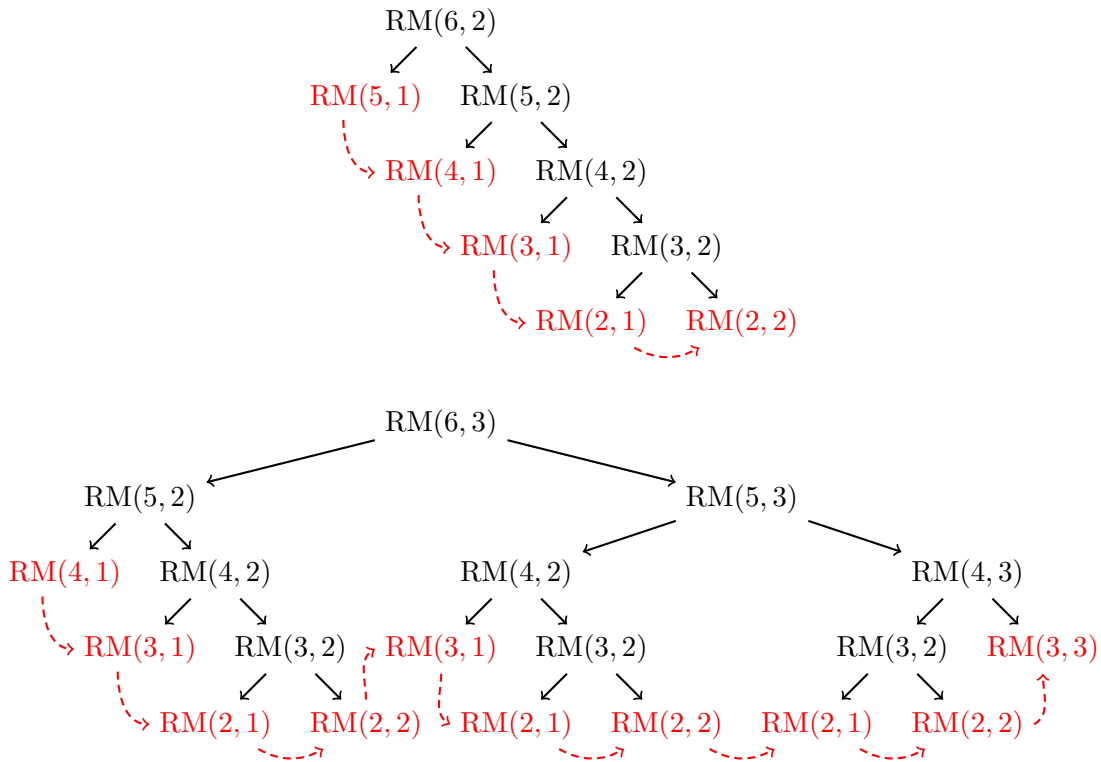

\begin{figure}
\centering
\begin{tikzpicture}
\node at (0, 9.5) [] {Decode $\mathrm{RM}(m,3)$};
\node at (-3.5, 5.25) [text width=1cm,align=center]  {\footnotesize{$N_{\max}$ \\ iterations}} ;
\node at (0.75, 5.25) [text width=1.5cm,align=center] {\footnotesize{Recursive decoding}};
\node at (0.75, 6) [] {\dots\dots};
\node at (0.75, 4.5) [] {\dots\dots};
\node at (1.6, 8) [] {\footnotesize{Projection}};
\node at (1.7, 2.5) [] {\footnotesize{Aggregation}};
\node at (0,9) [] (y) {$y$} ;
\node at (-2,6) [] (yb1) {$y^{[/\mathbb{B}_1]}$};
\node at (-0.5,6) [] (yb2) {$y^{[/\mathbb{B}_2]}$};
\node at (2,6) [] (ybn) {$y^{[/\mathbb{B}_{n-1}]}$};
\node at (-2,4.5) [] (hyb1) {$\hat{y}^{[/\mathbb{B}_1]}$};
\node at (-0.5,4.5) [] (hyb2) {$\hat{y}^{[/\mathbb{B}_2]}$};
\node at (2,4.5) [] (hybn) {$\hat{y}^{[/\mathbb{B}_{n-1}]}$};
\node at (0, 1.5) [] (hy) {$\hat{y}$};
\draw [->] (y)--(yb1);
\draw [->] (y)--(yb2);
\draw [->] (y)--(ybn);
\draw [->] (hyb1)-- (hy);
\draw [->] (hyb2)-- (hy);
\draw [->] (hybn)-- (hy);
\draw [->] (hy) -- ++(-2.7, 0) -- ++(0, 7.5) -- (y);
\draw [->] (y) -- ++(2.7, 0) -- ++(0, -7.5) -- (hy);

\draw [-] (-2, 5.75)--(-2.3, 5.25);
\draw [-] (-2.3, 5.25)--(-2, 4.75);
\draw [-] (-2, 5.75)--(-1.7, 5.25);
\draw [-] (-1.7, 5.25)--(-2, 4.75);

\draw [-] (-0.5, 5.75)--(-0.8, 5.25);
\draw [-] (-0.8, 5.25)--(-0.5, 4.75);
\draw [-] (-0.5, 5.75)--(-0.2, 5.25);
\draw [-] (-0.2, 5.25)--(-0.5, 4.75);

\draw [-] (2, 5.75)--(1.7, 5.25);
\draw [-] (1.7, 5.25)--(2, 4.75);
\draw [-] (2, 5.75)--(2.3, 5.25);
\draw [-] (2.3, 5.25)--(2, 4.75);

\node at (7, 8.5) [] {Decode $\mathrm{RM}(m-1,2)$};
\node at (7.5, 5.25) [] {\dots\dots};

\node at (7,8) [] (ys) {$y^{[/\mathbb{B}_{n-1}]}$} ;
\node at (7,2.5) [] (hys) {$\hat{y}^{[/\mathbb{B}_{n-1}]}$} ;
\draw [->] (ys)--(5.5,5.8);
\draw [->] (ys)--(6.5,5.8);
\draw [->] (ys)--(8.5,5.8);

\draw [->] (5.5, 4.7)--(hys);
\draw [->] (6.5, 4.7)--(hys);
\draw [->] (8.5, 4.7)--(hys);

\node at (5.5, 5.25) [] {FHT};
\node at (6.5, 5.25) [] {FHT};
\node at (8.5, 5.25) [] {FHT};

\draw [] (5.1, 5.05) rectangle ++ (0.8, 0.4);
\draw [] (6.1, 5.05) rectangle ++ (0.8, 0.4);
\draw [] (8.1, 5.05) rectangle ++ (0.8, 0.4);

\draw [->] (hys) -- ++(-2.1, 0) -- ++(0, 5.5) -- (ys);
\draw [->] (ys) -- ++(2.1, 0) -- ++(0, -5.5) -- (hys);

\draw [densely dashed] (2.4, 5.6) -- (4.7, 7.8);
\draw [densely dashed] (2.4, 4.9) -- (4.7, 2.7);

\draw [densely dashed] (9, 5.4) -- (9.8, 6);
\draw [densely dashed] (9, 5.1) -- (9.8, 4.5);

\node at (10.8, 5.25) [text width=3cm,align=center] 
{\footnotesize{Decode $\mathrm{RM}(m-2,1)$ with\\ Fast Hadamard Transform}};

\end{tikzpicture}
\caption{Recursive Projection-Aggregation decoding algorithm for third order RM codes} \label{fig:highlvl}
\end{figure}
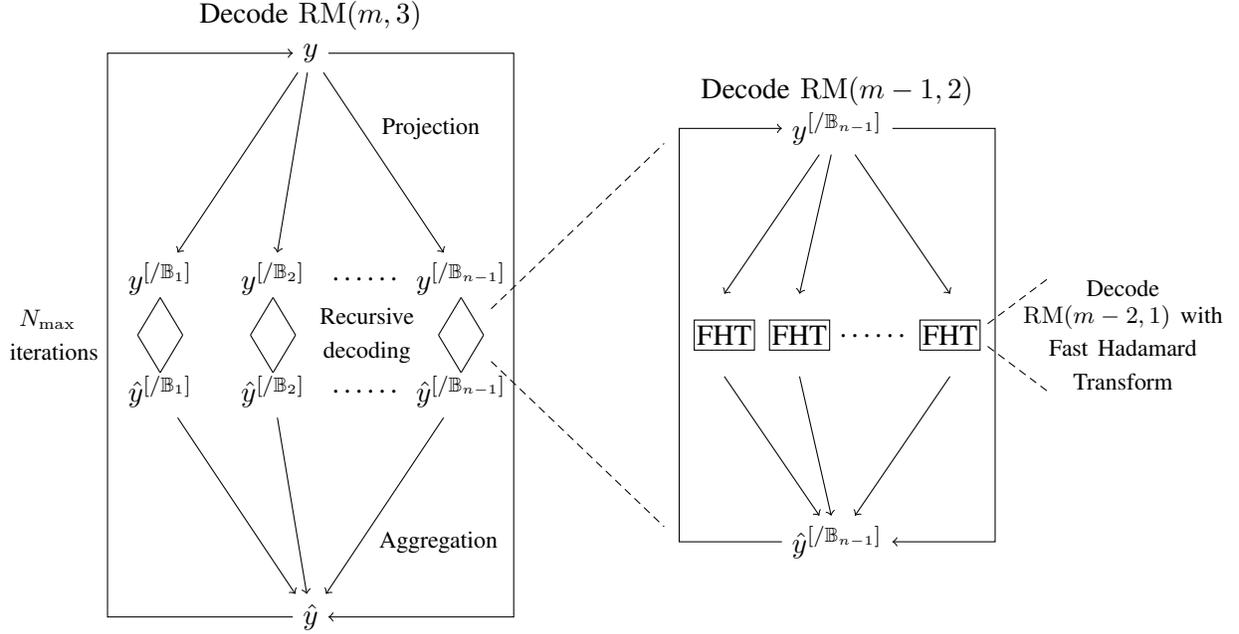

\subsubsection{Recursive projection aggregation decoding \cite{YA18}}
\label{sect:RPA}

The Recursive Projection Aggregation (RPA) decoding algorithm was proposed recently by Ye and Abbe \cite{YA18}. It works well for second and third order RM codes with short or medium code length (e.g. $\le 1024$). In particular, the RPA algorithm can efficiently achieve the same decoding error probability as the ML decoder for second order RM codes with length $\le 1024$.
Moreover, RPA decoder naturally allows parallel implementation.

We will focus mainly on the RPA decoder for BSC channels and briefly mention how to adapt it to general communication channels at the end of this section. 
In Section~\ref{sect:aip}, we have shown that
$\Eval^{[/p]}(f)\in \mathrm{RM}(m-1,r-1)$ whenever $\Eval(f)\in \mathrm{RM}(m,r)$ for any  $p=b_1z_1+\dots+b_mz_m$ with nonzero coefficients $\original{b}=(b_1,\dots,b_m)\neq \original{0}$.
Let $\bB$ be the one-dimensional subspace of $\bF_2^m$ consisting of the $\original{0}$ vector and $\original{b}$. Then $\Eval^{[/p]}(f)$ is obtained by taking the sums in each of the $2^{m-1}$ cosets of $\bB$, i.e., each coordinate in $\Eval^{[/p]}(f)$ is the sum $\sum_{\original{z}\in T} \Eval_{\original{z}}(f)$ for some coset $T$ of $\bB$.  
For this reason, we will use the two notations $\Eval^{[/\bB]}(f)$ and $\Eval^{[/p]}(f)$ interchangeably from now on. 
For a noisy codeword $y$, we also use $y^{[/\bB]}$ and $y^{[/p]}$ interchangeably, and we call $y^{[/\bB]}$ the projection of $y$ onto the cosets of $\bB$.
There are in total $n-1$ one-dimensional subspaces in $\bF_2^m$. We denote them as $\bB_1,\dots,\bB_{n-1}$. 

Suppose that we transmit a codeword $c=\Eval(f)\in\mathrm{RM}(m,r)$ through BSC, and that the channel output is $y$.
The RPA decoder for $\mathrm{RM}(m,r)$ consists of three steps: First, the projection step, then the recursive decoding step, and third, the aggregation step. More precisely, the first step is to project the noisy codeword $y$ onto all $n-1$ one-dimensional subspaces $\bB_1,\dots,\bB_{n-1}$. 
Note that this projection step also appears in Sidel'nikov-Pershakov algorithm \cite{Sidel92}
and Sakkour's algorithm \cite{Sakkour05};see Section~\ref{sect:sid}.
The second step is to decode each $y^{[/\bB_i]}$ using the RPA decoder for $\mathrm{RM}(m-1,r-1)$. If $r=2$, then we simply use the FHT decoder for $y^{[/\bB_i]}$. We denote the decoding result of the second step as $\hat{y}^{[/\bB_i]}$. Note that $\hat{y}^{[/\bB_1]},\dots,\hat{y}^{[/\bB_{n-1}]}$ consist of the (noisy) estimates of the sum $\Eval_{\original{z}}(f)+\Eval_{\original{z}'}(f)$ for all $\original{z}\neq\original{z}'$.
We denote the estimate of $\Eval_{\original{z}}(f)+\Eval_{\original{z}'}(f)$ as $\hat{y}_{(\original{z},\original{z}')}$.
Finally, in the aggregation step, observe that $\hat{y}_{(\original{z},\original{z}')}+y_{\original{z}'}$ is an estimate of $\Eval_{\original{z}}(f)$ for all $\original{z}'\neq \original{z}$. For a fixed $\original{z}$, we have in total $n-1$ such estimates of $\Eval_{\original{z}}(f)$, and we perform a majority vote among these $n-1$ estimates to obtain $\hat{y}_{\original{z}}$, i.e., we count the number of $0$'s and $1$'s in the set
$\{\hat{y}_{(\original{z},\original{z}')}+y_{\original{z}'}:\original{z}'\in\bF_2^m, \original{z}'\neq\original{z}\}$: If there are more $1$'s than $0$'s, then we set $\hat{y}_{\original{z}}$ to be 1. Otherwise, we set it to be $0$.
Next we replace the original channel output vector $y$ with $\hat{y}=(\hat{y}_{\original{z}}:\original{z}\in\bF_2^m)$, and run the Projection-Recursive decoding-Aggregation cycle again for a few more rounds\footnote{In practice, usually three rounds are enough for the algorithm to converge.}. The vector $\hat{y}=(\hat{y}_{\original{z}}:\original{z}\in\bF_2^m)$ in the last round is the final decoding result of the RPA decoder.
See Fig.~\ref{fig:highlvl} for a high-level illustration of the RPA decoder.

For general communication channels we need to work with LLR, and we only need to make two changes in the RPA decoder for BSC. The first change is in the projection step: In order to calculate $y^{[/\bB]}$, we need to calculate the sums $y_{\original{z}}+y_{\original{z}'}$ for the BSC case. We cannot do this for general communication channels because the channel output vector is not binary anymore.
Instead, we calculate the projected LLR vectors $L^{[/\bB]}$ using \eqref{eq:tsi}, and in the recursive decoding step, we decode from $L^{[/\bB_1]},\dots,L^{[/\bB_{n-1}]}$.
The second change is in the aggregation step: We replace the majority vote with a weighted sum of the LLRs. More precisely, for a fixed $\original{z}$, we calculate the sum $\hat{L}_{\original{z}}= \frac{1}{n-1}\sum_{\original{z}'\neq\original{z}} \tilde{y}_{(\original{z},\original{z}')} L_{\original{z}'}$, where we set $\tilde{y}_{(\original{z},\original{z}')}=1$ if $\hat{y}_{(\original{z},\original{z}')}=0$ and $\tilde{y}_{(\original{z},\original{z}')}=-1$ if $\hat{y}_{(\original{z},\original{z}')}=1$.
After each round, we replace the LLR vector of the original channel output with $\hat{L}=(\hat{L}_{\original{z}}:\original{z}\in\bF_2^m)$ and run the Projection-Recursive decoding-Aggregation cycle again.
After the last round, we decode $\hat{y}_{\original{z}}$ as $0$ if $\hat{L}_{\original{z}}>0$ and otherwise we decode $\hat{y}_{\original{z}}$ as $1$.
The vector $\hat{y}=(\hat{y}_{\original{z}}:\original{z}\in\bF_2^m)$ is the final decoding result of the RPA decoder.

In practical implementation, we combine the RPA decoder with the following list decoding procedure proposed by Chase \cite{Chase72} to boost the performance.
 We first sort $|L_{\original{z}}|,\original{z}\in\bF_2^m$ from small to large. Assume for example that $|L_{\original{z}_1}|,|L_{\original{z}_2}|,|L_{\original{z}_3}|$ are the three smallest components in the LLR vector, meaning that $y_{\original{z}_1},y_{\original{z}_2},y_{\original{z}_3}$ are the three most noisy symbols in the channel outputs. Next we enumerate all $8$ the possible cases of these three bits: We set $L_{\original{z}_i}=\pm L_{\max}$ for $i=1,2,3$, where $L_{\max}$ is some large real number. In practice, we can choose $L_{\max}:=2\max(|L_{\original{z}}|:\original{z}\in\bF_2^m)$.
For each of these 8 cases, we use the RPA decoder to obtain a decoded codeword (candidate).
Finally, we calculate the posterior probability for each of these $8$ candidates, and choose the largest one as the final decoding result, namely, we perform the ML decoding among the 8 candidates in the list.
This list decoding version of the RPA decoder allows us to efficiently achieve the same decoding error probability as the ML decoder for second order RM codes with length $\le 1024$.

\begin{center}
\begin{algorithm}
\caption{RPA decoder for RM codes over BSC}
{\bf Input}: The parameters of the Reed-Muller code $m$ and $r$; the received (noisy) codeword $y=(y_{\original{z}}: \original{z}\in\bF_2^m)$; the maximal number of iterations $N_{\max}$

{\bf Output}: $\hat{y}=
(\hat{y}_{\original{z}}: \original{z}\in\bF_2^m)
\in\bF_2^n$

\vspace*{0.05in}
\begin{algorithmic}[1]
\For {$j=1,2,\dots,N_{\max}$} 

\State{$\hat{y}^{/\mathbb{B}_i} \gets \texttt{RPA}(m-1,r-1,y^{/\mathbb{B}_i},N_{\max})$ for $i=1,2,\dots,n-1$}

\State{
$\{\hat{y}_{(\original{z},\original{z}')}:
\original{z},\original{z}'\in\bF_2^m,\original{z}\neq \original{z}'\}
\gets$ coordinates of $\hat{y}^{/\mathbb{B}_1},\dots,
\hat{y}^{/\mathbb{B}_{n-1}}$}

\For {every $\original{z}\in\bF_2^m$} 

\State{\texttt{num1} $\gets$ number of $\original{z}'\in\bF_2^m\setminus\{\original{z}\}$ such that $\hat{y}_{(\original{z},\original{z}')}+y_{\original{z}'}=1$}

\State{$\hat{y}_{\original{z}}\gets
\mathbf{1}[\texttt{num1}>\frac{n-1}{2}]$}

\EndFor

\If {$y=\hat{y}$}

\State{\bf break}

\EndIf

\State{$y\gets \hat{y}$}

\EndFor

\State{Output $\hat{y}$}

\end{algorithmic}
\end{algorithm}
\end{center}


\subsubsection{Additional methods \cite{Santi18,Mondelli14}}

In \cite{Santi18}, Santi et al. applied iterative decoding to a highly-redundant parity-check (PC) matrix that contains only the minimum-weight dual
codewords as rows. In particular, \cite{Santi18} proposed to use the peeling decoder
for the binary erasure channel, linear-programming and belief propagation (BP) decoding for the binary-input additive white Gaussian noise channel, and bit-flipping and BP decoding for the
binary symmetric channel. For short block lengths, it was shown
that near-ML performance can indeed be achieved in many
cases. \cite{Santi18} also proposed a method to tailor the PC matrix to the
received observation by selecting only a small fraction of useful
minimum-weight PCs before decoding begins. This allows one to
both improve performance and significantly reduce complexity
compared to using the full set of minimum-weight PCs.


In \cite{Mondelli14}, Mondelli et al. explored the relationship between polar and RM
codes, and they proposed a coding scheme which improves upon the performance of the standard polar codes at practical block lengths. The starting point is the experimental observation that RM codes have a smaller error probability than polar codes
under MAP decoding. This motivates one to introduce a family of
codes that ``interpolates" between RM and polar codes, call this
family
$\cC_{\inter}=
\{\cC_{\alpha}:\alpha\in[0,1]\}$, where
$\cC_{\alpha}|_{\alpha=1}$ is the original
polar code, and $\cC_{\alpha}|_{\alpha=0}$ is an RM code. Based on numerical
observations, one can see that the error probability under MAP
decoding is an increasing function of
$\alpha$. MAP decoding has in
general exponential complexity, but empirically the performance
of polar codes at finite block lengths is boosted by moving along the family
$\cC_{\inter}$ even under low-complexity decoding schemes such as belief propagation or successive cancellation
list decoder. Performance gain was also demonstrated in \cite{Mondelli14} via numerical simulations for transmission over the erasure channel as well as the Gaussian channel.

\subsection{Berlekamp-Welch type decoding algorithm \cite{Saptharishi17}}\label{sec:SSV}

In this section we explain the algorithm of Saptharishi, Shpilka and Volk \cite{Saptharishi17} for decoding RM codes of degrees up to $r=o(\sqrt{m})$. In fact, their algorithm also gives interesting results for degrees $r=m-o(\sqrt{m/\log m})$.
The algorithm is similar in spirit to the works of Pellikaan, Duursma and K\"{o}tter (\cite{Pellikaan92, DuursmaK94}),  which abstract the Berlekamp-Welch algorithm. Thus, it is very different from the  algorithms given in the other subsections of Section~\ref{sect:drm}.

Before stating their main theorem we will need the following notation. For $u,v \in \F_q^n$, we denote by $u *v \in \F_q^n$ the vector $(u_1 v_1, \ldots,  u_n  v_n)$. For $A,B \subseteq \F_q^n$ we similarly define $A*B = \{u * v \mid {u \in A, v \in B}\}$.

The algorithm considers three codes $C,E$ and $N$, all subsets of $\F_q^n$, such that $E*C\subseteq N$, and is able to correct in $C$ error patterns that are correctable from erasures in $N$, through the use of an {\em error-locating code} $E$.

\begin{algorithm}
  \caption{Decoding Algorithm of \cite{Saptharishi17}}
  \label{alg:abstract}
\begin{algorithmic}[1]
  \Require{received word $\boldy \in \F_q^n$ such that $\boldy = \boldc + \bolde$, with $\boldc \in C$ and $\bolde$ is supported on a set $U$}
  \State{Solve for $\bolda \in E, \boldb \in N$, the linear system $\bolda * \boldy = \boldb$.}
  \State{Let $\left\{\bolda_1, \ldots, \bolda_k \right\}$ be a basis for the solution space of $\bolda$, and let $\mathcal{E}$ denote the common zeros of $\{\bolda_i \mid i \in [k]\}$.}
  \State{For every $j \in \mathcal{E}$, replace $\boldy_j$ with `?', to get a new word $\boldy'$.}
  \State{Correct $\boldy'$ from erasures in $C$.}
\end{algorithmic}
\end{algorithm}

\begin{theorem}
\label{thm:abstraction}
Let $\F_q$ be a finite field and $E, C, N \subseteq \F_q^n$ be codes with the following properties.
\begin{enumerate}
\item \label{item:inclusion} $E*C \subseteq N$
\item \label{item:correction} For any pattern $\1_{U}$ that is
  correctable from erasures in $N$, and for any coordinate $i \not \in U$ there exists a codeword $\bolda \in E$ such that $\bolda_j = 0$ for all $j \in U$ and $\bolda_i=1$.
\end{enumerate}
Then Algorithm~\ref{alg:abstract}  corrects in $C$ any error pattern $\1_U$ which is correctable from erasures in $N$.
\end{theorem}

It is worth pointing out the differences between Algorithm~\ref{alg:abstract} and the abstract Berlekamp-Welch decoder of Pellikaan, Duursma and K\"{o}tter \cite{Pellikaan92, DuursmaK94}.  They similarly set up codes $E, C$ and $N$ such that $E*C \subseteq N$.
However, instead of Property~\ref{item:correction}, they require that for any $\bolde \in E$ and $\boldc \in C$, if $\bolde * \boldc = 0$ then $\bolde=0$ or $\boldc=0$ (or similar requirements regarding the distances of $E$ and $C$ that guarantee this property).
This property, as well as the distance properties, do not hold in the case of Reed-Muller codes, which is the main application of \Cref{thm:abstraction}.

\begin{proof}[Proof Sketch]
It is relatively easy to show that Property~\ref{item:correction} in the statement of the theorem guarantees that every erasure pattern that is correctable in $N$ is also correctable in $C$. The main point of the algorithm is that, under the hypothesis of the theorem, the common zeros of the possible solutions for $\bolda$ determine exactly the error locations. 

Denote $\boldy = \boldc + \bolde$, where $\boldc \in C$ is the transmitted codeword and $\bolde$ is supported on the set of error locations $U$.
The following two claims guarantee that the algorithm correctly finds the set of error locations. The first claim shows that error locations are common zeros and the second claim shows that no other coordinate is a common zero. 

  \begin{claim}\label{cla:each-zero}
    For every $\bolda \in E, \boldb \in N$ such that $\bolda * \boldy = \boldb$, it holds that $\bolda * \bolde = 0$.
  \end{claim}
  \begin{proof}
Observe that $\bolda*\bolde = \bolda * \boldy - \bolda * \boldc$ is also a codeword in $N$.
As $\bolda*\bolde$ is  supported on $U$, and since $U$ is an erasure-correctable pattern in $N$, it must be the zero codeword.
  \end{proof}

  \begin{claim}\label{cla:no-other-zero}
    For every $i \not \in U$ there exists $\bolda \in E, \boldb \in N$ such that $\bolda$ is 0 on $U$, $\bolda_i=1$ and $\bolda*\boldy=\boldb$.
  \end{claim}
  \begin{proof}
   Property~\ref{item:correction} implies that since $U$ is correctable from erasures in $N$, for every $i \not \in U$ we can pick $\bolda \in E$ such that $\bolda$ is 0 on $U$ and $\bolda_i = 1$. Set $\boldb = \bolda * \boldy$. As
   $\boldb = \bolda * \boldc + \bolda * \bolde = \bolda * \boldc$ it follows that $b\in N$.
  \end{proof}
  Together, Claims~\ref{cla:each-zero} and~\ref{cla:no-other-zero} imply the correctness of the algorithm.
\end{proof}

To apply \Cref{thm:abstraction} to RM codes we note that for $m \in \N$ and $r \le m/2 -1$ the codes $C=\reedmuller{m-2r-2}{m}$, $N=\reedmuller{m-r-1}{m}$ and $E=\reedmuller{r+1}{m}$ satisfy the conditions of Theorem~\ref{thm:abstraction}.

Theorem~\ref{thm:bec-cap} shows that $\reedmuller{m-r-1}{m}$ achieves capacity for $r=m/2 \pm O(\sqrt{m})$.
Letting $r=m/2-o(\sqrt{m})$ and looking at the code $\reedmuller{m-2r-2}{m} = \reedmuller{o(\sqrt{m})}{m}$ so that ${m \choose \le r} = (1/2 - o(1))2^m$, Saptharishi et al. obtained the following corollary to Theorem~\ref{thm:abstraction}.

\begin{corollary}
  There exists an efficient (deterministic) algorithm that is able to correct a fraction of $(1/2 - o(1))$ random errors in $\reedmuller{o(\sqrt{m})}{m}$, with probability $1-o(1)$.
\end{corollary}

Similar arguments allow \cite{Saptharishi17} to obtain results for high rate RM codes. 

\begin{theorem}[Theorem 2 of \cite{Saptharishi17}]\label{thm:SSV-high}
Let $r =o(\sqrt{m/\log m})$.
Then, there is an efficient algorithm that can correct $\reedmuller{m-(2r+2)}{m}$ from a random set of $(1-o(1))\binom{m}{\le r}$ errors.
Moreover, the running time of the algorithm is $2^m \cdot \poly(\binom{m}{\le r})$.
\end{theorem}

In \cite{KoppartyP18}, Kopparty and Potukuchi improved the running time of the decoding algorithm of Theorem~\ref{thm:SSV-high} to run in time polynomial in the length of the syndrome, that is in time $\poly(\binom{m}{\le r})$.

\section{Open problems}  \label{sect:open}

\subsection{RM codes and Twin-RM codes}
As mentioned in Section \ref{sect:RMpolar}, the twin-RM code, obtained by retaining the low-entropy components of the squared RM code (proceeding in the RM ordering) is proved to achieve capacity on any BMS. There are now three possible outcomes: (i) the twin-RM code is exactly the RM code, (ii) the twin-RM code is equivalent to the RM code, in that only a vanishing fraction of rows selected by the two codes are different, (iii) the twin-RM code is not equivalent to the RM code. If (i) or (ii) are true, the RM code will also be capacity achieving for any BMS. If (iii) is true instead, then the RM code is not capacity-achieving (for the considered BMS channel). 
\begin{conjecture}
The twin-RM code is the RM code, i.e., with the notation of Section \ref{sect:RMpolar} where $U^n=R_n X^n$ and $R_n$ is the squared RM code matrix, for $i,j \in [n]$,
\begin{align}
|A_i| > |A_j| \implies
H(U_i|U^{i-1},Y^n) \ge H(U_j|U^{j-1},Y^n), 
\end{align}
in words, the conditional entropy is non-decreasing as we go to lower degree layers.
\end{conjecture}

Using the symmetry of RM codes, \cite{AY18} shows that this is implied by the following conjecture. 
\begin{conjecture}
Let $i$ be the index of the last row in a layer of $R_n$ (besides the last layer), then
\begin{align}
H(U_i|U^{i-1},Y^n) \ge H(U_{i+1}|U^{i},Y^n), 
\end{align}
in words, the conditional entropy is non-decreasing as we cross layers. 
\end{conjecture}
Finally, since the RM code construction is the same for the BEC, the BSC or any other channel, and since it is already proved to achieve capacity on the BEC \cite{Kudekar16STOC}, it follows that the RM and twin-RM codes must be equivalent on the BEC. Consequently, we have the following conjecture that  also implies that RM codes achieve capacity on the BSC.
\begin{conjecture}The Twin-RM codes for the BEC and BSC are equivalent. I.e., for some $\epsilon =o(1/n)$, $|\{i \in [n]: H(U_{i}|U^{i-1},Y_{\mathrm{BSC}}^n) \le \epsilon , H(U_{i}|U^{i-1},Y_{\mathrm{BEC}}^n) \ge  \epsilon\}|=o(n)$.
\end{conjecture}
One could also exploit the ordering between BEC and BSC entropies, showing that for  any two indices $i,j$ of rows with $j$ in a lower degree layer than $i$,  $H(U_i|U^{i-1},Y_{\mathrm{BEC}}^n) \ge H(U_{j}|U^{j-1},Y_{\mathrm{BEC}}^n)$ implies  $H(U_i|U^{i-1},Y_{\mathrm{BSC}}^n) \ge H(U_{j}|U^{j-1},Y_{\mathrm{BSC}}^n)$. 
Finally, one may also look for a BMS channel such that the Twin-RM code is not equivalent to the RM code, which would imply that the RM code does not achieve capacity for this BMS.

\subsection{Weight enumerator}
\noindent {\it 1. Unified bounds for two different regimes}

As mentioned in Section~\ref{sec:wt-dist}, we now have two different approaches to bound the number of codewords in two different regimes. More precisely, the method proposed in \cite{Samorod18} gives strong and in some cases nearly optimal bounds in the constant relative weight regime for constant rate codes. Yet this method does not produce meaningful bound when the code rate approaches $0$ or in the small weight regime. On the other hand, the method in \cite{Sberlo18} gives good bound for small rate codes or in small weight regime, but it does not work well in the linear weight regime for constant rate codes. A natural open problem is thus to obtain a unified bound that is effective in both regimes.

\vspace*{0.1in}
\noindent {\it 2. Use the capacity-achieving results for BEC to prove the conjecture for BSC}

As discussed in Section~\ref{sec:wt-dist},
Samorodnitsky gave a nearly optimal upper bound in a certain linear weight regime for any code that achieves capacity for the BEC. This in particular means that the weight distribution of RM codes in this regime is (nearly) the same as that of random codes. Since random codes achieve capacity of BSC, can we thus extend this result to prove that RM codes achieve capacity on the BSC?

\vspace*{0.1in}
\noindent {\it 3. Close the gap between the upper and lower bounds for small weight codewords}

There is a gap between the existing upper and lower bounds on small weight codewords; see \Cref{thm:sberlo-bias} and \Cref{thm:biaslb}. A natural open problem is to close this gap.

\subsection{Algorithms}

Another open problem is to develop an efficient (polynomial time) decoder for the twin-RM codes. Similarly to polar codes, twin-RM codes are also suitable for successive decoding by construction.
Moreover, both twin-RM codes and polar codes have similar but different recursive structures that can be used to reduce the complexity of the successive decoder. The difference is that the recursive structure of polar codes allows for an FFT-like implementation of the successive decoder which only has $O(N\log N)$ complexity \cite{Arikan09}.
On the other hand, the naive utilization of the recursive structure of twin-RM codes  
does not give such an obvious reduction. 
The hope is to still find a more efficient  (polynomial-time) implementation of the successive decoder for twin-RM codes based on a divide-and-conquer method.\\

The results of \cite{Sberlo18}, as described in Section~\ref{sec:wt-to-cap}, show that $\reedmuller{m/2 - O(\sqrt{m\log m})}{m}$ can correct a fraction of $1/2-o(1)$ random errors. Currently, the best algorithm that we have can only handle degrees up to $o(\sqrt{m})$ \cite{Saptharishi17} (see Section~\ref{sec:SSV}). It is a natural  open problem to extend these up to any constant fraction of errors. 


\bibliographystyle{IEEEtran}
\bibliography{RM}

\end{document}